\documentclass{imanum}

\usepackage{amsmath,amsthm} 
\usepackage{amssymb,mathrsfs} 
\usepackage{a4wide} 
\usepackage{graphicx} 
\usepackage{color,subfigure} 
\usepackage{enumerate}  

\DeclareMathAlphabet{\mathpzc}{OT1}{pzc}{m}{it}
 
\newcommand{\rme}{\mathrm{e}} 
\newcommand{\etA}{\mathrm{e}^{t A}} 
\newcommand{\etB}{\mathrm{e}^{t B}} 
\newcommand{\etC}{\mathrm{e}^{t C}} 
\newcommand{\esA}{\mathrm{e}^{s A}} 
\newcommand{\esB}{\mathrm{e}^{s B}} 
\newcommand{\esC}{\mathrm{e}^{s C}}

\newcommand{\Id}{\mathrm{Id}} 
 
\newcommand{\dt}{{\Delta t}}

\newcommand{\RR}{\mathbb{R}}

\newcommand{\cE}{\mathcal{E}}
\newcommand{\cH}{\mathcal{H}}
\newcommand{\cL}{\mathcal{L}}
\newcommand{\wcL}{\widetilde{\mathcal{L}}}

\newcommand{\Li}{\mathcal{K}}
\newcommand{\Lovd}{\mathcal{L}_{\rm ovd}}
\newcommand{\Lgam}{\mathcal{L}_\gamma}

\newcommand{\Tr}{\mathrm{Tr}}
\newcommand{\cT}{\mathcal{T}}
\newcommand{\tempT}{\mathrm{T}}
\newcommand{\kappaK}{c}
\renewcommand{\leq}{\leqslant}
\renewcommand{\geq}{\geqslant}

\renewcommand{\d}[1]{\mathrm{d}#1}

\begin{document}

\title{The computation of averages from equilibrium and nonequilibrium Langevin molecular dynamics}
\shorttitle{The computation of averages in Langevin dynamics}

\author{%
{\sc
Benedict Leimkuhler\thanks{Email: b.leimkuhler@ed.ac.uk},
Charles Matthews\thanks{Email: c.matthews@ed.ac.uk}  } \\[2pt]
University of Edinburgh, School of Mathematics,   \\
James Clerk Maxwell Building,  Edinburgh, EH9 3JZ, UK\\[6pt]
{\sc and}\\[6pt]
{\sc Gabriel Stoltz}\thanks{Corresponding Author. Email: stoltz@cermics.enpc.fr}\\[2pt]
Universit\'e Paris-Est, CERMICS (ENPC),\\
INRIA, F-77455 Marne-la-Vall\'ee, FRANCE
}
\shortauthorlist{B. Leimkuhler, C. Matthews and G. Stoltz}

\maketitle

\begin{abstract}
{We consider numerical methods for thermodynamic sampling, \textit{i.e.} computing sequences of points distributed according to the Gibbs-Boltzmann distribution, using Langevin dynamics and overdamped Langevin dynamics (Brownian dynamics).  A wide variety of numerical methods for Langevin dynamics may be constructed based on splitting the stochastic differential equations into various component parts, each of which may be propagated exactly in the sense of distributions. Each such method may be viewed as generating samples according to an associated invariant measure that differs from the exact canonical invariant measure by a stepsize-dependent perturbation.  We provide error estimates \`a la Talay-Tubaro on the invariant distribution for small stepsize, and compare the sampling bias obtained for various choices of splitting method.  We further investigate the overdamped limit and apply the methods in the context of driven systems where the goal is sampling with respect to a nonequilibrium steady state. Our analyses are illustrated by numerical experiments.}
{Langevin dynamics; Stochastic differential equations; Numerical discretization; Canonical sampling; Molecular dynamics; Talay-Tubaro expansion; Nonequilibium.}
\end{abstract}

\section{Introduction}
\label{sec:introduction}

A fundamental purpose of molecular simulation is the computation of macroscopic quantities, typically through averages of functions of the variables of the system with respect to a given probability measure~$\mu$ which defines the macroscopic state of the system. We consider systems described by a separable Hamiltonian
\begin{equation} \label{ham}
H(q,p) = V(q) + \frac12 p^T M^{-1}p,
\end{equation}
where $q = (q_1,\dots,q_N)$ and $p=(p_1,\dots,p_N)$ respectively are the vectors of positions and momenta of~$N$ particles in dimension~$d$, $V$ is a potential energy function and $M$ is a positive definite mass matrix, typically a diagonal matrix. 

The Hamiltonian (\ref{ham}) represents a fully classical molecular dynamics model. For instance, a fluid of $N$ argon atoms is well described by pairwise interactions  among the nuclei, where the potential $V(q) = \sum_{1 \leq i < j \leq N} v(|q_i-q_j|).$   The distance based potential $v(r)$ may be fitted to Buckingham or Lennard-Jones forms (for instance, see \cite{frenksmit} or \cite{allentild}).  These short-ranged potentials model van der Waals type interactions including both Pauli repulsion (the inability of the populated electron shells to interpenetrate) and dispersion due to temporary dipoles forming in the charge clouds surrounding the nuclei. In more complicated molecular systems, other potential energy functions are used to capture local covalent bond structure and Coulombic interactions due to charges on the atoms.  Coarse-grained classical models may amalgamate several degrees of freedom, as for example when a molecule is replaced by a rigid body description.   Classical molecular dynamics models are now a standard and widespread tool in almost every field of science and engineering.  For example, see \cite{ne} for some applications in engineering, \cite{dd} for a discussion of the use of molecular dynamics in drug discovery and see also the motivation provided in classical textbooks on molecular simulation such as \cite{allentild,frenksmit,Schlick,Tuckerman}.

In the most common setting, the probability measure~$\mu$ with respect to which averages are computed corresponds to the canonical ensemble. Its distribution is defined by the Boltzmann-Gibbs density, which models the configurations of a conservative system in contact with a heat bath at fixed temperature~$\tempT$:
\begin{equation}
  \label{eq:canonical_measure}
  \mu(\d{q} \,\d{p}) = Z^{-1} \rme^{-\beta H(q,p)} \, \d{q} \,\d{p},
\end{equation}
where $\beta^{-1} = k_B\tempT$ with $k_B$ Boltzmann's constant and $Z$ is a normalization constant ensuring that the integral of $\mu$ over the entirety of phase space is unity.

Molecular dynamics can be used for the study of a wide range of thermodynamic and structural properties. Typically, observables are chosen which capture the features of interest and numerical studies are aimed at computing the averages of these observables accurately.  For instance, the average pressure in a three-dimensional fluid such as liquid argon is obtained by computing $\mathcal{P} = \mathbb{E}_\mu(\psi)$, the expectation of an observable $\psi$ with respect to the canonical measure $\mu$, where the pressure observable $\psi$ is defined as
\[
\psi(q,p) = \frac{1}{3 \mathcal{V}} \left( p^T M^{-1} p - \sum_{i=1}^N q_i \cdot \nabla_{q_i} V(q)
\right),
\]
$\mathcal{V}$ being the physical volume of the box occupied by the fluid. By studying the variation in pressure with changes in a thermodynamic parameter (temperature or density), one may obtain part of the phase diagram of the material.   Other observables may be used to model the determination of molecular form (shape and size) or structural rearrangement under different ambient conditions. It is for instance increasingly common to use molecular dynamics in biology to reveal  allosteric mechanisms related to protein function or drug binding; in such cases the observable may measure the  distance between two particular groups of atoms or their relative alignment; see~\cite{dd} for examples and further references contained therein.

\medskip

Numerically, the high-dimensional averages with respect to~$\mu$ are often approximated as ergodic averages along discrete stochastic paths (Markov chains) constructed through numerical solution of certain stochastic differential equations (SDEs). There are two principal sources of approximation error in the computation of average properties such as $\mathbb{E}_\mu(\psi)$: (i) systematic bias (or {\em perfect sampling bias}) related to the use of a discretization method for the SDEs (and usually proportional to a power of the integration stepsize $\dt$), and  (ii) statistical errors, due to the finite lengths of the sampling paths involved and the underlying variance of the random variables; see the presentation in Section~2.3.1 of~\cite{LRS10}. In this article we are concerned with the systematic bias, specifically the systematic bias in long-term simulation, i.e. with respect to the invariant (or nonequilibrium steady-state) distribution.   

One of the most popular choices of SDE system for sampling purposes is Langevin dynamics, which is given by:
\begin{equation}
\label{langevin}
\left\{ \begin{aligned}
\d{q}_t & = M^{-1} p_t \, \d{t}, \\
\d{p}_t & = -\nabla V(q_t) \, \d{t} - \gamma M^{-1} p_t \, \d{t} + \sqrt{\frac{2\gamma}{\beta}} \, \d{W}_t, 
\end{aligned} \right.
\end{equation}
where $\d{W}_t$ is a standard $dN$-dimensional Wiener process. The friction intensity $\gamma>0$ is a free parameter which may be adjusted to enhance sampling efficiency. Under suitable conditions, the dynamics~\eqref{langevin} is ergodic for the Boltzmann-Gibbs distribution (see for instance~\cite{Talay02,MSH01,CLS07} and references therein).  

We will also be interested in nonequilibrium situations where a given system is subject to nonconservative driving and dissipative perturbations. In this case, the averages may be taken with respect to a stationary distribution which has no simple functional form. The simulation of nonequilibrium systems in their steady-states is one popular way to compute transport coefficients such as the thermal conductivity or the shear viscosity, as the linear response of an appropriate average property (see for instance~\cite{EM08,Tuckerman}). We discuss a specific example in Section~\ref{sec:noneq_systems}: the computation of the mobility of a particle, which measures the tendency of the particle to flow in the direction of an external forcing. The mobility is related to the self-diffusion through Einstein's relation (see~\eqref{eq:def_nu_Einstein} below).

The aim of this work is to provide a numerical analysis of the perfect sampling bias in Langevin dynamics arising from numerical schemes obtained by a splitting strategy, building on studies such as~\cite{Talay02} or~\cite{BO10}, and clarifying the sampling  properties of recently proposed schemes (see~\cite{SkeelIzaguirre,Melchionna,BussiParinello,ParisDudesCharlieKnowsTheReference,LM12}). 
Of particular interest is the behavior of methods in the overdamped limit $\gamma\to+\infty$ and variations of Langevin dynamics incorporating nonequilibrium forcings such as the addition of non-gradient forces (in which case the invariant measure is unknown). The idea behind splitting schemes for stochastic differential equations is to decompose the generator of the dynamics into a sum of generators associated with dynamics which are analytically integrable, or at least very simple to integrate. We refer to the individual splitting  terms of the dynamics as ``elementary dynamics'' in the sequel. One example in the context of Langevin dynamics is the splitting scheme based on a symplectic integration of the Hamiltonian part of the dynamics combined with an exact treatment of the fluctuation-dissipation part. Such methods are more convenient to implement in molecular simulation codes than the implicit schemes proposed in~\cite{Talay02} or~\cite{MSH01}, and are also efficient in practice (see~\cite{LM13}). Some essential elements of the numerical analysis on the accuracy of such splitting schemes have been provided in~\cite{BO10}. 

\medskip

We focus in this article on the case where the position space is compact (e.g. a torus) since this is most relevant from the point-of-view of applications in condensed matter physics and biology, where periodic boundary conditions are typically used. This assumption simplifies the treatment of the Fokker-Planck operator associated to Langevin dynamics, and, with additional smoothness assumptions on the potential energy function, ensures regularity properties, discrete spectrum and spectral gap. In particular \eqref{eq:canonical_measure} is the unique invariant probability measure of the Langevin process. We assume for simplicity that the positions belong to the torus $\mathcal{M}=(L\mathbb{T})^{dN}$ where $L>0$ denotes the size of the simulation cell, and denote by $\cE = \mathcal{M} \times \mathbb{R}^{dN}$ the state space of the system, \textit{i.e.} the set of all admissible configurations~$(q,p)$. 

Let us emphasize that we expect our results to hold for unbounded position spaces, under appropriate assumptions on the potential energy function. Our proofs may however require non-trivial modifications, using in particular the tools and the results from~\cite{MSH01,Talay02,BO10,Kopec}. Generalizations to other dynamics similar to Langevin dynamics such as Generalized Langevin Dynamics (see~\cite{Mori,Z73}), Dissipative Particle Dynamics (see~\cite{HK92,EW95}) or Nos\'e-Hoover-Langevin dynamics (see~\cite{SaDeCh07,LeNoTh09}) are also possible, although a rigorous extension would require substantial work in view of the estimates needed involving the generator of the dynamics for instance (see the discussion in Remark~\ref{rmk:structure_proof}). 

In practice,  since Langevin dynamics is discretized, averages computed along a single trajectory converge to averages with respect to a measure~$\mu_{\gamma,\dt}$, which is an approximation to~$\mu$ in the sense that there exists a function $f_{\alpha,\gamma}$ for which
\begin{equation}
\label{eq:error_estimate_intro}
\int_\cE \psi(q,p) \, \mu_{\gamma,\dt}(\d{q} \, \d{p}) = \int_\cE \psi(q,p) \, \mu(\d{q} \, \d{p}) + \dt^\alpha \int_\cE \psi(q,p) f_{\alpha,\gamma}(q,p) \, \mu(\d{q} \, \d{p}) + \mathrm{O}(\dt^{\alpha+1}),
\end{equation}
see Section~\ref{sec:error_estimates_finite_friction} for precise statements. Of course, the momenta are usually trivial to sample since they are distributed according to a Gaussian measure. The primary issue is therefore to sample positions according to the marginal of the canonical measure:
\begin{equation}
\label{eq:marginal_mu}
\overline{\mu}(\d{q}) = \widetilde{Z}^{-1} \rme^{-\beta V(q)} \, \d{q}.
\end{equation}
Denoting by $\overline{\mu}_{\gamma,\dt}(\d q)$ the marginal of the invariant measure for the numerical scheme in the position variables, and by 
\begin{equation}
\label{eq:pi}
(\pi \varphi)(q) = \int_{\RR^{dN}} \varphi(q,p) \, \kappa(\d{p}), \qquad \kappa(\d{p}) = \left(\frac{2\pi}{\beta}\right)^{-dN/2} \sqrt{\mathrm{det}(M)} \exp\left(-\frac{\beta p^T M^{-1} p}{2}\right) \d{p},
\end{equation}
the partial average of a function $\varphi$ with respect to the momentum variable, the error estimate~\eqref{eq:error_estimate_intro} becomes, for observables which depend only on the position variable,
\[
\int_\mathcal{M} \psi(q) \, \overline{\mu}_{\gamma,\dt}(\d{q}) = \int_\mathcal{M} \psi(q) \, \overline{\mu}(\d{q}) + \dt^\alpha \int_\mathcal{M} \psi(q) (\pi f_{\alpha,\gamma})(q) \, \overline{\mu}(\d{q}) + \mathrm{O}(\dt^{\alpha+1}).
\]

Let us conclude this introduction by noting that alternative sampling strategies are available: the bias in the invariant measure sampled by discretization of Langevin dynamics could in principle be eliminated by employing a Metropolis-Hastings procedure (see~\cite{MRRTT53,Hastings70} and the discussion in Section~2.2 of~\cite{LRS10}). Another advantage of superimposing a Metropolis-Hastings procedure upon a discretization of Langevin dynamics is that it stabilizes the numerical scheme even for  forces~$-\nabla V$ which are not globally Lipschitz. The numerical analysis of Langevin-based Metropolis integrators has been performed in~\cite{BV09} and~\cite{BV12}, where strong error estimates are provided. On the other hand, it is not always possible or desirable to use a Metropolis correction. First, the average acceptance probability in the Metropolis step for Langevin-like dynamics in general decreases exponentially with the dimension of the system for a \emph{fixed} timestep (see for instance~\cite{KP91}). In fact, the timestep should be reduced as some inverse power of the system size in order to maintain a constant acceptance rate (see the recent works on Metropolization of Hamiltonian dynamics by \cite{BPRSS13}, following the strategy pioneered in \cite{RGG97,RR98}). There are ways to limit the decrease of the ratio, by either changing the dynamics or the measure used to compute the Metropolis ratio (see for instance~\cite{IH04} in the context of Hamiltonian dynamics), or by evolving only parts of the system (see~\cite{BV12}). The latter strategy may however complicate the implementation of parallel algorithms for the simulation of very large systems, especially if long-range potentials are used (as acknowledged in Remark~2.5 of~\cite{BV12}). This may be a reason why Metropolis corrections are not often implemented in popular molecular dynamics packages such as NAMD. Second, the variance of the computed averages may increase since rejections occur, and the numerical trajectory is therefore more correlated in general than for rejection-free dynamics. Lastly, the Metropolis procedure requires that the invariant measure of the system be known. This is the case for equilibrium systems, but no longer is the case for nonequilibrium systems subjected to external forcings such as a temperature gradient or a non-gradient force (this is the framework considered in Section~\ref{sec:noneq_systems} of this article, see for instance the dynamics~\eqref{eq:noneq_Langevin}). 

\subsection*{Summary of the results and organization of the paper}

We focus in this article on first- and second-order splitting schemes, relying on Lie-Trotter decompositions of the evolution. This restriction is motivated both by pedagogical purposes and by the dominant role in applications played by second-order splitting schemes. Let us however emphasize that most of our results could, in principle, be extended to higher-order decompositions.

Results corresponding to discretizations of the equilibrium Langevin dynamics and computation of static average properties are gathered in Section~\ref{sec:equilibrium}, while nonequilibrium systems and the computation of transport properties are discussed in Section~\ref{sec:noneq_systems} (relying on the computation of the mobility or autodiffusion coefficient as an illustration). The proofs of all our results can be found in Section~\ref{sec:proofs}.

\bigskip

Let us now highlight some of our contributions.
\begin{itemize}
\item In the equilibrium setting, we rigorously ground in Section~\ref{sec:error_estimates_finite_friction} the results presented in~\cite{LM12} giving the leading order correction to the invariant measure with respect to~$\dt$ for general splitting schemes, via a Talay-Tubaro expansion (see~\cite{TT90}). We carefully study all possible splitting schemes, taking advantage of what we call the ``TU lemma'' (Lemma~\ref{lem:TU}) to relate invariant measures of various splitting schemes where the elementary dynamics are integrated in different orders. From a technical viewpoint, our proofs are a variation on the standard way of establishing similar results since we use the specific structure of splitting schemes to conveniently write evolution operators as compositions of the semigroups of the elementary dynamics (working at the level of generators, as in~\cite{DF12}; see also~\cite{MST10} for a related approach based on solution of appropriate Poisson equations). The structure of the proof is highlighted in Section~\ref{sec:proof_thm:error_first_order_schemes}, see Remark~\ref{rmk:structure_proof}. 

\item We show in Section~\ref{sec:num_estimation_correction} how the leading order correction to equilibrium averages can be estimated on-the-fly by approximating a time-integrated correlation function. This can be seen as a practical way of numerically solving a Poisson equation (a standard way of proceeding when studying linear response of nonequilibrium systems) and is an alternative to Romberg extrapolation to eliminate the leading order correction as done in~\cite{TT90}.

\item We carefully study the overdamped regime $\gamma \to +\infty$ in Section~\ref{sec:ovd_limit}, making use in particular of uniform resolvent estimates obtained in Theorem~\ref{lem:bounds_CL_gamma} thanks to a uniform hypocoercivity property;

\item We provide error estimates for the computation of transport coefficients, by assessing the bias arising in the numerical discretization of either (i) the computation of integrated time-correlation functions expressing transport coefficients via Green-Kubo formulae; or (ii) ergodic averages of steady-state nonequilibrium dynamics where the equilibrium evolution~\eqref{langevin} is perturbed by a non-gradient force and the transport coefficient is extracted from the linear response of some quantity of interest (see Section~\ref{sec:noneq_systems}). The latter approach is illustrated by the study of the mobility, which measures the response in the average velocity arising from a constant external force exerted on the system. We also study the consistency of the numerical estimations in the overdamped limit.
\end{itemize}
Some numerical simulations are provided to illustrate the most important results (see Section~\ref{sec:numerics} and~\ref{sec:error_transport}).

\section{Error estimates for the invariant measure for equilibrium dynamics}
\label{sec:equilibrium}

We start by giving some properties of Langevin dynamics in Section~\ref{sec:ppties_Lang} (most results are well-known, except for the material on the overdamped limit $\gamma \to +\infty$ presented in Section~\ref{sec:ovd_lim_Lang}). The numerical schemes we consider are then described in Section~\ref{sec:splitting_schemes}, their ergodic properties being discussed in Section~\ref{sec:ergo_num_scheme}. Error estimates for the invariant measure are provided in Section~\ref{sec:error_estimates_finite_friction}. We then show in Section~\ref{sec:num_estimation_correction} how to estimate the leading order correction term through an appropriate integrated correlation function. An important side result of this section is the development error estimates for Green-Kubo type formulas. Finally, we study the errors on the invariant measures in the overdamped limit in Section~\ref{sec:ovd_limit}. Let us emphasize that we will make use of the following assumption throughout this work:\\

\noindent {\sc Assumption} 1:  The potential $V$ belongs to $C^\infty(\mathcal{M},\mathbb{R})$.\\


The above assumption is quite restrictive since typical potentials used in molecular simulation, such as the Lennard-Jones potential, have singularities. Although ergodicity for Langevin dynamics with singular potentials has been recently proved in~\cite{CG10}, there are still many issues with singular potentials, including the existence and uniqueness of an invariant measure for numerical schemes (see~\cite{MSH01}), and the derivation of appropriate bounds or estimates on the resolvent of the generator of Langevin dynamics (all the results presented in Section~\ref{sec:ergodicity} below are obtained under the assumption of smooth potentials). Since the latter estimates are fundamental for our work, we have to restrict ourselves to smooth potentials.
Of course, from a more practical viewpoint, it could also be argued that the potential energy function could be smoothed out by appropriate high energy truncations and regularizations, and that such regularizations should not affect too much the average properties of the system since high energy states are quite unlikely under the canonical measure.

\subsection*{Functional analysis setting and notation}

The reference Hilbert space for our analysis is the Hilbert space $L^2(\mu)$. As in~\cite{Talay02} for instance, we will consider errors in the average of smooth functions whose derivatives grow at most polynomially (the space $\mathcal{S}$ defined below). In fact, since the position space is compact, only the growth in the momentum variable has to be controlled. 

The polynomial growth of a function can be characterized by the Lyapunov functions:
\[
\Li_s(q,p) = 1 + |p|^{2s},
\]
for $s \in \mathbb{N}^* = \{1,2,3,\ldots \}$. This allows us to define the following Banach spaces of functions of polynomial growth
\[
L^\infty_{\Li_s} = \left \{ \psi \textrm{ measurable } \, \left| \, \frac{\psi}{\Li_s} \in L^\infty(\cE) \right. \right\},
\]
endowed with the norms
\[
\| \psi \|_{L^\infty_{\Li_s}} = \left\| \frac{\psi}{\Li_s}\right\|_{L^\infty}.
\]
To characterize the growth of the derivatives, we introduce the spaces $W^{m,\infty}_{\Li_s}$ defined as 
\[
W^{m,\infty}_{\Li_s} = \Big\{ f \in L^\infty_{\Li_s} \ \Big| \ \forall r \in \mathbb{N}^{2dN}, \ |r| \leq m, \ \partial^r f \in L^\infty_{\Li_s} \Big\},
\]
where $|r| = r_1+r_2+\dots+r_{2dN}$, and $\partial^r$ stands for $\partial_{q_1}^{r_1} \dots \partial_{q_{dN}}^{r_{dN}} \partial_{p_1}^{r_{dN+1}} \dots \partial_{p_{dN}}^{r_{2dN}}$.

\begin{definition}[Sufficiently smooth functions]
The set $\mathcal{S}$ of smooth functions is the set of functions $f \in L^2(\mu)$ such that, for any $m \geq 0$, there exists $s \geq 0$ (depending on $f$ and $m$) so that $f \in W^{m,\infty}_{\Li_s}$. The subset $\widetilde{\mathcal{S}} \subset \mathcal{S}$ is composed of the functions with average zero with respect to~$\mu$:
\[
\widetilde{\mathcal{S}} = \left\{ f \in \mathcal{S} \ \left| \ \int_\cE f\,d\mu = 0 \right.\right\}.
\]
\end{definition}

Some of our results will be stated in the weighted Sobolev spaces~$H^m(\mu)$ defined as 
\[
H^m(\mu) = \left\{ f \in L^2(\mu) \ \left| \ \forall r \in \mathbb{N}^{2dN}, \ |r| \leq m, \ \partial^r f \in L^2(\mu) \right.\right\},
\]
endowed with the norm
\[
\| f \|^2_{H^m(\mu)} = \| u \|^2_{L^2(\mu)} + \sum_{\substack{r \in \mathbb{N}^{2dN} \cr 1\leq |r| \leq m}} \|  \partial^r f \|^2_{L^2(\mu)}.
\]
Note that $W^{m,\infty}_{\Li_s} \subset H^m(\mu)$ since the function $\Li_s$ is in~$L^2(\mu)$. We will also occasionally need the Sobolev spaces $H^m(\kappa)$ of functions of the~$p$ variable only whose derivatives up to order~$m$ are square-integrable with respect to the probability measure~$\kappa(dp)$.

Unless stated otherwise, all the operators appearing below are by default considered as operators defined on the core $\mathcal{S}$, with range contained in $\mathcal{S}$. Some results are stated on extensions of the operators under consideration to (sub)spaces of $H^1(\mu)$ or $L^\infty_{\Li_s}$. With some abuse of notation, we will denote the extension of operators by the same letter. The appropriate domain of the operators should always be clear from the context. When an operator~$T$ is defined on the core~$\mathcal{S}$, we denote by~$T^*$ its formal adjoint, which is the operator defined on~$\mathcal{S}$ such that, for all $(f,g)\in\mathcal{S}^2$,
\[
\langle f, Tg\rangle_{L^2(\mu)} = \int_{\mathcal{E}} f(q,p) \, (Tg)(q,p) \, \mu(\d q \, \d p) = \int_{\mathcal{E}} (T^*f)(q,p) \, g(q,p) \, \mu(\d q \, \d p) = \langle T^*f, g\rangle_{L^2(\mu)}.
\]
When $T$ is a differential operator with smooth coefficient (which will be the case in many situations here), the action of the formal adjoint is found using integration by parts.

\subsection{Properties of equilibrium Langevin dynamics}
\label{sec:ppties_Lang}

Langevin dynamics can be seen as Hamiltonian dynamics perturbed by an Ornstein-Uhlenbeck process in the momenta with friction coefficient $\gamma > 0$:
\begin{equation}
\label{eq:Langevin}
\left\{ \begin{aligned}
\d{q}_t & = M^{-1} p_t \, \d{t}, \\
\d{p}_t & = -\nabla V(q_t) \, \d{t} - \gamma M^{-1} p_t \, \d{t} + \sqrt{\frac{2\gamma}{\beta}} \, \d{W}_t, 
\end{aligned} \right.
\end{equation}
where $W_t$ is a $dN$-dimensional standard Brownian motion and $M$ is the mass matrix of the system. We assume that the mass matrix is diagonal: $M = \mathrm{diag}(m_1 \mathrm{I}_d,\dots,m_N \mathrm{I}_d)$, so that momenta are Gaussian random vectors under the canonical measure, with unit covariance, and hence the components of~$p$ are very easy to sample. 
Note that we formulate here the dynamics using friction forces proportional to the velocity of the particles. 

The existence and uniqueness of strong solutions is guaranteed when the position space is compact since the kinetic energy function $1+|p|^2$ is a Lyapunov function, see for instance Theorem~5.9 in~\cite{rey-bellet}. We will sometimes denote by $(q_{\gamma,t},p_{\gamma,t})$ the solution of this equation to emphasize the dependence on the friction coefficient.

In order to describe more conveniently splitting schemes, it is useful to introduce the elementary dynamics with generators (defined on the core~$\mathcal{S}$)
\begin{equation}
\label{eq:def_ABC}
A = M^{-1} p \cdot \nabla_q, 
\qquad
B = -\nabla V(q) \cdot \nabla_p, 
\qquad
C = -M^{-1} p \cdot \nabla_p + \frac1\beta \Delta_p.
\end{equation}
The generator $\mathcal{L}_\gamma$ for equilibrium Langevin dynamics~\eqref{eq:Langevin}, defined on the core~$\mathcal{S}$, is the sum of the generators of the elementary dynamics:
\[
\mathcal{L}_\gamma = A + B + \gamma C,
\]
where $\mathcal{L}_0 = A+B$ is the generator associated with the Hamiltonian part of the dynamics. The invariance of the canonical measure $\mu$ defined in~\eqref{eq:canonical_measure} for Langevin dynamics can be rewritten in terms of the generator $\mathcal{L}_\gamma$: for any test function~$\varphi \in \mathcal{S}$,
\begin{equation}
\label{eq:inv_mu_by_L}
\int_\mathcal{E} \mathcal{L}_\gamma \varphi \, \d{\mu} = 0.
\end{equation}
In fact, the operators $A+B$ and $C$ separately preserve~$\mu$. Recall also that, thanks to the compact embedding of 
\[
H^1(\kappa) \cap \mathrm{Ker}(\pi) = \left\{ f \in H^1(\kappa) \, \left| \int_{\mathbb{R}^{dN}} f(p) \, \kappa(\d{p}) = 0 \right. \right \}
\]
in $L^2(\kappa) \cap \mathrm{Ker}(\pi)$, it is easy to show that the operator $C^{-1}$ is compact and positive definite on $L^2(\kappa) \cap \mathrm{Ker}(\pi)$. 
It is also easy to check that 
\[
(A+B)^* = -(A+B), \qquad C^* = C,
\]
where, we recall, the adjoints are formally defined as operators on~$\mathcal{S}$ through integration by parts. Note that the formal adjoint
\begin{equation}
  \label{eq:reversibility_Langevin}
  \mathcal{L}^*_\gamma = -(A + B) + \gamma C
\end{equation}
defined on~$\mathcal{S}$ has an action quite similar to the action of the generator~$\mathcal{L}_\gamma$ defined on~$\mathcal{S}$. Functional estimates valid for (extensions of) $\mathcal{L}_\gamma$ will therefore also hold for (extensions of) the formal adjoint of this operator. The equality~\eqref{eq:reversibility_Langevin} expresses the reversibility up to momentum reversal of Langevin dynamics with respect to the invariant measure~$\mu$ (see the discussion in Section~2.2.3 of~\cite{LRS10}). In particular, introducing the bounded, unitary operator on~$L^2(\mu)$
\begin{equation}
\label{eq:op_R}
(\mathcal{R}\varphi)(q,p) = \varphi(q,-p),
\end{equation}
\eqref{eq:reversibility_Langevin} can be reformulated $\mathcal{R} \mathcal{L}_\gamma \mathcal{R} = \mathcal{L}_\gamma^*$. 

\subsubsection{Ergodicity results}
\label{sec:ergodicity}

The ergodicity of Langevin dynamics for $\gamma > 0$, understood either as the almost sure convergence of time averages along a realization of the dynamics, or the long-time convergence of the law of the process to~$\mu$, is well established, see for instance~\cite{MSH01,Talay02,CLS07} and references therein. These references rely on the use of Lyapunov functions, following strategies of proofs pioneered in the Markov Chain community (see~\cite{MeynTweedie}), although alternative proofs relying on analytical tools exist (see~\cite{rey-bellet,HM11}). In any case, the evolution semigroup can be given a meaning in a weighted $L^\infty$ space, and the measure~$\mu$ is the unique invariant measure of the dynamics. This property can be translated as $\mathrm{Ker}(\Lgam) = \mathbb{C}\mathbf{1}$.

An alternative way to prove the long-time convergence of the law of the process is to use subelliptic or hypocoercive estimates as studied in~\cite{Talay02,EH03,HN04,Villani,HP08}. An important result of hypocoercivity in this case is that there exist $K_\gamma, \lambda_\gamma > 0$ such that the semigroup $\rme^{t \Lgam}$, defined on the core~$\widetilde{\mathcal{S}}$, can be extended to a bounded operator on an appropriate subspace of $H^1(\mu)$:
\begin{equation}
  \label{eq:semigroup_estimates_H1mu}
  \| \rme^{t \Lgam} \|_{\mathcal{B}(\cH^1)} \leq K_\gamma \rme^{-\lambda_\gamma t},
\end{equation}
where the subspace
\[
\cH^1 = H^1(\mu) \backslash \mathrm{Ker}(\Lgam) = \left\{ u \in H^1(\mu) \ \left| \ \int_\cE u \, \d{\mu} = 0 \right. \right\}
\]
of the Hilbert space $H^1(\mu)$ is endowed with the norm $\| u \|^2_{H^1(\mu)} = \| u \|_{L^2(\mu)}^2 + \| \nabla_p u \|_{L^2(\mu)}^2 + \| \nabla_q u\|_{L^2(\mu)}^2$, and $\| \cdot \|_{\mathcal{B}(\cH^1)}$ is the operator norm on $\cH^1$. A similar bound holds for $\rme^{t \Lgam^*}$. In particular, the operators $\mathcal{L}_\gamma$ and $\Lgam^*$ are invertible on~$\cH^1$, and 
\begin{equation}
  \label{eq:stability_H1}
  \left\| \Lgam^{-1} \right\|_{\mathcal{B}(\cH^1)} \leq \frac{K_\gamma}{\lambda_\gamma}.
\end{equation}
Note also that the same bound holds for $(\mathcal{L}_\gamma^*)^{-1}$.

For unbounded position spaces, the potential $V$ has to satisfy some assumptions for~\eqref{eq:semigroup_estimates_H1mu} to hold (such as a Poincar\'e inequality for $\rme^{-\beta V}$), but these assumptions are trivially satisfied when the position space is compact, as is the case here. An important issue is the dependence on~$\gamma$ of the constants $K_\gamma,\lambda_\gamma$, or at least the dependence on $\gamma$ of the resolvent norm $\left\| \Lgam^{-1} \right\|_{\mathcal{B}(\cH^1)}$. This is made precise in the results presented below in Section~\ref{sec:ham_limit} and~\ref{sec:ovd_lim_Lang}.

Before presenting these asymptotic estimates, let us first recall that a careful analysis of the proof presented in~\cite{Talay02}, as provided by~\cite{Kopec}, allows to prove the following result.

\begin{theorem}
\label{thm:stability_S}
The space $\widetilde{\mathcal{S}}$ is stable under~$\Lgam^{-1}$ and $(\Lgam^*)^{-1}$. 
\end{theorem}

This result is of fundamental importance in our proofs. It allows to state that, if the operators $T_1,\dots,T_M$ are well defined operators from~$\widetilde{\mathcal{S}}$ to $\widetilde{\mathcal{S}}$, then the operator $\Lgam^{-1} T_M \Lgam^{-1} \dots \Lgam^{-1} T_1 \Lgam^{-1}$ also is a well defined operator from $\widetilde{\mathcal{S}}$ to~$\widetilde{\mathcal{S}}$. 

\subsubsection{Hamiltonian limit $\gamma \to 0$}
\label{sec:ham_limit}

When $\gamma = 0$, Langevin dynamics reduces to the Hamiltonian dynamics, whose generator $A+B$ has a kernel much larger than $\mathrm{Ker}(\Lgam) = \mathbb{C}\mathbf{1}$. It is therefore expected that the operator norm of $\Lgam^{-1}$ diverges as $\gamma \to 0$. The rate of divergence is made precise in the following theorem, summarizing the results from Theorem~1.6 and Proposition~6.3 of~\cite{HP08}.

\begin{theorem}[see~\cite{HP08}]
\label{thm:Ham_limit_Lgam}
Denote by $\| \cdot \|_{\mathcal{B}(\cH^0)}$ the operator norm on the subspace
\begin{equation}
\label{eq:def_H0}
\cH^0 = \left\{ u \in L^2(\mu) \ \left| \ \int_\cE u \, \d{\mu} = 0 \right. \right\}
\end{equation}
of the Hilbert space $L^2(\mu)$. There exists two constants $c_-, c_+ > 0$ such that, for any $0 < \gamma \leq 1$,
\[
\frac{c_-}{\gamma} \leq \left\| \mathcal{L}_\gamma^{-1} \right\|_{\mathcal{B}(\mathcal{H}^0)} \leq \frac{c_+}{\gamma}.
\]
\end{theorem}

We state the result with the upper bound $\gamma \leq 1$, but it holds in fact for $0 < \gamma \leq \gamma_{\rm max}$ for any finite value $\gamma_{\rm max}>0$. Note also that the same bound holds for $(\mathcal{L}_\gamma^*)^{-1}$.

\subsubsection{Overdamped limit $\gamma \to +\infty$}
\label{sec:ovd_lim_Lang}

The overdamped limit can be obtained by either letting the friction go to infinity in~\eqref{eq:Langevin} together with an appropriate rescaling of time; or by letting masses go to~0. When discussing overdamped limits in this article, we will always set the mass matrix~$M$ to identity and consider the limit $\gamma \to +\infty$. Since we restrict our attention to the invariant measure of the system, the time rescaling is not relevant. 

Let us describe more precisely the convergence result. It is shown in Section~2.2.4 of~\cite{LRS10} for instance that the solutions of~\eqref{eq:Langevin} observed over long times, namely $(q_{\gamma,\gamma s},p_{\gamma,\gamma s})_{s \geq 0}$, converge pathwise on finite time intervals $s \in [0,t]$ to the solutions of overdamped Langevin dynamics
\begin{equation}
\label{eq:overdamped}
\d{Q}_t = -\nabla V(Q_t) \, \d{t} + \sqrt{\frac2\beta} \, \d{W}_t,
\end{equation}
with the same initial condition $Q_0 = q_{\gamma,0}$. The process~\eqref{eq:overdamped} is ergodic on the compact position space $\mathcal{M}$, with unique invariant probability measure~$\overline{\mu}(\d{q})$ defined in~\eqref{eq:marginal_mu}. Its generator
\[
\Lovd = -\nabla V(q) \cdot \nabla_q + \frac1\beta \Delta_q,
\]
defined on the core~$\mathcal{S} \cap \mathrm{Ker}(\pi) = C^\infty(\mathcal{M})$, is an elliptic operator which is symmetric on~$L^2(\overline{\mu})$, with compact resolvent (see for instance the discussion and the references in Section~2.3.2 of~\cite{LRS10}). It is easy to see that the inverse operator $\Lovd^{-1}$ can be extended to a bounded operator from 
\[
\widetilde{H}^m(\overline{\mu}) = \left\{ \varphi \in H^m(\overline{\mu}) \, \left| \int_\mathcal{M} \varphi \, \d{\overline{\mu}} = 0 \right.\right\}
\]
to $\widetilde{H}^{m+2}(\overline{\mu})$. Let us finally mention that the set of $C^\infty(\mathcal{M})$ functions with average zero with respect to~$\overline{\mu}$ is of course stable with respect to~$\Lovd^{-1}$. 

The following result gives bounds on the resolvent of the Langevin generator in the overdamped regime, and in fact quantifies the difference between the resolvent $\Lgam^{-1}$ and the resolvent $\Lovd^{-1}$ appropriately rescaled by a factor~$\gamma$. 

\begin{theorem}
\label{lem:bounds_CL_gamma}
There exist two constants $c_-, c_+ > 0$ such that, for any $\gamma \geq 1$,
\begin{equation}
\label{eq:crude_bounds_Lgamma_ovd}
c_- \gamma \leq \| \cL_\gamma^{-1} \|_{\mathcal{B}(\cH^1)} \leq c_+ \gamma.
\end{equation}
More precisely, there exists a constant $K > 0$ such that, for any $\gamma \geq 1$,
\begin{equation}
\label{eq:divergent_behavior_Lgamma}
\begin{aligned}
\left\| \cL_\gamma^{-1} - \gamma \Lovd^{-1} \pi - p^T \nabla_q\Lovd^{-1}\pi + \Lovd^{-1} \pi (A+B) C^{-1} (\Id - \pi) \right \|_{\mathcal{B}(\cH^1)} \leq \frac{K}{\gamma}, \\
\left\| \left(\cL_\gamma^*\right)^{-1} - \gamma \Lovd^{-1} \pi + p^T \nabla_q\Lovd^{-1}\pi - \Lovd^{-1} \pi (A+B) C^{-1} (\Id - \pi) \right \|_{\mathcal{B}(\cH^1)} \leq \frac{K}{\gamma},
\end{aligned}
\end{equation}
where the operator $\pi$ is defined in~\eqref{eq:pi}, and $(C^{-1}\psi)(q,p)$ is understood as applying the operator $C^{-1}$ to the function $\psi(q,\cdot) \in L^2(\kappa)$ for all values of $q \in \mathcal{M}$.
\end{theorem}

Note that the function $\Lovd^{-1} \pi f$ is well defined since, as $f$ belongs to $\cH^1$, the function~$\pi f$ has a vanishing average with respect to~$\overline{\mu}$. The fact that $\Lovd^{-1} \pi (A+B) C^{-1} (\Id - \pi)$ is bounded on~$\cH^1$ is discussed in the proof of Theorem~\ref{lem:bounds_CL_gamma}. An important ingredient in the proof is the following estimate, which we call uniform hypocoercivity estimate. 

\begin{lemma}[Uniform hypocoercivity for large frictions]
\label{lem:bounded_resolvent_perp}
Consider the following subspace of $\mathcal{H}^1$:
\[
\cH_\perp^1 = \left\{ u \in \cH^1 \ \left| \ \overline{u}(q) = \int_{\RR^{dN}} u(q,p) \, \kappa(\d{p}) = 0 \right. \right\}. 
\]
There exists a constant $K > 0$ such that, for any $\gamma \geq 1$,
\[
\forall f \in \cH^1_\perp, \qquad \| \cL_\gamma^{-1} f \|_{H^1(\mu)} \leq K \| f \|_{H^1(\mu)}.
\]
\end{lemma}

The proofs of Theorem~\ref{lem:bounds_CL_gamma} and Lemma~\ref{lem:bounded_resolvent_perp} are provided in Section~\ref{sec:proof_L_gamma_large}.

\subsection{Splitting schemes for equilibrium Langevin dynamics}
\label{sec:splitting_schemes}

We present in this section the splitting schemes to be examined in this article. These schemes can be described by evolution operators $P_\dt$ defined on the core~$\mathcal{S}$ (but which can be  extended to bounded operators on~$L^\infty(\mathcal{E})$), and which are such that the Markov chain $(q^n,p^n)$ generated by the discretization satisfies
\[
P_\dt \psi(q,p) = \mathbb{E}\Big( \psi\left(q^{n+1},p^{n+1}\right)\Big| (q^n,p^n) = (q,p)\Big).
\]
We also briefly give some ergodicity results obtained by minor extensions or variations of existing results in the literature (see in particular~\cite{MSH01,Talay02,BO10,BH13}). Since these ergodicity issues are by now a rather standard and well-understood matter, especially for compact position spaces, we provide only elements of proofs in Section~\ref{sec:proof_ergodicity_MC}.

\subsubsection{First-order splitting schemes}
\label{sec:splitting_schemes_1st}

First-order schemes are obtained by a Lie-Trotter splitting of the elementary evolutions generated by $A,B,\gamma C$. The motivation for this splitting is that all elementary evolutions are analytically integrable (see the expressions of the associated semigroups in~\eqref{eq:analytic_expressions_semigroups}). There are 6 possible schemes, whose evolution operators (defined on the core~$\mathcal{S}$) are of the general form
\[
P^{Z,Y,X}_\dt = \rme^{\dt Z} \rme^{\dt Y} \rme^{\dt X}, 
\]
with all possible permutations $(Z,Y,X)$ of $(A, B, \gamma C)$. For instance, the numerical scheme associated with $P_\dt^{B,A,\gamma C}$ is
\begin{equation}
\label{eq:Langevin_splitting}
\left\{ \begin{aligned}
\widetilde{p}^{n+1} & = p^n - \dt \, \nabla V(q^{n}), \\
q^{n+1} & = q^n + \dt \, M^{-1} \widetilde{p}^{n+1}, \\
p^{n+1} & = \alpha_\dt \widetilde{p}^{n+1} + \sqrt{\frac{1-\alpha^2_\dt}{\beta}M} \, G^n, 
\end{aligned} \right.
\end{equation}
where $\alpha_\dt = \exp(-\gamma M^{-1} \dt)$, and $(G^n)$ are independent and identically distributed Gaussian random vectors with identity covariance. The simulation of the dynamics with generator~$C$ is very simple for diagonal mass matrix~$M$ since $\alpha_\dt$ is a diagonal matrix.    Note that the order of the operations performed on the configuration of the system is the inverse of the order of the operations mentioned in the superscript of the evolution operator $P_\dt^{B,A,\gamma C}$ when read from right to left. This inversion is known as \emph{Vertauschungssatz} (see for instance the discussion in Section~III.5.1 of~\cite{HairerLubichWanner06}). It arises from the fact that the numerical method modifies the distribution of the variables, whereas the evolution operator encodes the evolution of observables (determined by the adjoint of the operator encoding the evolution of the distribution).


The iterations of the three schemes associated with $P^{\gamma C,B,A}_\dt,P^{B,A,\gamma C}_\dt,P^{A,\gamma C,B}_\dt$ share a common sequence of update operations, as for $P^{\gamma C,A,B}_\dt,P^{A,B,\gamma C}_\dt,P^{B,\gamma C,A}_\dt$. More precisely, we mean that equalities of the following form hold:
\begin{equation}
\label{eq:similar_evolution_operators}
\left(P^{A,B,\gamma C}_\dt\right)^n = T_{ \dt} \left(P^{\gamma C,A,B}_\dt\right)^{n-1} U_{\gamma,\dt}, 
\qquad 
U_{\gamma,\dt} = \rme^{\gamma \dt C},
\qquad
T_{ \dt} = \rme^{\dt A} \rme^{\dt B}.
\end{equation}
It is therefore not surprising that the invariant measures of the schemes with operators composed in the same order have very similar properties, as made precise in Theorem~\ref{thm:error_first_order_schemes}, relying on Lemma~\ref{lem:TU}.

\subsubsection{Second-order schemes}
\label{sec:splitting_schemes_2nd}

Second-order schemes are obtained by a Strang splitting of the elementary evolutions generated by $A,B,\gamma C$. There are also 6 possible schemes, which are of the general form
\[
P^{Z,Y,X,Y,Z}_\dt = \rme^{\dt Z/2} \rme^{\dt Y/2} \rme^{\dt X} \rme^{\dt Y/2} \rme^{\dt Z/2}, 
\]
with the same possible orderings as for first-order schemes. Again, these schemes can be classified into three groups depending on the ordering of the operators once the elementary one-step evolution is iterated: (i)~$P_\dt^{\gamma C,B,A,B,\gamma C}, P_\dt^{A,B,\gamma C,B,A}$, (ii)~$P_\dt^{\gamma C,A,B,A,\gamma C}, P_\dt^{B,A,\gamma C,A,B}$, and (iii)~$P_\dt^{B,\gamma C,A,\gamma C,B}, P_\dt^{A,\gamma C,B,\gamma C,A}$. We discard the latter category since the invariant measures of the associated numerical schemes are not consistent with $\overline{\mu}$ in the overdamped limit (see Section~\ref{sec:ovd_limit}).

\subsubsection{Geometric Langevin Algorithms}
\label{sec:splitting_schemes_GLA}

In fact, as already proved in~\cite{BO10} (see also Corollary~\ref{cor:error_GLA} below), second order accuracy of the invariant measure can be obtained by resorting to a first-order splitting between the Hamiltonian and the Ornstein-Uhlenbeck parts, and discretizing the Hamiltonian part with a second-order scheme. This corresponds to the following evolution operators of Geometric Langevin Algorithm (GLA) type:
\begin{equation}
\label{eq:GLA_schemes}
\begin{aligned}
P_\dt^{\gamma C,A,B,A} = \rme^{\gamma \dt C} \rme^{\dt A /2} \rme^{\dt B} \rme^{\dt A/2}, & \qquad P_\dt^{\gamma C,B,A,B} = \rme^{\gamma \dt C} \rme^{\dt B /2} \rme^{\dt A} \rme^{\dt B/2}, \\
P_\dt^{A,B,A,\gamma C} = \rme^{\dt A /2} \rme^{\dt B} \rme^{\dt A/2} \rme^{\gamma \dt C}, & \qquad P_\dt^{B,A,B,\gamma C} = \rme^{\dt B /2} \rme^{\dt A} \rme^{\dt B/2} \rme^{\gamma \dt C}.
\end{aligned}
\end{equation}

\subsection{Ergodicity results for splitting schemes}
\label{sec:ergo_num_scheme}

Let us now give some technical results on the ergodic behavior of the splitting schemes presented in Section~\ref{sec:splitting_schemes}.In this section we denote the evolution operator by $P_\dt$   (supressing the dependence on the friction parameter~$\gamma$ although the constants appearing in the results below a priori depend on this parameter). Ergodicity results for a fixed value of $\dt$ are obtained with techniques similar to the ones presented in~\cite{MeynTweedie}, by mimicking the proofs presented for certain discretization schemes of the Langevin equation in~\cite{MSH01,Talay02,BO10}. A more subtle point is to obtain rates of convergence which are uniform in the timestep~$\dt$, as done in~\cite{BH13} for a class of Metropolis-Hastings schemes based on a discretization of overdamped Langevin dynamics in unbounded spaces as the proposal. We are able here to prove such results by relying on the fact that the position space~$\mathcal{M}$ is compact.

\bigskip

The proof is based on two preliminary results, namely a uniform drift inequality or Lyapunov condition and a uniform minorization condition (see Section~\ref{sec:proof_ergodicity_MC} for the proofs). The term uniform refers to estimates which are independent of the timestep~$\dt$. To obtain such estimates, we have to consider evolutions over fixed times $T \simeq n\dt$, which amounts to iterating the elementary evolution $P_\dt$ over $\lceil T/\dt \rceil$ timesteps (where $\lceil x \rceil$ denotes the smallest integer larger than~$x$).

\begin{lemma}[Uniform Lyapunov condition]
\label{lem:Lyapunov}
For any $s^* \in \mathbb{N}^*$, there exist $\dt^* > 0$ and $C_a,C_b >0$ such that, for any $1 \leq s \leq s^*$ and $0 < \dt \leq \dt^*$,
\begin{equation}
\label{eq:moment_estimates}
P_\dt \Li_s \leq \rme^{-C_a \dt} \Li_s + C_b \dt. 
\end{equation}
In particular, for any $T > 0$,
\begin{equation}
\label{eq:moment_estimates_uniform}
P_\dt^{\lceil T/\dt \rceil} \Li_s \leq \exp(-C_a T) \Li_s + \frac{C_b \dt}{1 - \rme^{-C_a \dt}}.
\end{equation}
\end{lemma}

\begin{lemma}[Uniform minorization condition]
\label{lem:minorization}
Consider $T > 0$ sufficiently large, and fix any $p_{\rm max} > 0$. There exist $\dt^*, \alpha > 0$ and a probability measure~$\nu$ such that, for any bounded, measurable non-negative function $f$, and any $0 < \dt \leq \dt^*$,
\[
\inf_{|p| \leq p_{\rm max}} \left( P_\dt^{\lceil T/\dt \rceil}f \right)(q,p) \geq \alpha \int_\cE f(q,p) \, \nu(\d{q} \, \d{p}).
\]
\end{lemma}

Lemma~\ref{lem:minorization} ensures that Assumption~2 in~\cite{HM11} holds for any choice of Lyapunov function~$\Li_s$ ($s \geq 1$), provided $p_{\rm max}$ is chosen to be sufficiently large. The uniform minorization condition can formally be rewritten as
\[
\forall (q_0,p_0) \in \mathcal{M} \times B(0,p_{\rm max}), 
\qquad 
P_\dt\Big((q_0,p_0),\d{q}\,\d{p}\Big) \geq \alpha \nu(\d{q}\,\d{p}). 
\]
We present a direct proof of Lemma~\ref{lem:minorization} in Section~\ref{sec:proof_ergodicity_MC}. Extending this result to unbounded position spaces is much more difficult in general, see for instance the recent works~\cite{KV06,KV13} and~\cite{BH13} where non-degeneracy of the noise is assumed.

\medskip

Let us now precisely state the ergodicity result. 

\begin{proposition}[Ergodicity of numerical schemes]
\label{prop:ergodicity_MC}
Fix $s^* \geq 1$. For any $0 < \gamma < +\infty$, there exists $\dt^*>0$ such that, for any $0 < \dt \leq \dt^*$, the Markov chain associated with $P_\dt$ has a unique invariant probability measure $\mu_{\gamma,\dt}$, which admits a density with respect to the Lebesgue measure $\d{q} \, \d{p}$, and has finite moments: There exists $R>0$ such that, for any $1 \leq s \leq s^*$,
\begin{equation}
\label{eq:moment_estimate}
\int_\cE \Li_s \, \d{\mu}_{\gamma,\dt} \leq R < +\infty,
\end{equation}
uniformly in the timestep $\dt$. There also exist $\lambda, K > 0$ (depending on $s^*$ and $\gamma$ but not on~$\dt$) such that, and for all functions $f \in L^\infty_{\Li_s}$, the following holds for almost all $(q,p) \in \cE$:
\begin{equation}
\label{eq:ergodicity_num}
\forall n \in \mathbb{N}, \qquad \left| \left(P_\dt^n f\right)(q,p) - \int_\cE f \d{\mu}_{\gamma,\dt} \right| \leq K \, \Li_s(q,p) \, \rme^{-\lambda n \dt} \, \| f \|_{L^\infty_{\Li_s}}.
\end{equation}
\end{proposition}

Let us again emphasize that, compared to the results of~\cite{MSH01,Talay02,BO10}, the only new estimate is the uniform-in-$\dt$ decay rate in~\eqref{eq:ergodicity_num} as obtained in~\cite{BH13} for Metropolis schemes. These uniform estimates follow from an application of the results of~\cite{HM11} to the sampled chain $P_\dt^{\lceil T/\dt \rceil}$ (see Section~\ref{sec:proof_ergodicity_MC} for more detail). Recall also that the convergence rates we obtain of course depend on the friction parameter~$\gamma$.

\bigskip

An interesting consequence of the above estimates is that we are able to obtain  uniform control of the resolvent of the operator $\Id-P_\dt$ extended to appropriate Banach spaces. Such a bound will prove useful to control approximation errors in Green-Kubo type formulas (see Section~\ref{sec:num_estimation_correction}). Note indeed that the estimate~\eqref{eq:ergodicity_num} implies the operator bound 
\[
\left \| P_\dt^n \right\|_{\mathcal{B}(L^\infty_{\Li_s,\dt})} \leq K \, \rme^{-\lambda n \dt},
\]
on the Banach space
\[
L^\infty_{\Li_s,\dt} = \left\{ \psi \in L^\infty_{\Li_s} \, \left| \, \int_\cE \psi \, \d{\mu}_{\gamma,\dt} = 0 \right. \right\}.
\]
The Banach space $L^\infty_{\Li_s,\dt}$ depends both on~$\dt$ and $\gamma$ through $\mu_{\gamma,\dt}$, although the dependence on $\gamma$ is not explicitly written. This proves that the series
\[
\sum_{n=0}^{+\infty} P_\dt^n 
\]
is well defined as a bounded operator on $L^\infty_{\Li_s,\dt}$, and is in fact equal to $(\Id-P_\dt)^{-1}$ since
\[
(\Id-P_\dt) \sum_{n=0}^{+\infty} P_\dt^n = \Id.
\]
We also have the bound
\[
\left\| \left(\Id - P_\dt\right)^{-1} \right \|_{\mathcal{B}(L^\infty_{\Li_s,\dt})} \leq \sum_{n=0}^{+\infty} \left \| P_\dt^n \right\|_{\mathcal{B}(L^\infty_{\Li_s,\dt})} \leq \frac{K}{1-\rme^{-\lambda \dt}} \leq \frac{2K}{\lambda \dt}
\]
provided $\dt$ is sufficiently small. Let us summarize this result as follows.

\begin{corollary}
\label{corr:resolvent_estimates_I_Pdt}
For any $s^* \in \mathbb{N}^*$, there exist $\dt^* > 0$ and $R > 0$ such that, for all $0 \leq s \leq s^*$, a uniform resolvent bound holds: for any $0 < \dt \leq \dt^*$, 
\begin{equation}
\label{eq:unif_bound_Linf_Lis}
\left\| \left(\frac{\Id - P_\dt}{\dt}\right)^{-1} \right \|_{\mathcal{B}(L^\infty_{\Li_s,\dt})} \leq R.
\end{equation}
\end{corollary}

\subsection{Error estimates for finite frictions}
\label{sec:error_estimates_finite_friction}

In this section we study the error of the average of sufficiently smooth functions, which allows us to characterize the corrections to the invariant measure. In Theorems~\ref{thm:error_first_order_schemes} and~\ref{thm:error_second_order_schemes}, below, we characterize all the first- and second-order splittings; the technique of proof allows us to provide a rigorous study of the error estimates in the overdamped regime (see Section~\ref{sec:ovd_limit}) and for nonequilibrium systems (see Section~\ref{sec:noneq_systems}). 

\begin{remark}
If only the order of magnitude of the correction is of interest, and not the expression of the correction in itself, no regularity result with regard to the derivatives is required (see~\cite{BO10}), in contrast to situations where such corrections are explicitly considered, as in~\cite{Talay02} for instance.
\end{remark}

\subsubsection{Relating invariant measures of two numerical schemes}

We classified in Section~\ref{sec:splitting_schemes} the numerical schemes according to the order of appearance of the elementary operators. More precisely, we considered schemes to be similar when the global ordering of the operators is the same but the operations are started and ended differently, as in~\eqref{eq:similar_evolution_operators} above (see also~\eqref{eq:relation_Q_P} below for an abstract definition).
This choice of classification is motivated by the following lemma which demonstrates how we may straightforwardly obtain the expression of the invariant measure of one scheme when the expression for another one is given.

We state the result in an abstract fashion for two schemes $P_\dt = U_\dt T_\dt$ and $Q_\dt = T_\dt U_\dt$ (which implies the condition~\eqref{eq:relation_Q_P} below). See~\eqref{eq:similar_evolution_operators} for a concrete example. 

\begin{lemma}[Here and elsewhere: TU lemma]
\label{lem:TU}
Consider two numerical schemes with associated evolution operators $P_\dt,Q_\dt$ bounded on~$L^\infty(\mathcal{E})$, for which there exist bounded operators $U_\dt,T_\dt$ on~$L^\infty(\mathcal{E})$ such that, for all $n \geq 1$,
\begin{equation}
\label{eq:relation_Q_P}
Q_\dt^n = T_\dt P_\dt^{n-1} U_\dt.
\end{equation}
We also assume that both schemes are ergodic with associated invariant measures denoted respectively by $\mu_{P,\dt}$, $\mu_{Q,\dt}$: For almost all $(q,p) \in \cE$ and $f \in L^\infty(\mathcal{E})$,
\begin{equation}
\label{eq:ergoditicy_TU}
\lim_{n \to +\infty} P_\dt^n f(q,p) = \int_\cE f \, \d{\mu}_{P,\dt}, 
\qquad 
\lim_{n \to +\infty} Q_\dt^n f(q,p) = \int_\cE f \, \d{\mu}_{Q,\dt}.
\end{equation}
Then, for all $\varphi \in L^\infty(\mathcal{E})$, 
\begin{equation}
\label{eq:TU_relation}
\int_\cE \varphi \, \d{\mu}_{Q,\dt} = \int_\cE \left(U_\dt \varphi \right) \d{\mu}_{P,\dt}.
\end{equation}
\end{lemma}

Ergodicity results such as~\eqref{eq:ergoditicy_TU} are implied by conditions such as~\eqref{eq:ergodicity_num}.

\medskip

\begin{proof}
The proof of this result relies on the simple observation that, for a given initial measure $\rho$ with a smooth density with respect to the Lebesgue measure, the ergodicity assumption ensures that, for a bounded measurable function~$\varphi$, 
\[
\int_\cE \varphi \, \d{\mu}_{Q,\dt} = \lim_{n \to +\infty} \int_\cE Q_\dt^n \varphi \, \d{\rho} = \lim_{n \to +\infty} \int_\cE T_\dt P_\dt^{n-1} \left(U_\dt \varphi \right) \, \d{\rho}.
\]
Now, we use the ergodicity property~\eqref{eq:ergoditicy_TU} with $f$ replaced by $U_\dt \varphi$ to obtain the following convergence for almost all $(q,p) \in \mathcal{E}$:
\[
\lim_{n \to +\infty} P_\dt^{n-1} \left(U_\dt \varphi \right)(q,p) = \int_\cE U_\dt \varphi \, \d{\mu}_{P,\dt} = a_\dt.
\]
Since $T_\dt$ preserves constant functions, there holds
\[
\int_\mathcal{E} T_\dt (a_\dt \mathbf{1}) \, \d\rho = a_\dt \int_\mathcal{E} \mathbf{1} \, \d\rho = a_\dt,
\]
which finally gives~\eqref{eq:TU_relation}.
\end{proof}

\medskip

Let us now show how we will use Lemma~\ref{lem:TU} in the sequel. Assume that a weak error estimate holds on the invariant measure $\mu_{P,\dt}$: there exist $\alpha \geq 1$ and a function $f_\alpha \in \mathcal{S}$ such that
\[
\int_\cE \psi \, \d{\mu}_{P,\dt} = \int_\cE \psi \, \d{\mu} + \Delta t^\alpha \int_\cE \psi \, f_\alpha \, \d{\mu} + \Delta t^{\alpha + 1} r_{\psi,\alpha,\dt},
\]
with $|r_{\psi,\alpha,\dt}| \leq K$ for $\dt$ sufficiently small.
Combining this equality and~\eqref{eq:TU_relation}, the following expansion is obtained for $\mu_{Q,\dt}$:
\[
\int_\cE \psi \, \d{\mu}_{Q,\dt} = \int_\cE \left(U_\dt \psi \right) \d{\mu}_{P,\dt}
= \int_\cE \left(U_\dt \psi \right) \d{\mu} + \Delta t^\alpha \int_\cE \left(U_\dt \psi \right) f_\alpha \, \d{\mu} + \Delta t^{\alpha + 1} r_{U_\dt\psi,\alpha,\dt}.
\]
In general, for an evolution operator $U_\dt$ preserving the measure $\mu$ at order~$\delta \geq 1$, we can write 
\[
U_\dt = \Id + \Delta t \, \mathcal{A}_1 + \dots  + \Delta t^{\delta-1} \mathcal{A}_{\delta-1} + \Delta t^{\delta} \, S_\delta + \Delta t^{\delta + 1} \, R_{\delta, \dt},
\]
where all the operators on the right hand side are defined on the core $\mathcal{S}$, and the operators $\mathcal{A}_k$ preserve the measure~$\mu$:
\[
\forall \varphi \in \mathcal{S}, \qquad \int_\mathcal{E} \mathcal{A}_k \varphi \, d\mu = 0,
\]
while the operator $S_{\delta}$ does not. Typically, $\mathcal{A}_k$ is a composition of the operators~$A+B$ and~$C$. In addition, for a given function $\varphi \in \mathcal{S}$, the remainder $R_{\delta, \dt} \varphi$ is uniformly bounded for~$\dt$ sufficiently small. Three cases should then be distinguished:
\begin{enumerate}[(i)]
\item When $\delta \geq \alpha+1$, the weak error in the invariant measure $\mu_{Q,\dt}$ is of the same order as for $\mu_{P,\dt}$ since
\[
\int_\cE \psi \, \d{\mu}_Q = \int_\cE \psi \, \d{\mu} + \Delta t^\alpha \int_\cE \psi \, f_\alpha \, \d{\mu} + \Delta t^{\alpha + 1} \widetilde{r}_{\psi,\alpha,\delta,\dt}.
\]
\item For $\delta \leq \alpha-1$, the weak error in the invariant measure $\mu_Q$ arises at dominant order from the operator~$U_\dt$:
\[
\int_\cE \psi \, \d{\mu}_Q = \int_\cE \psi \, \d{\mu} + \Delta t^\delta \int_\cE \psi \, \left(S_\delta^*\mathbf{1}\right) \, \d{\mu} + \Delta t^{\delta + 1} \widetilde{r}_{\psi,\alpha,\delta,\dt}.
\]
\item The interesting case corresponds to $\alpha = \delta$. In this situation, 
\begin{equation}
\label{eq:interesting_TU_situation}
\int_\cE \psi \, \d{\mu}_Q = \int_\cE \psi \, \d{\mu} + \Delta t^\alpha \int_\cE \psi \, \left(f_\alpha+S_\alpha^*\mathbf{1}\right) \, \d{\mu} + \Delta t^{\alpha + 1} \widetilde{r}_{\psi,\alpha,\delta,\dt}.
\end{equation}
An increase in the order of the error on the invariant measure is obtained when the leading order correction vanishes for all admissible observables~$\psi$, that is, if and only if $f_\alpha+S_\alpha^*\mathbf{1} = 0$.
\end{enumerate}

\subsubsection{First-order schemes}

The following result characterizes at leading order the invariant measure of the schemes based on a first-order splitting (see Section~\ref{sec:splitting_schemes_1st}). We first study the error estimates in the invariant measure of the schemes $P_\dt^{\gamma C, B, A}$, $P_\dt^{\gamma C, A,B}$ (which can be interpreted as GLA schemes with a symplectic Euler discretization of the Hamiltonian part, see~\cite{BO10}), and then deduce error estimates for the four remaining schemes introduced in Section~\ref{sec:splitting_schemes_1st} by making use of Lemma~\ref{lem:TU}. The proof can be read in Section~\ref{sec:proof_thm:error_first_order_schemes}.

\begin{theorem}
\label{thm:error_first_order_schemes}
Consider any of the first order splittings presented in Section~\ref{sec:splitting_schemes_1st}, and denote by $\mu_{\gamma,\dt}(\d{q} \, \d{p})$ its invariant measure. Then there exists a function $f_{1,\gamma} \in \widetilde{\mathcal{S}}$ such that, for any function $\psi \in \mathcal{S}$,
\begin{equation}
\label{eq:error_first_order_schemes}
\int_\cE \psi(q,p) \, \mu_{\gamma,\dt}(\d{q} \, \d{p}) = \int_\cE \psi(q,p) \, \mu(\d{q} \, \d{p}) + \dt \int_\cE \psi(q,p) f_{1,\gamma}(q,p) \, \mu(\d{q} \, \d{p}) + \dt^2 r_{\psi,\gamma,\dt},
\end{equation}
where the remainder $r_{\psi,\gamma,\dt}$ is uniformly bounded for $\dt$ sufficiently small. The expressions of the correction functions $f_{1,\gamma}$ depend on the numerical scheme at hand. They are defined as
\begin{equation}
\label{eq:correction_first_order_schemes}
\begin{aligned}
\Lgam^* f_{1}^{\gamma C,B,A} & = -\frac12 (A+B)g, \qquad g(q,p) = \beta p^T M^{-1} \nabla V(q),\\
f_{1}^{\gamma C,A,B} & = f_{1}^{A,B,\gamma C} = - f_{1}^{B,A,\gamma C} = - f_{1}^{\gamma C,B,A},\\
f_{1}^{A,\gamma C,B} & = - f_{1}^{B,\gamma C,A} = f_{1}^{\gamma C,B,A} - g.\\ 
\end{aligned}
\end{equation}
\end{theorem}

It would in fact possible to obtain bounds on the the remainder $r_{\psi,\gamma,\dt}$ with respect to~$\psi$, thanks to functional inequalities given in Appendix~A of~\cite{Kopec}.

\begin{remark}
The equations~\eqref{eq:correction_first_order_schemes} could be analytically solved if, instead of the fluctuation/dissipation operator $C$, we were using the mass-weighted differential operator as in~\cite{LM12}:
\[
C_M = -p^T \nabla_p + \frac1\beta M : \nabla_p^2.
\]
The corresponding generator $\mathcal{L}_{\gamma,M} = A+B + \gamma C_M$ defined on the core $\mathcal{S}$ is associated with Langevin dynamics where the friction force is proportional to the momenta rather than velocities. A simple computation shows that
\[
-\frac12 (A+B)g = \mathcal{L}_{\gamma,M}^* \left(\frac{\beta}{2} V - g\right).
\]
The condition~\eqref{eq:correction_first_order_schemes} would be replaced by $\mathcal{L}_{\gamma,M}^* f_1^{\gamma C,B,A} = -(A+B)g/2$, so that $f_1^{\gamma C,B,A} = \beta V/2 - g + c$ where $c$ is a constant ensuring that $f_1^{\gamma C,B,A}$ has a vanishing average with respect to~$\mu$.
\end{remark}

\subsubsection{Hamiltonian limit of the correction term}

For first order splitting schemes, the limit of the leading order correction term in~\eqref{eq:error_first_order_schemes} can be studied in the limit when $\gamma \to 0$. Not surprisingly, it turns out that the leading order correction is the first term in the expansion of the modified Hamiltonian of the symplectic Euler method in powers of $\dt$. In contrast to the more complete proof we are able to present for the overdamped limit (see Section~\ref{sec:ovd_limit}), we were not able to study the behavior of the remainder terms $r_{\psi,\gamma,\dt}$ in~\eqref{eq:error_first_order_schemes}. There is a technical obstruction to controlling these remainders from the way we prove our results since the limiting operator $\mathcal{L}_0 = A+B$ is not invertible. Let us also mention that studying the corresponding Hamiltonian limit for second order schemes turns out to be a much more difficult question (see Remark~\ref{rmk:Hamiltonian_limit_second_order}).

\begin{proposition}
\label{prop:Ham_limit_correction}
There exists a constant $K > 0$ such that, for all $0 < \gamma \leq 1$,
\[
\left\| f_1^{\gamma C, B,A} - \frac\beta2 p^T M^{-1} \nabla V \right\|_{L^2(\mu)} \leq K \gamma,
\]
with similar estimates for $f_1^{B, \gamma C, A}$ and $f_1^{B,A,\gamma C}$; and
\[
\left\| f_1^{\gamma C, A,B} + \frac\beta2 p^T M^{-1} \nabla V \right\|_{L^2(\mu)} \leq K \gamma,
\]
with similar estimates for $f_{1}^{A,\gamma C,B}$ and $f_{1}^{A,B,\gamma C}$.
\end{proposition}

The proof of this result is provided in Section~\ref{sec:proof_Ham_limit}.

\subsubsection{Second-order schemes}

The following result characterizes at leading order the invariant measure of the schemes based on a second-order splitting (see Section~\ref{sec:splitting_schemes_2nd}).

\begin{theorem}
\label{thm:error_second_order_schemes}
Consider any of the second order splittings presented in Section~\ref{sec:splitting_schemes_2nd}, and denote by $\mu_{\gamma,\dt}(\d{q} \, \d{p})$ its invariant measure. Then there exists a function $f_{2,\gamma} \in \widetilde{\mathcal{S}}$ such that, for any function $\psi \in \mathcal{S}$,
\begin{equation}
\label{eq:error_second_order_schemes}
\int_\cE \psi(q,p) \, \mu_{\gamma,\dt}(\d{q} \, \d{p}) = \int_\cE \psi(q,p) \, \mu(\d{q} \, \d{p}) + \dt^2 \int_\cE \psi(q,p) f_{2,\gamma}(q,p) \, \mu(\d{q} \, \d{p}) + \dt^4 r_{\psi,\gamma,\dt},
\end{equation}
where the remainder $r_{\psi,\gamma,\dt}$ is uniformly bounded for $\dt$ sufficiently small. The expressions of the correction functions $f_{2,\gamma}$ depend on the numerical scheme at hand. They are defined as
\begin{equation}
\label{eq:correction_second_order_schemes}
\begin{aligned}
\Lgam^* f_{2}^{\gamma C,B,A,B,\gamma C} & = \frac{1}{12} (A+B) \left[\left(A+\frac{B}{2}\right)g\right], \qquad g(q,p) = \beta p^T M^{-1} \nabla V(q),\\
\Lgam^* f_{2}^{\gamma C,A,B,A,\gamma C} & = -\frac{1}{12} (A+B) \left[\left(B+\frac{A}{2}\right)g\right], \\
f_{2}^{A,B,\gamma C,B,A} & = f_{2}^{\gamma C,B,A,B,\gamma C} + \frac18 (A+B)g, \\
f_{2}^{B,A,\gamma C,A,B} & = f_{2}^{\gamma C,A,B,A,\gamma C} - \frac18 (A+B)g. \\
\end{aligned}
\end{equation}
\end{theorem}

It can be checked that the expressions of $f_{2}^{B,A,\gamma C,A,B}$ and $f_{2}^{A,B,\gamma C,B,A}$ agree with the ones presented in~\cite{LM12}. 
Let us emphasize that no $\dt^3$ correction term appears in~\eqref{eq:error_second_order_schemes} after the $\dt^2$ term. In fact, a more careful treatment would allow us to write an error expansion in terms of higher orders of $\dt$, with only even powers of $\dt$ appearing.

The proof of this result is given in Section~\ref{sec:proof_thm:error_second_order_schemes}. We use as reference schemes for the proofs the schemes $P_\dt^{\gamma C,A,B,A,\gamma C}$, $P_\dt^{\gamma C,B,A,B,\gamma C}$. These schemes indeed turn out to be particularly convenient to study the overdamped limit.

\bigskip

The results from Theorem~\ref{thm:error_second_order_schemes} allow us to obtain error estimates for the so-called Geometric Langevin Algorithms (GLA) introduced in~\cite{BO10}. Recall the somewhat surprising result that the error in the invariant measure of the GLA schemes is of order $\Delta t^p$ for a discretization of order~$p$ of the Hamiltonian part, even though the weak and strong orders of the scheme are only one. The following result complements the estimate given in~\cite{BO10} by making precise the leading order corrections to the invariant measure of the numerical scheme with respect to the canonical measure (see the proof in Section~\ref{sec:proof_cor:error_GLA}).

\begin{corollary}[Error estimates for GLA schemes]
\label{cor:error_GLA}
Consider one of the GLA schemes defined in~\eqref{eq:GLA_schemes}, and denote by $\mu_{\gamma,\dt}(\d{q} \, \d{p})$ its invariant measure. Then there exist functions $f_{2,\gamma}, f_{3,\gamma} \in \widetilde{\mathcal{S}}$ such that, for any function $\psi \in \mathcal{S}$,
\begin{equation}
\label{eq:error_GLA_general}
\begin{aligned}
\int_\cE \psi(q,p) \, \mu_{\gamma,\dt}(\d{q} \, \d{p}) & = \int_\cE \psi(q,p) \, \mu(\d{q} \, \d{p}) + \dt^2 \int_\cE \psi(q,p) f_{2,\gamma}(q,p) \, \mu(\d{q} \, \d{p}) \\
& \quad + \dt^3 \int_\cE \psi(q,p) f_{3,\gamma}(q,p) \, \mu(\d{q} \, \d{p}) + \dt^4 r_{\psi,\gamma,\dt},
\end{aligned}
\end{equation}
where the remainder $r_{\psi,\gamma,\dt}$ is uniformly bounded for $\dt$ sufficiently small. The expressions of the correction functions $f_{2,\gamma}$ and $f_{3,\gamma}$ are
\[
\begin{aligned}
f_2^{\gamma C, A,B,A} = f_2^{\gamma C, A,B,A, \gamma C}, & \qquad f_3^{\gamma C, A,B,A} = -\frac{\gamma}{2} Cf_2^{\gamma C, A,B,A}, \\
f_2^{\gamma C, B,A,B} = f_2^{\gamma C, B,A,B \gamma C}, & \qquad f_3^{\gamma C, B,A,B} = -\frac{\gamma}{2} Cf_2^{\gamma C, B,A,B}.
\end{aligned}
\]
\end{corollary}

Note that the leading order term of the error is the same as for the corresponding second order splitting schemes. The next order correction (of order $\dt^3$) vanishes for functions $\psi$ depending only on the position variable~$q$. 

\begin{remark}[Hamiltonian limit of the correction functions $f_{2,\gamma}$]
\label{rmk:Hamiltonian_limit_second_order}
Proving a result similar to Proposition~\ref{prop:Ham_limit_correction} for second order splitting schemes or GLA schemes turns out to be much more difficult, although we formally expect that the limit of $f_{2,\gamma}$ as $\gamma \to 0$ is the first order correction of the modified Hamiltonian constructed by backward analysis. From~\eqref{eq:correction_second_order_schemes}, it should indeed be the case that $f_{2}^{\gamma C,B,A,B,\gamma C}$ converges to 
\[
f_2^{B,A,B} = -\frac{1}{12} \left(A+ \frac{B}{2}\right)g.
\]
Moreover, as we already mentioned before Proposition~\ref{prop:Ham_limit_correction}, we are not able to uniformly control remainder terms in the error expansion~\eqref{eq:error_second_order_schemes} as $\gamma\to 0$.
\end{remark}

\subsection{Numerical estimation of the correction term}
\label{sec:num_estimation_correction}

The results of Section~\ref{sec:error_estimates_finite_friction} show that the leading order correction terms for the average of an observable~$\psi \in \mathcal{S}$ can be written as 
\begin{equation}
\label{eq:leading_correction}
\int_\cE \psi(q,p) f_{\gamma}(q,p) \, \mu(\d{q} \, \d{p}),
\end{equation}
where the function $f_\gamma \in \widetilde{\mathcal{S}}$ is the solution of a Poisson equation
\begin{equation}
\label{eq:Poisson_leading_correction}
\Lgam^* f_{\gamma} = g_\gamma,
\end{equation}
the function $g_\gamma \in \widetilde{\mathcal{S}}$ depending on the numerical scheme at hand (the fact that $f_\gamma \in \widetilde{\mathcal{S}}$ is a consequence of Theorem~\ref{thm:stability_S}). It is in general impossible to analytically solve~\eqref{eq:Poisson_leading_correction}, and very difficult to numerically approximate its solution since it is a  high-dimensional partial differential equation. It is however possible to rewrite~\eqref{eq:leading_correction} as an integrated correlation function, a quantity which is amenable to numerical approximation. This is a standard way of computing transport coefficients based on Green-Kubo formulae, see the summary provided in Section~\ref{sec:def_transport_coeff}. It provides here a way to compute the first order correction in the perfect sampling bias with a single simulation (as an alternative to Romberg extrapolation, which requires at least two simulations at different timesteps, see~\cite{TT90}). 

\subsubsection{Error estimates}

The approach we follow is based on the following operator identity (which makes sense in~$\cH^1$ for instance, in view of~\eqref{eq:semigroup_estimates_H1mu})
\[
\Lgam^{-1} = -\int_0^{+\infty} \rme^{t \Lgam} \, \d{t}.
\]
Since 
\[
\int_\cE \left( \rme^{t \Lgam} \psi \right) \, g_\gamma \, \d{\mu} = \mathbb{E} \Big( \psi(q_t,p_t) g_\gamma(q_0,p_0) \Big),
\]
where the expectation is taken over all initial conditions $(q_0,p_0)$ distributed according to~$\mu$  and over all realizations of equilibrium Langevin dynamics~\eqref{eq:Langevin}, the leading order correction term~\eqref{eq:leading_correction} can be rewritten as
\begin{equation}
\label{eq:rewriting_correction}
\int_\cE \psi(q,p) f_{\gamma}(q,p) \, \mu(\d{q} \, \d{p}) = -\int_0^{+\infty} \mathbb{E} \Big( \psi(q_t,p_t) g_\gamma(q_0,p_0) \Big) \d{t}.
\end{equation}
The following result (proved in Section~\ref{sec:proof_approx_GK_formula}) shows how to approximate quantities such as~\eqref{eq:rewriting_correction} up to errors $\mathrm{O}(\dt^\alpha)$, when the invariant measure of the numerical scheme is correct to terms of order~$\mathrm{O}(\dt^\alpha)$ (as discussed in Section~\ref{sec:error_estimates_finite_friction}). The fundamental ingredient is the replacement of the observable $\psi$ by some modified observable, in the spirit of backward analysis. Let us emphasize that we do not require the numerical scheme to be of weak or strong order~$p$ in itself. For instance, GLA schemes are only first order correct on trajectories (as proved in~\cite{BO10}), but nonetheless may have invariant measures which are very close to~$\mu$. To somewhat simplify the notation and state our result in a more general fashion since it can be used in other contexts than Langevin dynamics (see~\cite{FHS14} for an application to Metropolis-Hastings schemes), we do not denote explicitly all the dependencies on~$\gamma$ although the reader should keep them in mind. 

\begin{theorem}
\label{thm:approx_GK_formula}
Consider a numerical method with an invariant measure $\mu_\dt$ having bounded moments at all orders (i.e.~\eqref{eq:moment_estimate} is satisfied) and such that, for $\psi \in \mathcal{S}$, 
\begin{equation}
\label{eq:asssumption_order_method}
  \int_\cE \psi \, \d{\mu}_\dt = \int_\cE \psi \, \d{\mu} + \dt^\alpha r_{\psi,\dt},
\end{equation}
where the remainder $r_{\psi,\dt}$ is uniformly bounded for $\dt$ small enough. Suppose in addition that its evolution operator $P_\dt$ is such that, for any $\psi \in \mathcal{S}$, 
\begin{equation}
\label{eq:expansion_I_Pdt}
-\frac{\Id-P_\dt}{\dt} \psi = \cL_\gamma\psi + \dt S_1\psi + \dots + \dt^{\alpha-1} S_{\alpha-1}\psi + \dt^\alpha \widetilde{R}_{\alpha,\dt}\psi, 
\end{equation}
where $S_1\psi,\dots,S_{\alpha-1}\psi,\widetilde{R}_{\alpha,\dt}\psi \in \mathcal{S}$ and there exists $s>0$ such that the remainder $\widetilde{R}_{\alpha,\dt}\psi$ is uniformly bounded in $L^\infty_{\Li_s}$ for~$\dt$ sufficiently small. Assume finally that $P_\dt$ satisfies the uniform ergodicity condition~\eqref{eq:ergodicity_num} (hence~\eqref{eq:unif_bound_Linf_Lis} holds). Then, the integrated correlation of two observables $\psi, \varphi \in \widetilde{\mathcal{S}}$ can be approximated by a Riemann sum up to an error of order $\dt^\alpha$:
\begin{equation}
\label{eq:correction_GK_general}
\int_0^{+\infty} \mathbb{E} \Big( \psi(q_t,p_t) \varphi(q_0,p_0) \Big) \d{t} = \dt \sum_{n=0}^{+\infty} \mathbb{E}_\dt \left(\widetilde{\psi}_{\dt,\alpha}\left(q^{n},p^{n}\right)\varphi\left(q^0,p^0\right)\right) + \dt^\alpha r^{\psi,\varphi}_\dt, 
\end{equation}
where $r_\dt^{\psi,\varphi}$ is uniformly bounded for $\dt$ sufficiently small, the expectation $\mathbb{E}_\dt$ is over all initial conditions $(q_0,p_0)$ distributed according to~$\mu_\dt$ and over all realizations of the Markov chain induced by $P_\dt$, and the modified observable $\widetilde{\psi}_{\dt,\alpha} \in \mathcal{S}$ reads 
\[
\widetilde{\psi}_{\dt,\alpha} = \psi_{\dt,\alpha} - \int_\cE \psi_{\dt,\alpha} \, \d{\mu}_\dt, \qquad
\psi_{\dt,\alpha} = \left(\Id + \dt \,S_1 \Lgam^{-1} + \dots + \dt^{\alpha-1} S_{\alpha-1}\Lgam^{-1} \right)\psi.
\]
\end{theorem}

The assumptions of this theorem are satisfied for the splitting schemes considered in this article (see the comment after~\eqref{eq:P_dt_eq_order1} for the boundedness of the remainder $\widetilde{R}_{\alpha,\dt}\psi$). 

In the particular case $\alpha=2$, which is in fact the most relevant one from a practical viewpoint, it is possible not to modify the observable~$\psi$ when the discrete generator is correct at order~$2$ (see~\eqref{eq:expansion_I_Pdt_order_2} below for a precise statement), upon considering a time discretization of the integral which leads to errors of order~$\dt^2$, for instance a trapezoidal rule. The following result is obtained by an appropriate application of Theorem~\ref{thm:approx_GK_formula} (see Section~\ref{sec:proof_approx_GK_formula} for the proof).

\begin{corollary}[Trapezoidal rule for second order schemes]
\label{cor:GK_trapezoidal}
Consider a numerical scheme satisfying the assumptions of Theorem~\ref{thm:approx_GK_formula}, and whose discrete generator is in addition correct at order~2: for any $\psi \in \mathcal{S}$,
\begin{equation}
\label{eq:expansion_I_Pdt_order_2}
-\frac{\Id-P_\dt}{\dt}\psi = \cL_\gamma\psi + \frac\dt2 \cL_\gamma^2\psi + \dt^2 \widetilde{R}_{\dt}\psi.
\end{equation}
Then, for two observables $\varphi,\psi \in \widetilde{\mathcal{S}}$, 
\begin{equation}
\label{eq:correction_GK_second}
\begin{aligned}
& \int_0^{+\infty} \mathbb{E} \Big( \psi(q_t,p_t) \varphi(q_0,p_0) \Big) \d{t} \\
& \qquad \qquad = \frac\dt2 \mathbb{E}_\dt \Big(\psi_{\dt,0}\left(q^0,p^0\right) \varphi\left(q^0,p^0\right)\Big) + \dt \sum_{n=1}^{+\infty} \mathbb{E}_\dt \Big(\psi_{\dt,0}\left(q^{n},p^{n}\right)\varphi\left(q^0,p^0\right)\Big) + \dt^2 r^{\psi,\varphi}_\dt, 
\end{aligned}
\end{equation}
where $r_\dt^{\psi,\varphi}$ is bounded for $\dt$ sufficiently small and
\[
\psi_{\dt,0} = \psi - \int_\cE \psi \, \d{\mu}_\dt.
\]
\end{corollary}

\subsubsection{Numerical approximation}

There are two principal ways to estimate the expectations in~\eqref{eq:correction_GK_general} or~\eqref{eq:correction_GK_second}, using either several independent realizations of the nonequilibrium dynamics or a single, long trajectory, see for instance the discussion in Section~13.4 of~\cite{Tuckerman}. When $K$ independent realizations $(q^{n,k},p^{n,k})$ are generated for $N_{\rm iter}$ timesteps each, starting from initial conditions distributed according to $\mu_\dt$, the expectation in~\eqref{eq:correction_GK_general} may be approximated using empirical averages of the correlation functions as
\[
\frac{\dt}{K}  \sum_{k=1}^K \sum_{n=0}^{N_{\rm iter}} \left[ \psi_{\dt,\alpha}\left(q^{n,k},p^{n,k}\right) - \Psi_{\dt,\alpha}^{K,N_{\rm iter}} \right]\varphi\left(q^{0,k},p^{0,k}\right), 
\]
where $\alpha = 1$ and $\psi_{\dt,1} = \psi$ for first order splittings; while $\alpha = 2$ and $\psi_{\dt,2} = (1+\dt\cL_\gamma/2)\psi$ for second order ones since $S_1 = \cL_\gamma^2/2$ for the schemes presented in Section~\ref{sec:splitting_schemes_2nd} (see for instance~\eqref{eq:I_Pdt_eq_order2}). The empirical average $\Psi_{\dt,p}^{M,N_{\rm iter}}$ reads
\[
\Psi_{\dt,\alpha}^{M,N_{\rm iter}} = \frac{1}{K (1+N_{\rm iter})} \sum_{k=1}^K \sum_{n=0}^{N_{\rm iter}}
\psi_{\dt,\alpha}\left(q^{n,k},p^{n,k}\right).
\] 
This formula highlights the other errors arising from the discretization: (i) a statistical error related to the finiteness of $K$ and to the fact that initial conditions are obtained in practice by subsampling a single, long trajectory; (ii) a truncation error related to the finiteness of $N_{\rm iter}$.

\subsubsection{Numerical illustration}
\label{sec:numerics}

We illustrate the convergence results~\eqref{eq:correction_GK_general} and~\eqref{eq:correction_GK_second} for a simple two-dimensional system. We denote $q=(x,y) \in \mathcal{M} = (2\pi \mathbb{T})^2$, and consider the potential
\[
V(q) = 2 \cos(2x) + \cos(y).
\]
The inverse temperature is fixed to $\beta = 1$ and we consider a trivial mass matrix $M = \mathrm{Id}$ with unit friction $\gamma=1$. Trajectory data is taken from $10^3$ independent runs of fixed time interval $2\times 10^8$, with the aim to compute the integral of the velocity autocorrelation function, which corresponds to $\psi(q,p) = \varphi(q,p) = M^{-1}p$ in~\eqref{eq:correction_GK_general}. Using the second order $P_\dt^{\gamma C, B,A,B, \gamma C}$ scheme, applying the appropriate correction function~\eqref{eq:correction_GK_second} gives the predicted order~$\dt^2$ result, while the standard Riemann approximation has errors of order~$\dt$. In the numerical results in Figure \ref{fig:corr} the corrected approximation gives marginally better results than the trapezoidal rule (though of the same order) due to additional higher order terms being included.

Let us now numerically confirm the error estimates~\eqref{eq:error_first_order_schemes}-\eqref{eq:error_second_order_schemes}-\eqref{eq:error_GLA_general}. More precisely, we show that, provided the leading correction term~\eqref{eq:leading_correction} is estimated by discretizing~\eqref{eq:rewriting_correction} using~\eqref{eq:correction_GK_second} and subtracted from the estimated result, canonical averages are estimated up to errors of order~$\dt^4$ for second order splittings instead of~$\dt^2$ without the correction. We use the same trajectory data as above to approximate the canonical average of the total system energy~$H$. We test the effectiveness of the correction both in practice and principle, by computing the observed average and correction term in the same simulation in the former case, while computing a more accurate correction term independently in the latter case (using a smaller timestep~$\dt=0.1$). The results are shown in the right panel of Figure \ref{fig:corr}. 

\begin{figure}
\begin{center}
\includegraphics[width=\textwidth]{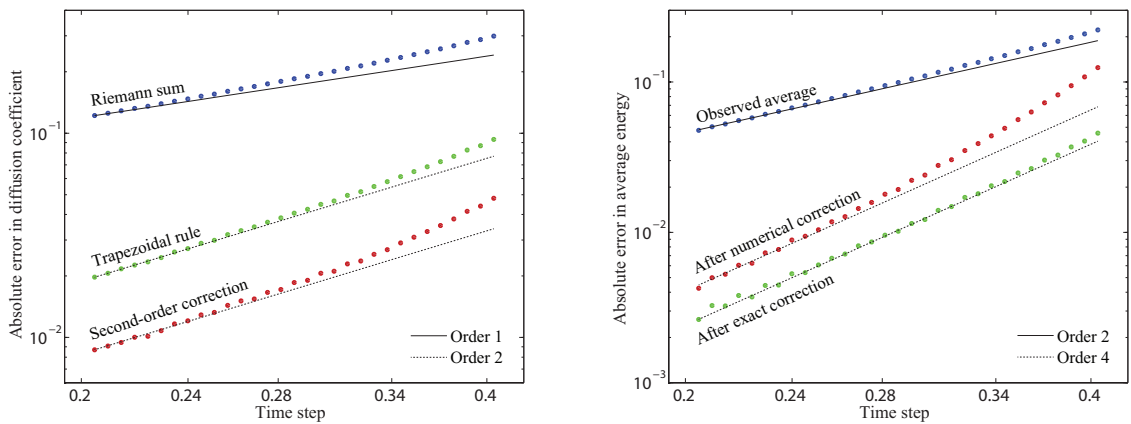}
\end{center}
\caption{\label{fig:corr} 
Left: The error in the value of the integrated velocity autocorrelation function is compared at a number of timesteps when computed using a Riemann sum or the correction term provided in \eqref{eq:correction_GK_general}. The result from computing the integral using the trapezoidal rule is also shown. Right: The error in the computed average of total energy is plotted, with the correction term computed using the same stepsize demonstrating the practical application of the method. We can test the validity of \eqref{eq:error_second_order_schemes} in principle by computing the correction more accurately at a smaller timestep in a separate simulation, this result is labelled as the `exact correction'.  All results are computed using the scheme associated with $P_\dt^{\gamma C, B,A,B, \gamma C}$ with $\beta=\gamma = 1$. } 
\end{figure}

\subsection{Overdamped limit}
\label{sec:ovd_limit}

We study in this section the overdamped limit $\gamma \to +\infty$, assuming that the mass matrix is $M = \Id$. We first study the consistency of the invariant measures of limiting numerical schemes in Section~\ref{sec:ovd_schemes}, before stating precise convergence results for second order splitting schemes in Section~\ref{sec:rigorous_estimates_ovd}.Ultimately, we relate in Section~\ref{sec:ovd_limit_correction} the overdamped limit of the correction terms obtained for finite $\gamma$ to the correction obtained by directly studying the overdamped limit.

\subsubsection{Overdamped limits of splitting schemes}
\label{sec:ovd_schemes}

The only part of the numerical schemes where the friction enters is the Ornstein-Uhlenbeck process on momenta. The limit $\gamma \to +\infty$ for $\dt > 0$ fixed amounts to resampling momenta according to the Gaussian distribution~$\kappa(\d{p})$ at all timesteps. For instance, the numerical scheme associated with the evolution operator $P_\dt^{\gamma C, B,A,B,\gamma C}$ reduces to
\[
q^{n+1} = q^n - \frac{\dt^2}{2} \nabla V(q^n) + \frac{\dt}{\sqrt{\beta}} \, G^n,
\]
where $(G^n)$ are independent and identically distributed Gaussian random vectors with identity covariance. This is indeed a consistent discretization of the overdamped process~\eqref{eq:overdamped} with an effective timestep $h = \dt^2/2$, and the invariant measure of this numerical scheme is close to~$\overline{\mu}$. Other schemes may have non-trivial large friction limits and invariant measures close to~$\overline{\mu}$. This is the case for the scheme associated with the evolution operator $P_\dt^{B,A,\gamma C,A,B}$, for which the limiting discrete dynamics reads (see~\cite{LM12})
\begin{align*}
q^1 &= q^0- \frac{\dt^2}{4} \nabla V(q^0) + \frac{\dt}{2\sqrt{\beta}} (G^0+G^{1}),\\
q^{n+1} &= q^n - \frac{\dt^2}{2} \nabla V(q^n) + \frac{\dt}{2\sqrt{\beta}} (G^n+G^{n+1}), \quad \textrm{for } n>0.
\end{align*}
Note that $(q^n)$ is not a Markov chain due to the correlations in the random noises. 

On the other hand, the limits of the invariant measures associated with certain schemes are not consistent with the canonical measure~$\overline{\mu}$. This is the case for the first-order schemes, as well as the second order splittings listed in item~(iii) in Section~\ref{sec:splitting_schemes_2nd}. For instance, the limit of the scheme associated with $P_\dt^{\gamma C,A,B}$ reads
\[
q^{n+1} = q^n + \frac{\dt}{\sqrt{\beta}} \, G^n.
\]
The invariant measure of this Markov chain is the uniform measure on $\mathcal{M}$, and is therefore very different from the invariant measure~$\overline{\mu}$ of the continuous dynamics~\eqref{eq:overdamped} (it amounts to setting $V = 0$). As another example, consider the limit of the scheme associated with $P_\dt^{\gamma C,B,A}$:
\[
q^{n+1} = q^n - \dt^2 \nabla V(q^n) + \frac{\dt}{\sqrt{\beta}} \, G^n.
\]
This is the Euler-Maruyama discretization of~\eqref{eq:overdamped} with an effective timestep $h=\dt^2$ but an inverse temperature $2\beta$ rather than $\beta$.

\subsubsection{Rigorous error estimates}
\label{sec:rigorous_estimates_ovd}

The following result quantifies the errors of the invariant measure of second order splitting schemes of Langevin dynamics, for large values of $\gamma$. We restrict ourselves to the second order splittings where the Ornstein-Uhlenbeck part is either at the ends or in the middle (categories~(i) and~(ii) in Section~\ref{sec:splitting_schemes_2nd}). From a technical viewpoint, we are able here to bound remainder terms uniformly in $\gamma$ by relying on the properties of the limiting operator $\Lovd^{-1}$. The result we obtain is the following (see Section~\ref{sec:proof_thm:ovd_limit} for the proof).

\begin{theorem}
\label{thm:ovd_limit}
Consider any of the second order splittings presented in Section~\ref{sec:splitting_schemes_2nd}, denote by $\mu_{\gamma,\dt}(\d{q} \, \d{p})$ its invariant measure, and by $\overline{\mu}_{\gamma,\dt}(\d{q})$ its marginal in the position variable. Then there exists a function $f_{2,\infty} = f_{2,\infty}(q) \in C^\infty(\mathcal{M})$, with average zero with respect to~$\overline{\mu}$, such that, for any smooth $\psi = \psi(q) \in C^\infty(\mathcal{M})$ and $\gamma \geq 1$, 
\[
\int_\mathcal{M} \psi(q) \, \overline{\mu}_{\gamma,\dt}(\d{q}) = \int_\mathcal{M} \psi \, \d{\overline{\mu}} + \dt^2 \int_\mathcal{M} \psi \, f_{2,\infty}\, \d{\overline{\mu}} + r_{\psi,\gamma,\dt},
\]
where the remainder is of order $\dt^4$ up to terms exponentially small in $\gamma \dt$. More precisely, there exist constants $a,b \geq 0$ and $\kappaK > 0$ (all depending on $\psi$) such that
\[
\left|r_{\psi,\gamma,\dt}\right| \leq a \dt^4 + b \, \rme^{-\kappaK \gamma \dt}.
\]
The expression of $f_{2,\infty}$ depends on the numerical scheme at hand:
\begin{equation}
\label{eq:correction_overdamped}
\begin{aligned}
f^{\gamma C,B,A,B,\gamma C}_{2,\infty}(q) & = \frac18 \left( -2\Delta V + \beta |\nabla V|^2 + a_{\beta,V}\right), \qquad a_{\beta,V} = \int_\mathcal{M} \Delta V \, \d{\overline{\mu}} = \beta \int_\mathcal{M} \left|\nabla V\right|^2 \, \d{\overline{\mu}},\\
f_{2,\infty}^{A,B,\gamma C, B,A}(q) & = -\frac18 \left(\Delta V - a_{\beta,V}\right), \\
f^{\gamma C,A,B,A,\gamma C}_{2,\infty}(q) & = \frac18 \left( \Delta V - \beta |\nabla V|^2\right), \\
f^{B,A,\gamma C,A,B}_{2,\infty}(q) & = 0. 
\end{aligned}
\end{equation}
\end{theorem}

The real number $a_{\beta,V}$ ensures that all functions $f_{2,\infty}$ are of average zero with respect to~$\overline{\mu}$. Two comments are in order. Note first that the result is stated for observables which depend only on the position variable~$q$ since the limiting case $\gamma \to +\infty$ corresponds to a dynamics on the positions only. There is anyway no restriction in stating the result using such observables since, as already discussed in the introduction, the error on the marginal in the position variables is the relevant error, momenta being trivial to sample exactly under the canonical measure. Secondly, let us emphasize that the $\dt^2$ correction term vanishes for the method associated with $P_\dt^{B,A,\gamma C,A,B}$ (as already noted in~\cite{LM12}). This means that the corresponding discretization of overdamped Langevin dynamics (formally obtained by setting $\gamma = +\infty$) has an invariant measure which is correct at second-order in the effective timestep $h = \dt^2/2$. 

\subsubsection{Overdamped limit of the correction terms}
\label{sec:ovd_limit_correction}

In order to relate the convergence result from Theorem~\ref{thm:ovd_limit} to the error estimates from Theorem~\ref{thm:error_second_order_schemes}, we prove that the limits of the correction functions $f_{2,\gamma}$ as $\gamma \to +\infty$ agree with the functions defined in~\eqref{eq:correction_overdamped} (see Section~\ref{sec:proof_prop:ovd_limit_correction} for the proof). This can be seen as a statement regarding the permutation of the limits $\gamma \to +\infty$ and $\dt \to 0$ for the leading correction term, namely, for a smooth function $\psi = \psi(q) \in C^\infty(\mathcal{M})$,
\[
\begin{aligned}
\lim_{\dt \to 0} \lim_{\gamma \to +\infty} \frac{1}{\dt^2} \left( \int_\mathcal{M} \psi \, \d{\overline{\mu}}_{\gamma,\dt} - \int_\mathcal{M} \psi \, \d{\overline{\mu}}\right) & = \lim_{\gamma \to +\infty} \lim_{\dt \to 0} \frac{1}{\dt^2} \left( \int_\mathcal{M} \psi \, \d{\overline{\mu}}_{\gamma,\dt} - \int_\mathcal{M} \psi \, \d{\overline{\mu}}\right) \\ 
& = \lim_{\gamma \to +\infty} \int_\mathcal{M} \psi \left(\pi f_{2,\gamma}\right) \d{\overline{\mu}}\\
& = \int_\mathcal{M} \psi \, f_{2,\infty} \, \d{\overline{\mu}}.
\end{aligned}
\]
The precise result is the following:

\begin{proposition}
\label{prop:ovd_limit_correction}
There exists a constant $K > 0$ such that, for all $\gamma \geq 1$,
\[
\begin{aligned}
\left\| f_2^{\gamma C, B,A,B,\gamma C} - \frac18 \left( -2 \Delta V + \beta |\nabla V|^2 + a_{\beta,V} \right)\right\|_{H^1(\mu)} & \leq \frac{K}{\gamma}, \\
\left\| f_2^{A,B,\gamma C, B,A} - \frac18 \left(-2\Delta V + \beta p^T (\nabla^2 V)p + a_{\beta,V} \right)\right\|_{H^1(\mu)} & \leq \frac{K}{\gamma}, \\
\left\| f_2^{\gamma C, A,B,A,\gamma C} - \frac18 \left( \Delta V - \beta |\nabla V|^2\right) \right\|_{H^1(\mu)} & \leq \frac{K}{\gamma}, \\
\left\| f_2^{B,A,\gamma C,A,B} - \frac18 \left( \Delta V - \beta p^T (\nabla^2 V)p \right) \right\|_{H^1(\mu)} & \leq \frac{K}{\gamma}, \\
\end{aligned}
\]
where the constant $a_{\beta,V}$ is defined in~\eqref{eq:correction_overdamped}.
\end{proposition}
Note that, as expected, the averages with respect to~$\kappa(\d{p})$ of the above limiting functions coincide with the functions~$f_{2,\infty}$ given in~\eqref{eq:correction_overdamped}, that is, $\pi f_{2,\gamma} = f_{2,\infty} + \mathrm{O}(\gamma^{-1})$.

Let us also mention that the overdamped limit of the correction function $f_{1,\gamma}$ for first order splittings is not well defined. This is not surprising since the invariant measures of the corresponding numerical schemes are not consistent with~$\overline{\mu}$, as discussed in Section~\ref{sec:ovd_schemes}. For instance, combining~\eqref{eq:divergent_behavior_Lgamma} and the expressions of the correction functions~\eqref{eq:correction_first_order_schemes}, we see that there exists a constant $K>0$ such that
\begin{equation}
\label{eq:divergence_f1}
\left\| f_{1}^{\gamma C,B,A} + \frac{\gamma \beta}{2} \Lovd^{-1} \mathcal{L}_{\mathrm{ovd},M} V \right\|_{H^1(\mu)} \leq K,
\end{equation}
where the operator
\[
\mathcal{L}_{\mathrm{ovd},M} = -M^{-1} \nabla V \cdot \nabla_q + \frac1\beta M : \nabla^2,
\]
defined on $\mathcal{S}$, is the generator of the overdamped Langevin dynamics with non-trivial mass matrix:
\[
\d{q}_t = -M^{-1} \nabla V(q_t)\,\d{t} + \sqrt{\frac{2}{\beta}} M^{-1/2} \, \d{W}_t.
\]
Note that, when $M = \Id$, the solution can in fact be analytically computed as $f_{1}^{\gamma C,B,A} = -\beta ( \gamma V + p^T \nabla V)/2$. In any case, \eqref{eq:divergence_f1} shows that $f_{1}^{\gamma C,B,A}$ diverges as $\gamma \to +\infty$.

\section{Nonequilibrium dynamics and the computation of transport coefficients}
\label{sec:noneq_systems}

We discuss in this section the numerical estimation of transport properties such as the thermal conductivity, the shear stress, etc. (see~\cite{EM08,Tuckerman} for general physical presentations of the computation of transport coefficients, and Section~3.1 of~\cite{HDR} for a mathematically oriented introduction). 

We consider the prototypical case of the estimation of the autodiffusion coefficient. In this situation, it is relevant to consider a nonequilibrium perturbation of  standard equilibrium Langevin dynamics, where some external forcing arising from a constant force~$F \in \mathbb{R}^{dN}$ is imposed on the system:
\begin{equation}
\label{eq:noneq_Langevin}
\left\{ \begin{aligned}
\d{q}_t & = M^{-1} p_t \, \d{t}, \\
\d{p}_t & = \Big( -\nabla V(q_t) + \eta F \Big) \d{t} - \gamma M^{-1} p_t \, \d{t} + \sqrt{\frac{2\gamma}{\beta}} \, \d{W}_t.
\end{aligned} \right.
\end{equation}
We denote by 
\[
\wcL = F \cdot \nabla_p
\]
the generator of the perturbation (considered as an operator on $L^2(\mu)$, with core~$\mathcal{S}$). Note that the constant force $F$ does not derive from the gradient of a smooth function defined on~$\mathcal{M}$. (It would indeed seem that this force derives from $-F^T q$, but this potential is not periodic.) Therefore, the expression of the invariant measure is unknown, but can be nonetheless obtained as an expansion in powers of~$\eta$ when the magnitude of the forcing is sufficiently small (see Section~\ref{sec:def_transport_coeff}). The effect of the force is to create a non-zero average velocity in the direction of~$F$. The magnitude of the average velocity is a property of the system under consideration. For small forcings, it is linear in~$\eta$, with a constant of proportionality called the \emph{mobility} (see the definition~\eqref{eq:def_nu_NEMD} below), related to the autodiffusion coefficient through \eqref{eq:def_nu_Einstein}.

\begin{remark}
  As shown in~\cite{JPS14}, it is possible to consider more general forcing terms $F(q)$ which do not derive from the gradient of a periodic function. A popular example is provided by shearing forces where the particles experience a force in some direction, whose intensity depends on the coordinates of the system in another direction. 
\end{remark}

We will also be interested in the overdamped limit of the nonequilibrium dynamics~\eqref{eq:noneq_Langevin}, which reads
\begin{equation}
\label{eq:noneq_ovd_Langevin}
\d{q}_t = \Big( -\nabla V(q_t) + \eta F \Big) \d{t} + \sqrt{\frac{2}{\beta}} \, \d{W}_t.
\end{equation}
The generator of this dynamics is $\Lovd + \eta \wcL_{\rm ovd}$ with $\wcL_{\rm ovd} = F \cdot \nabla_q$ (all operators being defined on the core~$\mathcal{S}$). In this case the physically relevant response turns out to be the average force $-F\cdot \nabla V$ exerted in the direction $F$.

\subsection{Definition of transport coefficients}
\label{sec:def_transport_coeff}

Following the strategy advertised in~\cite{rey-bellet} (using the kinetic energy as a Lyapunov function), it is easy to show that the dynamics~\eqref{eq:noneq_Langevin} has a unique invariant probability measure $\mu_{\gamma,\eta}(\d{q}\,\d{p})$ with a smooth density with respect to the Lebesgue measure for any value of $\eta \in \mathbb{R}$. The mobility $\nu_{F,\gamma}$ is defined as the linear response of the velocity in the direction $F$ as the magnitude of the forcing goes to~0:
\begin{equation}
\label{eq:def_nu_NEMD}
\nu_{F,\gamma} = \lim_{\eta \to 0} \frac{1}{\eta} \int_\cE F^T M^{-1} p \, \mu_{\gamma,\eta}(\d{q} \, \d{p}). 
\end{equation}
From linear response theory (see for example the presentation in~\cite[Section~3.1]{HDR}, and the short summary provided in Section~\ref{sec:proof_LRT}), it can be shown that 
\begin{equation}
\label{eq:def_nu_LRT}
\nu_{F,\gamma} = \int_\cE F^T M^{-1} p \, f_{0,1,\gamma}(q,p) \, \mu(\d{q} \, \d{p}), \qquad \Lgam^* f_{0,1,\gamma} = -\wcL^* \mathbf{1} = -\beta F^T M^{-1}p.
\end{equation}
The mobility can therefore be rewritten as the integrated autocorrelation function of the velocity in the direction~$F$:
\begin{equation}
\label{eq:def_nu_GK}
\nu_{F,\gamma} = \beta \int_0^{+\infty} \mathbb{E}\Big[\big(F^T M^{-1}p_t\big)\big(F^T M^{-1}p_0\big)\Big] \d{t},
\end{equation}
where the expectation is over all initial conditions $(q_0,p_0)$ distributed according to~$\mu$ and over all realizations of the equilibrium Langevin dynamics~\eqref{eq:Langevin}. From this relation, it is easily seen that the mobility in the direction~$F$ is related to the autodiffusion coefficient
\begin{equation}
\label{eq:def_nu_Einstein}
D_{F,\gamma} = \lim_{t \to +\infty} \frac{\mathbb{E}\Big[\big(F\cdot(q_t-q_0)\big)^2\Big]}{2t}
\end{equation}
as
\[
\nu_{F,\gamma} = \beta D_{F,\gamma}.
\]
In practice, the two most popular ways of estimating a transport coefficient rely on the Green-Kubo formula~\eqref{eq:def_nu_GK} and the linear response of nonequilibrium dynamics in their steady-states~\eqref{eq:def_nu_NEMD}. Since the error estimates for Green-Kubo type formulas have already been discussed in Theorem~\ref{thm:approx_GK_formula}, we will restrict ourselves in the sequel to the analysis of the numerical errors introduced by nonequilibrium methods.

\subsubsection{Overdamped limit}

The overdamped limit of the mobility $\nu_{F,\gamma}$ is studied in~\cite{HP08}, where the authors  consider the autodiffusion coefficient $D_{F,\gamma}$. First, it is easily shown that the overdamped dynamics~\eqref{eq:noneq_ovd_Langevin} admits a unique invariant probability measure, which we denote by $\overline{\mu}_\eta(\d{q})$. The mobility for the overdamped dynamics~\eqref{eq:noneq_ovd_Langevin} is defined from the linear response of the projected force $-F \cdot \nabla V$ as
\begin{equation}
\label{eq:ovd_mobility}
\overline{\nu}_F = \lim_{\eta \to 0} \frac1\eta \int_\mathcal{M} -F^T \nabla V(q) \, \overline{\mu}_\eta(\d{q}) 
= \beta \int_\mathcal{M} F^T \nabla V(q) \Lovd^{-1} \left(F^T \nabla V(q)\right) \overline{\mu}(\d{q}).
\end{equation}
The derivation of this formula is very similar to that leading to~\eqref{eq:def_nu_NEMD}. The following result summarizes the limiting behavior of the mobility as the friction increases (recall that we set mass matrices to identity when studying overdamped limits). 

\begin{lemma}
\label{lem:ovd_mobility}
There exists $K > 0$ such that, for any $\gamma \geq 1$,
\[
\left|\gamma \nu_{F,\gamma} - \overline{\nu}_F - |F|^2 \right| \leq \frac{K}{\gamma}.
\]
\end{lemma}

This result is already contained in~\cite{HP08}, but we nonetheless provide a short alternative proof in Section~\ref{sec:proof_lem:ovd_mobility} (see Remark~\ref{rmk:ovd_Einstein} for a more precise comparison of the results). It shows that, in the overdamped regime $\gamma \to +\infty$, 
\begin{equation}
\label{eq:computatio_nu_F_gamma}
\nu_{F,\gamma} = \frac{|F|^2 + \overline{\nu}_F}{\gamma} + \mathrm{O}\left(\frac{1}{\gamma^2}\right),
\end{equation}
which suggests to estimate $\nu_{F,\gamma}$ using the linear response of $F^T \nabla V$ for large frictions since this quantity is expected to be a good approximation of $\overline{\nu}_F$ -- instead of relying on the standard linear response result~\eqref{eq:def_nu_NEMD}, for which the response is of order $1/\gamma$ and is hence difficult to reliably estimate. Error estimates on the numerical approximation are deduced from~\eqref{eq:estimate_overline_nu_F} below.

\subsection{Numerical schemes for the nonequilibrium Langevin dynamics}

We present in this section numerical schemes approximating solutions of~\eqref{eq:noneq_Langevin}. These schemes reduce to the schemes presented in Section~\ref{sec:splitting_schemes} when $\eta = 0$. Since the aim is to decompose the evolution generated by $\Lgam + \eta \wcL$ into analytically integrable parts, there are two principal options: either replace $B$ by 
\[
B_\eta = B+\eta\wcL,
\]
or replace $C$ by $C+\eta\wcL$. However, the schemes built on the latter option do not perform correctly in the overdamped limit since their invariant measures are not consistent with the invariant measures of nonequilibrium overdamped Langevin dynamics~\eqref{eq:noneq_ovd_Langevin}. More precisely, consider for instance the first order scheme generated by $P_{\dt}^{A,B,\gamma C + \eta \wcL} = \rme^{\Delta t \, A} \rme^{\Delta t \, B} \rme^{\Delta t (\gamma C + \eta \wcL)}$ in the case when $M = \Id$:
\[
\left\{ \begin{aligned}
q^{n+1} & = q^n + \dt \, p^n, \\
\widetilde{p}^{n+1} & = p^n - \dt \, \nabla V(q^{n+1}), \\
p^{n+1} & = \alpha_\dt \widetilde{p}^{n+1} + \frac{1-\alpha_\dt}{\gamma} \, \eta F + \sqrt{\frac{1-\alpha^2_\dt}{\beta}} \, G^n,
\end{aligned} \right.
\]
where $\alpha_\dt$ is defined after~\eqref{eq:Langevin_splitting}, and $(G^n)$ is a sequence of independent and identically distributed Gaussian random vectors with identity covariance. As $\gamma \to +\infty$, a standard Euler-Maruyama discretization of the equilibrium overdamped Langevin dynamics (\textit{i.e.} $\eta = 0$) is obtained, whereas we would like to obtain a consistent discretization of nonequilibrium overdamped Langevin dynamics~\eqref{eq:noneq_ovd_Langevin}. We therefore instead consider schemes obtained by replacing $B$ with $B + \eta \wcL$, such as the first order splitting
\[
P_{\dt}^{A,B+ \eta \wcL,\gamma C} = \rme^{\Delta t \, A} \rme^{\Delta t (B+ \eta \wcL)} \rme^{\gamma \Delta t \, C},
\]
or the second order splitting 
\[
P_{\dt}^{\gamma C,B+ \eta \wcL,A,B+ \eta \wcL,C} = \rme^{\gamma \Delta t \, C/2} \rme^{\Delta t (B+ \eta \wcL)/2} \rme^{\Delta t \, A} \rme^{\Delta t (B+ \eta \wcL)/2} \rme^{\gamma \Delta t \, C/2}.
\]
The numerical scheme associated with the first order splitting scheme $P_{\dt}^{A,B+ \eta \wcL,\gamma C}$
\[
\left\{ \begin{aligned}
q^{n+1} & = q^n + \dt \, p^n, \\
\widetilde{p}^{n+1} & = p^n + \dt \Big(-\nabla V(q^{n+1}) + \eta F \Big), \\
p^{n+1} & = \alpha_\dt \widetilde{p}^{n+1} + \sqrt{\frac{1-\alpha^2_\dt}{\beta}} \, G^n,
\end{aligned} \right.
\]
indeed is, in the limit as $\gamma \to +\infty$, a consistent discretization of the nonequilibrium Langevin dynamics~\eqref{eq:noneq_ovd_Langevin}, and its invariant measure turns out to converge to the invariant measure of~\eqref{eq:noneq_ovd_Langevin} in the limit $\dt \to 0$.

Following the method of proof of Proposition~\ref{prop:ergodicity_MC}, it can be shown that there exists a unique invariant measure $\mu_{\gamma,\eta,\dt}$ for the corresponding Markov chain. The crucial point is that the gradient structure of the force term is never used explicitly in the proofs since we  rely solely on the boundedness of the force, so that we are able to obtain convergence results and moment estimates that are independent of the magnitude~$\eta$ of the forcing term provided $\eta$ is in a bounded subset of~$\mathbb{R}$. We denote below by $P_{\gamma,\eta,\dt}$ the evolution operator associated with the numerical schemes.

\begin{proposition}[Ergodicity of numerical schemes for nonequilibrium systems]
\label{prop:ergodicity_MC_noneq}
Fix $s^* \geq 1$ and $\eta^* > 0$. For any $0 < \gamma < +\infty$, there exists $\dt^*$ such that, for any $0 < \dt \leq \dt^*$ and $\eta \in [-\eta^*,\eta^*]$, the Markov chain associated with $P_{\gamma,\eta,\dt}$ has a unique invariant probability measure $\mu_{\gamma,\eta,\dt}$, which admits a density with respect to the Lebesgue measure $\d{q} \, \d{p}$, and has finite moments: There exists $R > 0$ such that, for any $1 \leq s \leq s^*$,
\[
\int_\cE \Li_s \, \d{\mu}_{\gamma,\eta,\dt} \leq R < +\infty,
\]
uniformly in the timestep $\dt$ and the forcing magnitude~$\eta$. There also exist $\lambda, K > 0$ (depending on $s^*$, $\gamma$ and $\eta^*$ but not on~$\dt$) such that, for all functions $f \in L^\infty_{\Li_s}$, the following holds for almost all $(q,p) \in \cE$:
\[
\forall n \in \mathbb{N}, \qquad \left| \left(P_{\gamma,\eta,\dt}^n f\right)(q,p) - \int_\cE f \d{\mu}_{\gamma,\eta,\dt} \right| \leq K \, \Li_s(q,p) \, \rme^{-\lambda n\dt} \, \| f \|_{L^\infty_{\Li_s}}.
\]
\end{proposition}

Let us emphasize that we do not have any control on the convergence rate~$\lambda$ in terms of~$\eta^*$, and it could well be that $\lambda$ goes to~0 as $\eta^*$ increases.

\subsection{Error estimates on transport coefficients from nonequilibrium methods}
\label{sec:error_transport}

The following result provides error estimates for the invariant measure of the first order or second order splittings schemes of Section~\ref{sec:splitting_schemes_2nd} when $B$ is replaced by $B_\eta$.

\begin{theorem}
\label{thm:error_noneq}
Denote by $p$ the order of the splitting scheme, by $f_{\alpha,0,\gamma}$ the leading order correction function in the case $\eta = 0$ as given by Theorem~\ref{thm:error_first_order_schemes} for $\alpha=1$ and by Theorem~\ref{thm:error_second_order_schemes} for $\alpha=2$. Then, there exists a function $f_{\alpha,1,\gamma} \in \widetilde{\mathcal{S}}$ such that, for any smooth function $\psi \in \mathcal{S}$, there exist $\dt^*,\eta^* > 0$ and a constant $K > 0$ for which, for all $\eta \in [-\eta^*,\eta^*]$ and $0 < \dt \leq \dt^*$,
\[
\int_\cE \psi \, \d{\mu}_{\gamma,\eta,\dt} = \int_\cE \psi \Big(1+ \eta f_{0,1,\gamma} + \dt^\alpha f_{\alpha,0,\gamma} + \eta \dt^\alpha f_{\alpha,1,\gamma} \Big) \d{\mu} + r_{\psi,\gamma,\eta,\dt},
\]
where $f_{0,1,\gamma}$ is defined in~\eqref{eq:def_nu_LRT}, and
\[
\left|r_{\psi,\gamma,\eta,\dt}\right| \leq K(\eta^2 + \dt^{\alpha+1}), 
\qquad 
\left|r_{\psi,\gamma,\eta,\dt} - r_{\psi,\gamma,0,\dt}\right| \leq K \eta (\eta + \dt^{\alpha+1}).
\]
\end{theorem}

The proof of this result can be found in Section~\ref{sec:proof_noneq}. Note that the remainder term now collects higher order terms both as powers of the timestep and the nonequilibrium parameter~$\eta$. The estimates we obtain on the remainder are however compatible with taking the linear response limit, as made precise by the following error estimate on the transport coefficient (which is an immediate consequence of Theorem~\ref{thm:error_noneq}). In order to state the result, we introduce the reference linear response for an observable $\psi$
\[
\mathscr{D}_{\psi,\gamma,0} = \lim_{\eta \to 0} \frac{1}{\eta} \left(\int_\cE \psi \, \d{\mu}_{\gamma,\eta} - \int_\cE  \psi \, \d{\mu}_{\gamma} \right),
\]
and its numerical approximation
\[
\mathscr{D}_{\psi,\gamma,\dt} = \lim_{\eta \to 0} \frac{1}{\eta} \left(\int_\cE \psi \, \d{\mu}_{\gamma,\eta,\dt} - \int_\cE  \psi \, \d{\mu}_{\gamma,\dt} \right).
\]
It is often the case that $\psi$ has a vanishing average with respect to~$\mu$, as is the case for the function $F^T M^{-1}p$ in~\eqref{eq:def_nu_NEMD}. In general, it however has a non-zero average with respect to the invariant measure $\mu_{\gamma,\dt}$ of the numerical scheme associated with a discretization of the equilibrium dynamics.

\begin{corollary}
\label{cor:noneq}
There exist $\dt^*,\eta^* > 0$ and a constant $K > 0$ such that, for all $\eta \in [-\eta^*,\eta^*]$ and $0 < \dt \leq \dt^*$,
\[
\mathscr{D}_{\psi,\gamma,\dt} = \mathscr{D}_{\psi,\gamma,0} + \dt^\alpha \int_\cE \psi \, f_{\alpha,1,\gamma} \, \d{\mu} + \dt^{\alpha+1} r_{\psi,\gamma,\dt},
\]
where $r_{\psi,\gamma,\dt}$ is uniformly bounded.
\end{corollary}

In particular, we obtain the following estimate on the numerically computed mobility:
\begin{align}
\nu_{F,\gamma,\dt} & = \lim_{\eta \to 0} \frac{1}{\eta} \left(\int_\cE F^T M^{-1} p \, \mu_{\gamma,\eta,\dt}(\d{q}\,\d{p}) - \int_\cE  F^T M^{-1} p \, \mu_{\gamma,0,\dt}(\d{q}\,\d{p}) \right) \label{eq:LR_num_mob} \\
& = \nu_{F,\gamma} + \dt^\alpha \int_\cE  F^T M^{-1} p  \, f_{\alpha,1,\gamma} \, \d{\mu} + \dt^{\alpha+1} r_{\gamma,\dt}, \label{eq:prediction_num}
\end{align}
where the reference mobility $\nu_{F,\gamma}$ is defined in~\eqref{eq:def_nu_LRT}.

\subsubsection{Numerical illustration}

We consider the same system as in Section~\ref{sec:numerics}, with an external force $F=(1,0)$ and $K+1$ forcing strengths $\eta_k = (k-1) \Delta \eta$ uniformly spaced in the interval $[0,\eta_{\rm max}]$ with $\eta_{\rm max} = 0.5$ (so that $\Delta \eta = \eta_{\rm max}/K$). We fix the friction to $\gamma = 1$ and the inverse temperature to $\beta = 1$. We use a coupling strategy to reduce the statistical noise in the computation of the linear response~\eqref{eq:LR_num_mob}. The $K+1$ replicas of the system are started at the same position $q= (0,0)$, with the same velocity (sampled according to the canonical measure~$\mu$). Each replica experiences the force $-\nabla V + \eta_k F$ (Note that the first replica experiences the reference force $-\nabla V$ corresponding to a discretization of the equilibrium dynamics). Most importantly, the same Gaussian random numbers $G^n$ are used for all replicas to discretize the Brownian motion. Although not carefully documented here, this coupling strategy tremendously decreases the statistical error in the computed linear response. Such a coupling strategy was already proposed for exclusion processes in~\cite{GoodmanKin}. However, our experience shows that it fails for higher dimensional systems with more complex potentials (such as Lennard-Jones fluids).

For a given value of the timestep $\dt$, we denote by $(q^{k,n},p^{k,n})_{n \geq 0}$ the discrete trajectory of the $k$th replica. The linear response in the projected average velocity $\delta v_{\eta_k}$ is approximated over $N_{\rm iter}$ integration steps as
\[
\begin{aligned}
\delta v_{\eta_k} & = \int_\cE F^T M^{-1} p \, \mu_{\dt,\eta_k}(\d{q}\,\d{p}) - \int_\cE F^T M^{-1} p \, \mu_{\dt,0}(\d{q}\,\d{p}) \\
& \simeq \frac{1}{N_{\rm iter}} \sum_{n=1}^{N_{\rm iter}}F^T M^{-1} \Big(p^{k,n}-p^{1,n}\Big) = {\widehat{v}_{\eta_k}}^{\,\,N_{\rm iter}}.
\end{aligned}
\]
We then estimate the mobility by a linear fit on the first $K' = 10$ values of $\widehat{v}_{\eta_k}^{\,\,N_{\rm iter}}$ considered as a function of~$\eta_k$ (see Figure~\ref{fig:noneq}, left). The value $\nu_{F,\gamma,\dt}$ is the estimated slope in the fit. The behavior of the mobility $\nu_{F,\gamma,\dt}$ as a function of the timestep is presented in Figure~\ref{fig:noneq} (right) for the numerical schemes associated with the first order splitting $P_\dt^{A,B_\eta,\gamma C}$ and the second order splitting $P_\dt^{\gamma C, B_\eta,A,B_\eta, \gamma C}$. We used $N_{\rm iter} = 4 \times 10^{11}$ for the first order scheme, and $N_{\rm iter} = 2.5 \times 10^{11}$ for the second order one. The statistical error is very small and error bars are therefore not reported. The computed mobilities can be fitted for small~$\dt$ as
\[
\nu_{F,\gamma,\dt} \simeq 0.0740 + 0.0817\dt
\] 
for the first-order splitting and 
\[
\nu_{F,\gamma,\dt} \simeq 0.0741 + 0.197\dt^2
\]
for the second order splitting scheme, in agreement with the theoretical prediction~\eqref{eq:prediction_num}.

\begin{figure}
\begin{center}
\includegraphics[width=7cm]{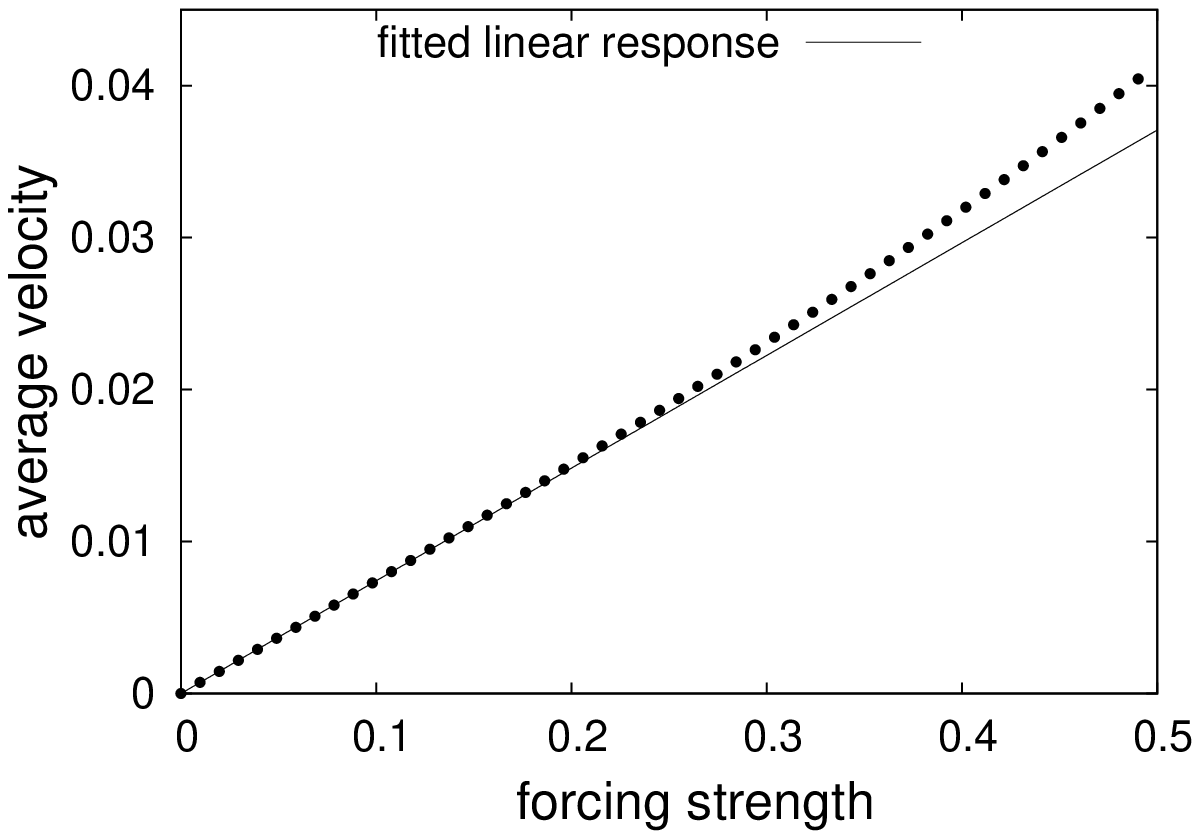}
\includegraphics[width=7cm]{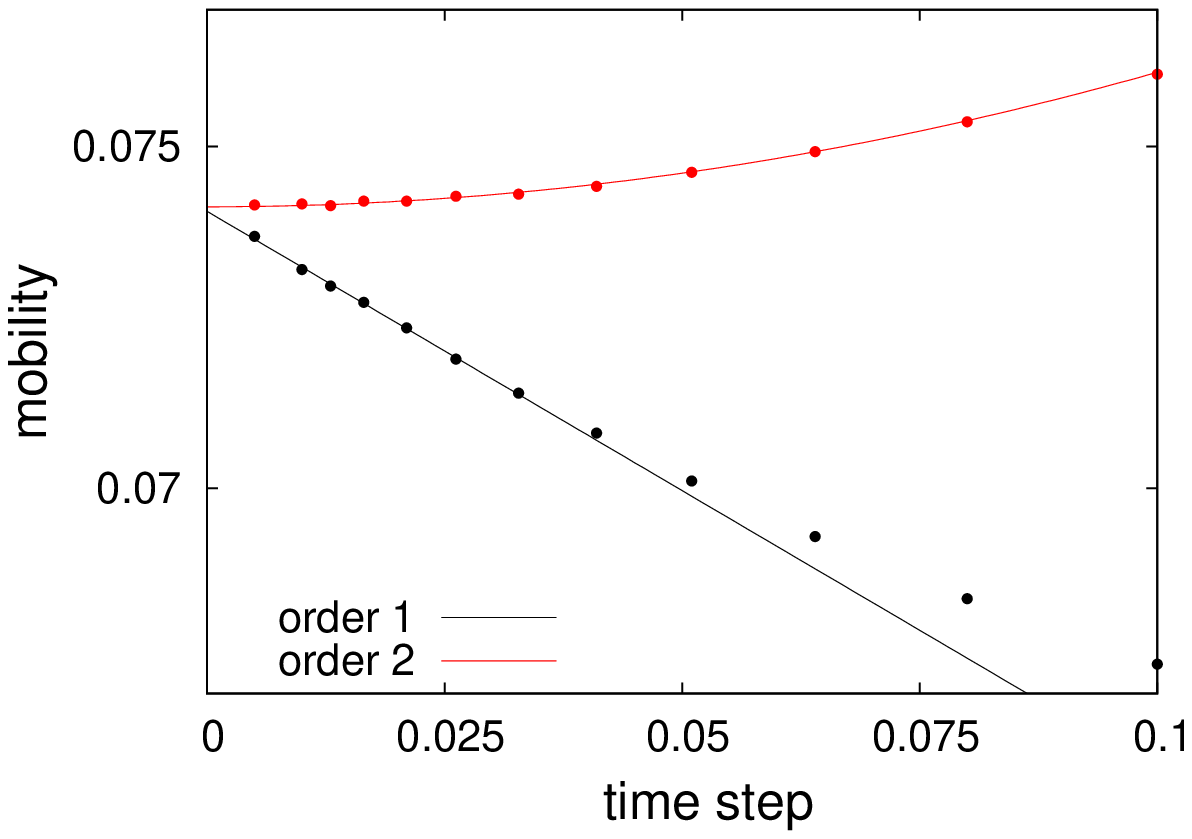}
\end{center}
\caption{\label{fig:noneq} 
Left: Linear response of the average velocity $\delta v_\eta$ as a function of $\eta$ ($K = 50$) for the scheme associated with $P_\dt^{\gamma C, B_\eta,A,B_\eta, \gamma C}$ and $\dt = 0.01, \gamma = 1$. A linear fit on the first ten values gives $\delta v_\eta \simeq 0.07416 \eta$, so that $\nu_{F,\gamma,\dt} = 0.07416$ in this case. Right: Scaling of the mobility $\nu_{F,\gamma,\dt}$ for the first order scheme $P_\dt^{A,B_\eta,\gamma C}$ and the second order scheme $P_\dt^{\gamma C, B_\eta,A,B_\eta, \gamma C}$ (with $\gamma = 1$). The fits respectively give $\nu_{F,\gamma,\dt} \simeq 0.0740 + 0.0817\dt$ and $\nu_{F,\gamma,\dt} \simeq 0.0741 + 0.197\dt^2$.}
\end{figure}

\subsection{Error estimates in the overdamped limit}

We now study the numerical errors arising in the simulation of nonequilibrium systems in the large friction limit. We restrict ourselves to the second order splittings where the Ornstein-Uhlenbeck part is either at the ends or in the middle (categories~(i) and~(ii) in Section~\ref{sec:splitting_schemes_2nd}). To state the result, we introduce the first order correction to the invariant measure in terms of the magnitude of the nonequilibrium forcing, namely (recall $\wcL_{\rm ovd} = F \cdot \nabla_q$)
\[
\Lovd^* f_{0,1,\infty} = -\wcL_{\rm ovd}^* \mathbf{1} = -\beta F^T \nabla V.
\]
A simple computation based on~\eqref{eq:divergent_behavior_Lgamma} shows that the functions $f_{0,1,\gamma}$ defined in~\eqref{eq:def_nu_LRT} converge in $H^1(\mu)$ to $f_{0,1,\infty}$ (recall that we assume $M = \Id$ in the overdamped regime).

\begin{theorem}
\label{thm:error_estimate_noneq_ovd}
Denote by $\overline{\mu}_{\gamma,\eta,\dt}(\d{q})$ the marginal of the invariant measure $\mu_{\gamma,\eta,\dt}$ of an admissible second order splitting scheme in the position variable, and by $f_{2,0,\infty}$ the leading order correction function in the case $\eta = 0$ as given by Theorem~\ref{thm:ovd_limit}. Then, there exists a function $f_{2,1,\infty} \in \widetilde{\mathcal{S}}$ such that, for any $\psi \equiv \psi(q) \in C^\infty(\mathcal{M})$, there exist $\dt^*,\eta^* > 0$ and constants $K,\kappaK > 0$ such that, for all $\eta \in [-\eta^*,\eta^*]$, $0 < \dt \leq \dt^*$ and $\gamma \geq 1$,
\[
\int_\mathcal{M} \psi(q) \, \overline{\mu}_{\gamma,\eta,\dt}(\d{q}) = \int_\mathcal{M} \psi(q) \Big(1+ \eta f_{0,1,\infty}(q) + \dt^2 f_{2,0,\infty}(q) + \eta \dt^2 f_{2,1,\infty} \Big) \overline{\mu}(\d{q}) + r_{\psi,\gamma,\eta,\dt},
\]
with
\[
\left|r_{\psi,\gamma,\eta,\dt}\right| \leq K\left(\eta^2 + \dt^{3} + \rme^{-\kappaK \gamma \dt}\right), 
\qquad 
\left|r_{\psi,\gamma,\eta,\dt} - r_{\psi,\gamma,0,\dt}\right| \leq K \eta (\eta + \dt^{3} + \rme^{-\kappaK \gamma \dt}).
\]
\end{theorem}

The proof is presented in Section~\ref{sec:proof_thm:error_estimate_noneq_ovd}. This result allows us to estimate the error in the computation of the transport coefficient $\nu_{F,\gamma}$ based on~\eqref{eq:ovd_mobility} and Lemma~\ref{lem:ovd_mobility}. Indeed, studying the linear response of the observable $-F^T \nabla V$ and defining
\[
\overline{\nu}_{F,\gamma,\dt} = -\lim_{\eta \to 0} \frac1\eta \left(\int_\mathcal{M} F^T \nabla V(q) \, \overline{\mu}_{\gamma,\eta,\dt}(\d{q}) - \int_\mathcal{M} F^T\nabla V(q) \, \overline{\mu}_{\gamma,\dt}(\d{q}) \right),
\] 
there holds
\[
\overline{\nu}_F = \overline{\nu}_{F,\gamma,\dt} - \dt^2 \int_\mathcal{M} F^T \nabla V(q) f_{2,1,\infty}(q) \, \overline{\mu}(\d{q}) + r_{\psi,\gamma,\dt},
\]
with $|r_{\psi,\gamma,\dt}| \leq a(\dt^3 + \rme^{-\kappaK \gamma \dt})$ for some $a > 0$. Therefore, in view of~\eqref{eq:computatio_nu_F_gamma},
\begin{equation}
\label{eq:estimate_overline_nu_F}
\nu_{F,\gamma} = \frac{|F|^2 + \overline{\nu}_F}{\gamma} + \mathrm{O}\left(\frac{1}{\gamma^2}\right) = \frac{|F|^2 + \overline{\nu}_{F,\gamma,\dt}}{\gamma} + \mathrm{O}\left(\frac{1}{\gamma^2},\frac{\dt^2}{\gamma},\frac{\rme^{- \kappaK \gamma \dt}}{\gamma}\right).
\end{equation}
In the latter expression, $\overline{\nu}_{F,\gamma,\dt}$ can be numerically estimated, in a manner similar to that presented at the end of Section~\ref{sec:error_transport}.

%
%

\section{Proofs of the results}
\label{sec:proofs}

Unless otherwise stated, the default norm $\| f \|$ and scalar product $\langle f,g \rangle$ are the ones associated with the Hilbert space~$L^2(\mu)$. Recall that, unless otherwise mentioned, all operators are defined on~$\mathcal{S}$, and that formal adjoint operators are by default considered on~$L^2(\mu)$. Recall also that
\begin{equation}
\label{eq:decomposition_C}
C = -\frac1\beta \nabla_p^* \nabla_p = -\frac1\beta \sum_{i=1}^N \sum_{\alpha = 1}^d \partial_{p_{i,\alpha}}^* \partial_{p_{i,\alpha}},
\end{equation}
with $p_i = (p_{i,1},\dots,p_{i,d})$ since $\partial_{p_{i,\alpha}}^* = -\partial_{p_{i,\alpha}} + \beta p_{i,\alpha}$.

\subsection{Large friction behavior of $\mathcal{L}_\gamma^{-1}$}
\label{sec:proof_L_gamma_large}

The proof of Lemma~\ref{lem:bounded_resolvent_perp} follows the same lines as the proof of uniform hypocoercive estimates in the corrected version of Theorem~3 in~\cite{JS12} (see the erratum~\cite{JS12erratum} or the updated preprint version~\cite{JS12preprint}). We provide a simplified version of it for completeness.

\begin{proof}[Proof of Lemma~\ref{lem:bounded_resolvent_perp}]
We show that the operator $\mathcal{L}_\gamma$ is uniformly hypocoercive for $\gamma \geq 1$. The aim is to obtain bounds on the inverse $\mathcal{L}^{-1}_\gamma$ extended to $\cH^1_\perp$. To this end, we decompose $\Lgam$ for $\gamma \geq 1$ as
\[
\Lgam = \mathcal{L}_1 + (\gamma-1) C.
\]
The proof of Theorem~6.2 in~\cite{HP08} shows that there exists $\widetilde{\alpha} > 0$ such that, for all $u \in \mathcal{S}$,
\[
-\left\langle\left\langle u, \mathcal{L}_1 u \right\rangle\right\rangle
\geq \widetilde{\alpha} \left\langle\left\langle u, u \right\rangle\right\rangle,
\]
where the norm induced by 
$\left\langle\left\langle \cdot, \cdot\right\rangle\right\rangle$ is equivalent
to the $H^1(\mu)$ norm. More precisely, $\left\langle\left\langle \cdot, \cdot\right\rangle\right\rangle$
is the  bilinear form defined by
\[
\left\langle\left\langle u, v\right\rangle\right\rangle 
= a\left\langle u, v\right\rangle
+ b \left\langle \nabla_p u,\nabla_pv\right\rangle 
- \langle \nabla_p u, \nabla_q v\rangle 
- \langle \nabla_q u, \nabla_p v\rangle 
+ b\langle \nabla_q u,\nabla_q v\rangle,
\]
with appropriate coefficients $a \gg b \gg 1$.
It follows that there exists $\alpha > 0$ independent of $\gamma$ such that
\begin{equation}
\label{eq:eq_for_coercivity_Lgam}
\alpha \left\| u \right\|_{H^1(\mu)}^2 - (\gamma-1)
\left\langle\left\langle u, C u \right\rangle\right\rangle
\leq -\left\langle\left\langle u, \Lgam  u \right\rangle\right\rangle.
\end{equation}
Let us now show that 
\begin{equation}
\label{eq:estimate_positiviy_A}
\forall u \in \cH^1_\perp \cap \mathcal{S}, \qquad -\left\langle\left\langle u, C u \right\rangle\right\rangle 
\geq 0.
\end{equation}
Using the rewriting~\eqref{eq:decomposition_C} of the operator~$C$, and the commutation relations $[\partial_{p_{i,\alpha}}, \partial_{p_{j,\alpha'}}^*] = \beta \delta_{\alpha,\alpha'}\delta_{ij}$, a simple computation shows
\begin{align}
\left\langle\left\langle u, \left(\partial_{p_{i,\alpha}}\right)^*\partial_{p_{i,\alpha}} u \right\rangle\right\rangle 
& = (a+\beta b) \| \partial_{p_{i,\alpha}} u \|^2 + b \|\nabla_p \partial_{p_{i,\alpha}} u\|^2  \nonumber \\
& \quad + b \|\nabla_q \partial_{p_{i,\alpha}} u\|^2 - 2\langle \nabla_q \partial_{p_{i,\alpha}} u, \nabla_p \partial_{p_{i,\alpha}} u\rangle - \beta \langle \partial_{q_{i,\alpha}} u, \partial_{p_{i,\alpha}} u\rangle \nonumber \\
& \geq \left(a+\beta \left(b-\frac12\right)\right) \| \partial_{p_{i,\alpha}} u \|^2 + (b-1) \|\nabla_p \partial_{p_{i,\alpha}} u\|^2 \label{eq:ineg_pi*pi} \\
& \qquad + (b-1) \|\nabla_q \partial_{p_{i,\alpha}} u\|^2 - \frac{\beta}{2} \| \partial_{q_{i,\alpha}} u \|^2. \nonumber
\end{align}
Now, since the Gaussian measure $\kappa(\d{p})$ satisfies a Poincar\'e inequality, there exists a constant $A > 0$ such that, for all $i = 1,\dots,N$ and $\alpha = 1,\dots,d$,
\[
  \| \partial_{q_{i,\alpha}}u \|^2 \leq A \| \nabla_p \partial_{q_{i,\alpha}} u \|^2.
\]
Note indeed that $\partial_{q_{i,\alpha}}u$ has a vanishing average with respect to the Gaussian measure $\kappa(\d{p})$ because 
\[
\int_{\RR^{dN}} \partial_{q_{i,\alpha}} u(q,p) \, \kappa(\d{p}) = \partial_{q_{i,\alpha}} \overline{u}(q) = 0 
\]
for functions $u \in \cH^1_\perp$. Therefore,
\[
\sum_{i=1}^N \sum_{\alpha=1}^d \| \partial_{q_{i,\alpha}} u \|^2 \leq A \sum_{i,j=1}^N \sum_{\alpha,\alpha' = 1}^d \| \partial_{p_{j,\alpha'}} \partial_{q_{i,\alpha}} u \|^2 = A \sum_{j=1}^N \sum_{\alpha' = 1}^d \| \nabla_{q} \partial_{p_{j,\alpha'}}\|^2.
\]
Summing~\eqref{eq:ineg_pi*pi} on $i \in \{ 1,\dots,N\}$ and $\alpha \in \{1,\dots,d\}$, the quantity~\eqref{eq:estimate_positiviy_A} is seen to be non-negative for an appropriate choice of constants $a \gg b \gg 1$. 

From~\eqref{eq:eq_for_coercivity_Lgam}, we then deduce that there exists a constant $K>0$ such that, for any $\gamma \geq 1$ and for any $u \in \cH^1_\perp \cap \mathcal{S}$, it holds $\left\| u \right\|_{H^1(\mu)} \leq K \| \Lgam u \|_{H^1(\mu)}$. Taking inverses and passing to the limit in $\cH^1_\perp$ gives
\[
\forall \gamma \geq 1, \quad \forall u \in \cH^1_\perp, \qquad \left\| \Lgam^{-1}  u \right\|_{H^1(\mu)} \leq K \| u \|_{H^1(\mu)},
\]
which is the desired result.
\end{proof}

\medskip

We are now in position to give the proof  of Theorem~\ref{lem:bounds_CL_gamma}.

\smallskip

\begin{proof}[Proof of Theorem~\ref{lem:bounds_CL_gamma}]
We write the proof for $\Lgam^{-1}$. The estimates for $(\Lgam^*)^{-1}$ are obtained by using $\Lgam^* = \mathcal{R} \Lgam \mathcal{R}$ (the momentum reversal operator being defined in~\eqref{eq:op_R}), and the fact that $\mathcal{R} C \mathcal{R} = C$, $\mathcal{R} \Lovd \mathcal{R} = \Lovd$ and $\mathcal{R} (A+B) \mathcal{R} = -(A+B)$.

The lower bound in~\eqref{eq:crude_bounds_Lgamma_ovd} could be obtained directly provided $V$ is not constant, by considering the special case 
\[
\cL_\gamma \Big( p^T \nabla V + \gamma (V-v)\Big) = p^T M^{-1} \left(\nabla^2 V\right)p - |\nabla V|^2, 
\]
where $v$ is a constant chosen such that $p^T \nabla V + \gamma (V-v)$ has a vanishing average with respect to~$\mu$. This example is also useful to motivate the fact that, in general, solutions of the Poisson equation $\cL_\gamma u_\gamma = f$ have divergent parts of order~$\gamma$ as $\gamma \to +\infty$.

Let us now turn to the refined upper and lower bounds~\eqref{eq:divergent_behavior_Lgamma}, which we prove using techniques from asymptotic analysis. Consider $f \in \cH^1$, and $u_\gamma \in \cH^1$ the unique solution of the following Poisson equation $\cL_\gamma u_\gamma = f$. The above example suggests the following expansion in inverse powers of $\gamma$:
\begin{equation}
\label{eq:ansatz_u_pi}
u_\gamma = \gamma u^{-1} + u^0 + \frac1\gamma u^1 + \dots
\end{equation}
To rigorously prove this expansion, we first proceed formally, taking~\eqref{eq:ansatz_u_pi} as an ansatz, plugging it into  $\cL_\gamma u = f$ and identifying terms according to powers of~$\gamma$. This leads to
\begin{align*}
C u^{-1} & = 0, \\
(A+B)u^{-1} + Cu^0 & = 0, \\
(A+B)u^0 + Cu^1 & = f.
\end{align*}
The first equality implies that $u^{-1} = u^{-1}(q)$ since $C$ satisfies a Poincar\'e inequality on $L^2(\kappa)$ (where $\kappa$ is defined in~\eqref{eq:pi}). The second then reduces to $Cu^0 = -M^{-1} p \cdot \nabla_q u^{-1}$, from which we deduce $u^0(q,p) = p^T\nabla u^{-1}(q) + \widetilde{u}^0(q)$. Inserting this expression in the third equality gives
\[
Cu^1 = f - p^T M^{-1} \left(\nabla^2 u^{-1}\right) p - p^T M^{-1} \nabla \widetilde{u}^0 + (\nabla V)^T \nabla u^{-1}.
\]
The solvability condition for this equation is that the right-hand side has a vanishing average with respect to~$\kappa$, \textit{i.e.} belongs to the kernel of~$\pi$. This condition reads
\[
\frac1\beta \Delta u^{-1} - (\nabla V)^T \nabla u^{-1} = \pi f,
\]
so that $u^{-1} = \Lovd^{-1} \pi f$ (which is well defined since $\pi f$ has a vanishing average with respect to~$\overline{\mu}$). Note that the function $u^{-1}$ is in $H^{n+2}(\overline{\mu})$ when $f \in H^n(\mu)$ (by elliptic regularity, using also the fact that $\rme^{-\beta V(q)}$ is a smooth function bounded from above and below on~$\mathcal{M}$), so that $p^T M^{-1} (\nabla^2 u^{-1}) p$ belongs to $L^2(\mu)$. The equation determining $u^1$ then reduces to
\[
Cu^1 = (f-\pi f) - p^T M^{-1} \nabla \widetilde{u}^0 - p^T M^{-1} \left(\nabla^2 u^{-1}\right) p + \frac1\beta\Delta u^{-1}.
\]
Since $C(p^T  A p) = -p^T M^{-1} (A + A^T)p + 2\beta^{-1} \mathrm{Tr}(A)$, we can choose 
\[
u^1(q,p) = \left[C^{-1}(f -\pi f)\right](q,p) + \frac12 p^T (\nabla^2 u^{-1}(q)) p + p^T \nabla_q \widetilde{u}^0(q). 
\]
Coming back to~\eqref{eq:ansatz_u_pi}, we see that the proposed approximate solution is such that
\begin{equation}
\label{eq:difference_with_exact_sol}
\cL_\gamma\left(u_\gamma - \gamma u^{-1} - u^0 - \frac1\gamma u^1\right) = -\frac1\gamma (A+B)u^1.
\end{equation}
We now choose $\widetilde{u}^0$ such that $(A+B)u^1$ belongs to $\cH^1_\perp$, which amounts to
\[
\pi(A+B)p^T \nabla_q \widetilde{u}^0 = \Lovd \widetilde{u}^0 = -\pi(A+B)C^{-1}(f -\pi f).
\]
It is easily checked that $\widetilde{u}^0 = -\Lovd^{-1}\pi(A+B)C^{-1}(f -\pi f)$ is a well defined element in $\cH^1$ for $f \in H^1(\mu)$: first, $C^{-1}(f -\pi f) \in \cH^1$, so $(A+B)C^{-1}(f -\pi f) \in L^2(\mu)$. Finally, the image under $\Lovd^{-1}\pi$ of any function in $L^2(\mu)$ is a function of average zero with respect to $\overline{\mu}$, depending only on the position variable~$q$ and belonging to~$H^2(\overline{\mu})$; hence to~$\cH^1$.

Combining~\eqref{eq:difference_with_exact_sol} and Lemma~\ref{lem:bounded_resolvent_perp}, we see that there exists a constant $R > 0$, such that, for all $\gamma \geq 1$, it holds $\| u_\gamma - \gamma u^{-1} - u^0 \|_{H^1(\mu)} \leq R \| f \|_{H^1(\mu)}/\gamma$ for the above choices of functions $u^{-1}, u^0$. This  gives~\eqref{eq:divergent_behavior_Lgamma}.
\end{proof}

\subsection{Ergodicity results for numerical schemes}
\label{sec:proof_ergodicity_MC}

\begin{proof}[Proof of Lemma~\ref{lem:Lyapunov}]
We write the proof for the scheme associated with the evolution operator $P_\dt^{B,A,\gamma C}$, starting by the case $s=1$, before turning to the general case $s \geq 2$. The proofs for other schemes are very similar, and we therefore skip them.

The numerical scheme corresponding to $P_\dt^{B,A,\gamma C}$ is~\eqref{eq:Langevin_splitting}. We introduce $m \in (0,+\infty)$ such that $m \leq M \leq m^{-1}$ (in the sense of symmetric matrices). A simple computation shows that
\[
\begin{aligned}
\mathbb{E}\left[\left.\left(p^{n+1}\right)^2\,\right|\,\mathcal{F}_n\right] & = \left(p^n- \dt \nabla V(q^n)\right)^T \alpha_\dt^2 \left(p^n- \dt \nabla V(q^n)\right) + \frac1\beta \Tr\left[\left(1-\alpha_\dt^2\right)M^2\right] \\
& \leq \rme^{-2m\gamma\dt} \left(p^n\right)^2 + 2 \dt \left\|\nabla V\right\|_{L^\infty} \left|p^n\right| + \dt^2 \left\|\nabla V\right\|_{L^\infty}^2 + \frac{1-\rme^{-2\gamma\dt/m}}{\beta m^2} \\
& \leq \left(\rme^{-2m\gamma\dt} + \varepsilon \dt \right) \left(p^n\right)^2 + \dt \left( \frac1\varepsilon + \dt\right) \left\|\nabla V\right\|_{L^\infty}^2 + \frac{1-\rme^{-2\gamma\dt/m}}{\beta m^2}.
\end{aligned}
\]
We choose for instance $\varepsilon = m\gamma$, in which case 
\[
0 \leq \rme^{-2m\gamma\dt} + \varepsilon \dt \leq \exp\left( -C_a \dt \right), \qquad C_a = \frac{m\gamma}{2},
\]
for $\dt$ sufficiently small, and
\[
0 \leq \dt \left( \frac1\varepsilon + \dt\right) \left\|\nabla V\right\|_{L^\infty}^2 + \frac{1-\rme^{-2\gamma\dt/m}}{\beta m^2} \leq \widetilde{C}_b \dt, 
\qquad
\widetilde{C}_b = \frac{2}{m\gamma} \left\|\nabla V\right\|_{L^\infty}^2 + \frac{4\gamma}{\beta m^3},
\]
for $\dt$ sufficiently small. Finally, since $\Li_2(q,p) = 1+|p|^2$,
\[
\mathbb{E}\left[\left. \Li_2\left(q^{n+1},p^{n+1}\right)\,\right|\,\mathcal{F}_n\right] \leq \rme^{-C_a\dt} \Li_2\left(q^n,p^n\right) + 1-\rme^{-C_a\dt} + \widetilde{C}_b \dt \leq \rme^{-C_a\dt} \Li_2\left(q^n,p^n\right) + C_b \dt,
\]
for $\dt$ sufficiently small. This gives~\eqref{eq:moment_estimates}. To obtain~\eqref{eq:moment_estimates_uniform}, we iterate the bound~\eqref{eq:moment_estimates}:
\[
\begin{aligned}
P_\dt^n \Li_s & \leq \rme^{-C_a \,n\dt} \Li_s + C_b \dt \left(1 + \rme^{-C_a\dt} + \dots + \rme^{-C_a\, (n-1)\dt}\right) \leq \rme^{-C_a \, n\dt} \Li_s + \frac{C_b \dt}{1-\rme^{-C_a\dt}}. 
\end{aligned}
\]

The computations are similar for a general power $s \geq 2$. We write $p^{n+1} = \alpha_\dt p^n + \delta_\dt$ with $\delta_\dt = - \alpha_\dt \dt \nabla V(q^n) + \sqrt{\beta^{-1}(1-\alpha_\dt^2)M} \, G^n$. Note that $\delta_\dt$ is of order~$\dt^{1/2}$ because of the random term. We work componentwise, using the assumption that $M$ is diagonal, so that, denoting by $m_i$ the mass of the $i$th degree of freedom,
\[
\begin{aligned}
\left(p_i^{n+1}\right)^{2s} & = \left(\rme^{-\gamma \dt/m_i} p_i^n + \delta_{i,\dt} \right)^{2s} \\
& = \rme^{-2s \gamma \dt/m_i} \left(p_i^n\right)^{2s} + 2s \, \rme^{-(2s-1)\gamma \dt/m_i} \left(p_i^n\right)^{2s-1} \delta_{i,\dt} \\
& \ \ \ + s(2s-1) \rme^{-2(s-1)\gamma \dt/m_i} \left(p_i^n\right)^{2(s-1)} \delta_{i,\dt}^2 + \dots
\end{aligned}
\]
Taking expectations,
\[
\begin{aligned}
& \mathbb{E}\left[ \left. \left(p_i^{n+1}\right)^{2s} \, \right| \, \mathcal{F}_n \right]
= \rme^{-2s \gamma \dt/m_i} \left(p_i^n\right)^{2s} - 2s \, \dt\, \rme^{-2s \gamma \dt/m_i} \left(p_i^n\right)^{2s-1} \partial_{q_i}V(q^n) \\
& \ \ \ + s(2s-1) \rme^{-2(s-1)\gamma \dt/m_i} \left(p_i^n\right)^{2(s-1)} \left(\dt^2 \rme^{-2 \gamma \dt/m_i} \partial_{q_i}V(q^n) + \frac{(1-\rme^{-2 \gamma \dt/m_i})m_i}{\beta} \right) \\
& \ \ \ + \dt^2 r_{s,\dt,i}(q^n) \left(1+\left(p^n\right)^{2s-3}\right),
\end{aligned}
\]
where the remainder $r_{s,\dt}(q^n)$ is uniformly bounded as $\dt \to 0$. Distinguishing between $|p_i| \geq 1/\varepsilon$ and $|p_i| \leq 1/\varepsilon$, we have 
\[
|p_i|^{2s -m} \leq \varepsilon^m (p_i)^{2s} + \frac{1}{\varepsilon^{2s-m}},
\]
from which we obtain
\[
\mathbb{E}\left[ \left. \left(p_i^{n+1}\right)^{2s} \, \right| \, \mathcal{F}_n \right] \leq 
\widehat{a}_{\dt,\varepsilon,i} \left(p_i^n\right)^{2s} + \widehat{b}_{\dt,\varepsilon,i},
\]
with
\[
\begin{aligned}
\widehat{a}_{\dt,\varepsilon,i} & = \rme^{-2s \gamma \dt/m_i} + 2s\varepsilon \dt \| \partial_{q_i}V \|_{L^\infty} \\
& \ \ \ + s(2s-1) \varepsilon^2 \left(\dt^2 \| \partial_{q_i}V \|_{L^\infty} + \frac{(1-\rme^{-2 \gamma \dt/m_i})m_i}{\beta}\right) + \varepsilon^3 \dt^2 \| r_{s,\dt,i} \|_{L^\infty} ,
\end{aligned}
\]
and
\[
\begin{aligned}
\widehat{b}_{\dt,\varepsilon,i} & = \frac{2s}{\varepsilon} \dt \| \partial_{q_i}V \|_{L^\infty} \\ 
& \ \ \ + \frac{s(2s-1)}{\varepsilon^2} \left(\dt^2 \| \partial_{q_i}V \|_{L^\infty} + \frac{(1-\rme^{-2 \gamma \dt/m_i})m_i}{\beta}\right) + \dt^2\left(1 + \frac{1}{\varepsilon^3}\right) \| r_{s,\dt,i} \|_{L^\infty}.
\end{aligned}
\]
The proof is then concluded as in the case $s=1$ by choosing $\varepsilon$ sufficiently small (independently of $\dt$).
\end{proof}

\bigskip

\begin{proof}[Proof of Lemma~\ref{lem:minorization}]
  It is sufficient to prove the result for indicator functions of Borel sets~$A = A_q \times A_p \subset \cE$, where $A_q \subset \mathcal{M}$ while $A_p \subset \mathbb{R}^{dN}$ (see~\cite{Rudin}). We therefore aim at proving 
\[
\mathbb{P}\left( \left(q^n,p^n\right) \in A \, \left| \, \left|p^0\right| \leq p_{\rm max} \right.\right) \geq \alpha \, \nu(A),
\]
for a well chosen probability measure~$\nu$ and a constant $\alpha > 0$. The idea of the proof is to explicitly rewrite $q^n$ and $p^n$ as perturbations of the reference evolution corresponding to $\nabla V = 0$ and $(q^0,p^0) = (0,0)$. Since we consider smooth potentials and the position space is compact, the perturbation can be uniformly controlled when the initial momenta are within a compact set.

We write the proof for the scheme associated with the evolution operator $P_\dt^{B,A,\gamma C}$, as in the proof of Lemma~\ref{lem:Lyapunov}. A simple computation shows that, for $n \geq 1$,
\[
q^n = q^0 + \dt M^{-1} \left(p^{n-1} + \dots + p^0 \right) - \dt^2 M^{-1} \Big( \nabla V(q^{n-1}) + \dots + \nabla V(q^0) \Big),
\]
and
\[
\begin{aligned}
p^n = \alpha_\dt^n \, p^0 & - \dt \, \alpha_\dt \left( \nabla V(q^{n-1}) + \alpha_\dt \nabla V(q^{n-2}) + \dots + \alpha_\dt^{n-1} \nabla V(q^0) \right) \\
& + \sqrt{\frac{1-\alpha_\dt^2}{\beta} M}\left(G^{n-1} + \alpha_\dt G^{n-2} + \dots + \alpha_\dt^{n-1} G^0 \right).
\end{aligned}
\]
Denote by $\mathcal{G}^n$ the centered Gaussian random variable
\[
\mathcal{G}^n = \sqrt{\frac{1-\alpha_\dt^2}{\beta} M}\left(G^{n-1} + \alpha_\dt G^{n-2} + \dots + \alpha_\dt^{n-1} G^0 \right).
\]
Introduce also 
\[
\begin{aligned}
F^n & = -\alpha_\dt \left( \nabla V(q^{n-1}) + \alpha_\dt \nabla V(q^{n-2}) + \dots + \alpha_\dt^{n-1} \nabla V(q^0) \right), \\
\mathscr{P}^n & = \alpha_\dt^n \, p^0 + \dt \, F^n, \\
\mathscr{Q}^n & = q^0 + \dt M^{-1} \left(\dt \sum_{m=0}^{n-1} F^m + \frac{1-\alpha_\dt^n}{1-\alpha_\dt} p^0 \right) - \dt^2 M^{-1} \Big( \nabla V(q^{n-1}) + \dots + \nabla V(q^0) \Big).\\
\end{aligned}
\]
With this notation,
\[
p^n = \mathscr{P}^n + \mathcal{G}^n, \qquad q^n = \mathscr{Q}^n + \widetilde{\mathcal{G}}^n,
\]
where
\[
\begin{aligned}
\widetilde{\mathcal{G}}^n & = \dt M^{-1} \sum_{m=1}^{n-1} \mathcal{G}^m \\
& = \dt \sqrt{\frac{1-\alpha_\dt^2}{\beta} M^{-1}}\Big(G^{n-2} + (1+\alpha_\dt) G^{n-3} + \dots + (1+\alpha_\dt+\dots+\alpha_\dt^{n-2}) G^0 \Big)
\end{aligned}
\]
is a centered Gaussian random variable. Now, 
\begin{equation}
\label{eq:first_ineq_minorization}
\mathbb{P}\left( \left(q^n,p^n\right) \in A \, \left| \, \left|p^0\right| \leq p_{\rm max} \right.\right)
= \mathbb{P}\left(\left. \left(\widetilde{\mathcal{G}}^n,\mathcal{G}^n\right) \in \left(A_q -\mathscr{Q}^n\right) \times \left( A_p - \mathscr{P}^n\right) \right| \, \left|p^0\right| \leq p_{\rm max} \right).
\end{equation}
In fact, we consider in the sequel that the random variable $\widetilde{\mathcal{G}}^n$ has values in~$\mathbb{R}^{dN}$ rather than~$\mathcal{M}$ and understand $A_q -\mathscr{Q}^n$ as a subset of~$\mathbb{R}^{dN}$ rather than~$\mathcal{M}$. This amounts to neglecting the possible periodic images, and henceforth reduces the probability on the right-hand side of the above inequality. This is however not a problem since we seek a lower bound.

Note that $\dt \, F^n$ is uniformly bounded: using $0 \leq \alpha_\dt \leq \exp(-\gamma m \dt)$ in the sense of symmetric, positive matrices (with $m \leq M \leq m^{-1}$),
\[
\left| \dt \, F^n \right| \leq   \| \nabla V \|_{L^\infty} \, \frac{\dt}{1-\exp(-\gamma m \dt)} \leq \frac{2 }{m\gamma} \, \| \nabla V \|_{L^\infty}
\]
provided $\dt$ is sufficiently small. Therefore, there exists a constant $R > 0$ (depending on $p_{\rm max}$) and $\dt^*>0$ such that, for all timesteps $0 < \dt \leq \dt^*$ and corresponding integration steps $0 \leq n \leq T/\dt$, 
\begin{equation}
\label{eq:second_ineq_minorization}
\left| \mathscr{Q}^n \right|\leq R, \qquad \left| \mathscr{P}^n \right|\leq R.
\end{equation}
A lengthy but straightforward computation shows that the variance of the centered Gaussian vector $\left(\widetilde{\mathcal{G}}^n,\mathcal{G}^n\right)$ is
\[
\mathscr{V}^n = \mathbb{E}\left[ \left(\widetilde{\mathcal{G}}^n,\mathcal{G}^n\right)^T \left(\widetilde{\mathcal{G}}^n,\mathcal{G}^n\right) \right] = \begin{pmatrix} \mathscr{V}^n_{qq} & \mathscr{V}^n_{qp} \\ \mathscr{V}^n_{qp} & \mathscr{V}^n_{pp} \end{pmatrix}
\]
with
\[
\left\{ \begin{aligned}
\mathscr{V}^n_{qq} & = \frac{\dt \, (1-\alpha_\dt^2)}{(1-\alpha_\dt)^2} M^{-1}\left( (n-1)\dt - \frac{2\dt \, \alpha_\dt}{1-\alpha_\dt} (1-\alpha_\dt^{n-1}) + \frac{\dt \, \alpha_\dt^2}{1-\alpha_\dt^2} \left(1-\alpha_\dt^{2(n-1)}\right)\right),\\
\mathscr{V}^n_{qp} & = \frac{\dt \, \alpha_\dt}{\beta (1-\alpha_\dt)} \Big(1 - \alpha_\dt^{n-1} (1+\alpha_\dt) + \alpha_\dt^{2n-1} \Big),\\
\mathscr{V}^n_{pp} & =  \frac{M}{\beta}(1-\alpha_\dt^{2n}).
\end{aligned} \right.
\]
To check that this expression is appropriate, we note that it converges as $\dt \to 0$ with $n\dt \to T$ to the variance of the limiting continuous process
\[
\d{q}_t = M^{-1} p_t \, \d{t}, \qquad \d{p}_t = -\gamma M^{-1} p_t \, \d{t} + \sqrt{\frac{2\gamma}{\beta}} \, \d{W}_t,
\]
starting from $(q_0,p_0) = (0,0)$, which reads
\[
\mathscr{V} = \begin{pmatrix} \mathscr{V}_{qq} & \mathscr{V}_{qp} \\ \mathscr{V}_{qp} & \mathscr{V}_{pp} \end{pmatrix},
\]
with
\[
\left\{ 
\begin{aligned}
\mathscr{V}_{qq} & = \frac{1}{\beta \gamma}\left(2T - \frac{M}{\gamma}\left(3 - 4 \, \alpha_T + \alpha_T^2\right)\right),\\
\mathscr{V}_{qp} & = \frac{M}{\beta \gamma}\left(1-\alpha_T\right)^2, \\
\mathscr{V}_{pp} & = \frac{M}{\beta} \left(1-\alpha_T^2\right).
\end{aligned}
\right.
\]
Upon reducing~$\dt^* > 0$, it holds $\mathscr{V}/2 \leq  \mathscr{V}^{\lceil T/\dt \rceil} \leq 2\mathscr{V}$ for~$0 < \dt \leq \dt^*$. In particular, $\mathscr{V}^{\lceil T/\dt \rceil}$ is invertible for $T$ sufficiently large. For a set $E_q \times E_p \subset \mathbb{R}^{2dN}$, it then holds that
\begin{align}
\mathbb{P}\left( \left(\widetilde{\mathcal{G}}^{\lceil T/\dt \rceil},\mathcal{G}^{\lceil T/\dt \rceil}\right) \in E \right)  &= 
(2\pi)^{-dN} \mathrm{det}\left(\mathscr{V}^{\lceil T/\dt \rceil}\right)^{-1/2} \int_{E_q \times E_p} \exp\left( -\frac12 x^T \left(\mathscr{V}^{\lceil T/\dt \rceil}\right)^{-1} x \right) \, \d{x} \nonumber \\
& \geq \pi^{-dN} 2^{-3dN/2} \mathrm{det}\left(\mathscr{V}\right)^{-1/2} \int_{E_q \times E_p} \exp\left( -x^T \mathscr{V}^{-1} x \right) \, \d{x}. \label{eq:third_ineq_minorization}
\end{align}
The result follows by combining~\eqref{eq:first_ineq_minorization}-\eqref{eq:second_ineq_minorization}-\eqref{eq:third_ineq_minorization} and introducing the probability measure
\[
\nu(A_q \times A_p) = Z_R^{-1} \inf_{|\mathscr{Q}|, |\mathscr{P}| \leq R} \int_{(A_q - \mathscr{Q}) \times (A_p-\mathscr{P})} \exp\left( -x^T \mathscr{V}^{-1} x \right) \, \d{x},
\]
as well as $\alpha  = Z_R \pi^{-dN} 2^{-3dN/2} \mathrm{det}\left(\mathscr{V}\right)^{-1/2}$.
\end{proof}

\bigskip

\begin{proof}[Proof of Proposition~\ref{prop:ergodicity_MC}]
We only prove~\eqref{eq:ergodicity_num} and~\eqref{eq:moment_estimate} since the other results are standard. To obtain the bound~\eqref{eq:ergodicity_num}, we first note that, by the results of~\cite{HM11}, there exists $\widetilde{\lambda} > 0$ such that, for any function $f \in L^\infty_{\Li_s,\dt}$ and $0 < \dt \leq \dt^*$ (the critical timestep being given by Lemmas~\ref{lem:Lyapunov} and~\ref{lem:minorization}), the following holds for almost all $(q,p) \in \mathcal{E}$:
\[
\forall m \in \mathbb{N}, \qquad \left| \left( \left[P_\dt^{\lceil T/\dt \rceil}\right]^m f\right)(q,p) \right| \leq K \, \Li_s(q,p) \, \rme^{-\widetilde{\lambda} m} \, \| f \|_{L^\infty_{\Li_s}}.
\]
For a general index $n \in \mathbb{N}$, we write 
\[
n = m_n \left\lceil \frac{T}{\dt} \right\rceil + \widetilde{n}, \qquad 0 \leq \widetilde{n} \leq \left\lceil \frac{T}{\dt} \right\rceil - 1, 
\]
so that, using the contractivity property $|P_\dt f(q,p) | \leq |f(q,p)|$,
\[
\left| P_\dt^n f(q,p) \right| \leq K \, \Li_s(q,p) \, \rme^{-\widetilde{\lambda} m_n} \, \| f \|_{L^\infty_{\Li_s}}. 
\]
Introducing $\lambda = \widetilde{\lambda}/T$, the argument of the exponent reads 
\[
\widetilde{\lambda} m_n = \lambda (n-\widetilde{n})\dt \, \frac{T}{\dt}\left\lceil \frac{T}{\dt} \right\rceil^{-1} \geq \frac{\lambda n \dt}{2} - \lambda T,
\]
when $\dt$ is sufficiently small. This gives~\eqref{eq:ergodicity_num}.

The moment estimate~\eqref{eq:moment_estimate} is obtained by averaging~\eqref{eq:moment_estimates} with respect to the invariant measure: 
\[
\int_\cE \left(P_\dt \Li_s\right) \d{\mu}_{\gamma,\dt} \leq \rme^{-C_a \dt} \int_\cE \Li_s \, \d{\mu}_{\gamma,\dt} + C_b \dt.
\]
Since $\mu_{\gamma,\dt}$ is invariant, 
\[
\int_\cE \left(P_\dt \Li_s\right) \d{\mu}_{\gamma,\dt} = \int_\cE \Li_s \, \d{\mu}_{\gamma,\dt},
\]
so that 
\[
\left(1-\rme^{-C_a \dt}\right) \int_\cE \Li_s \, \d{\mu}_{\gamma,\dt} \leq C_b \dt,
\]
which gives the desired result with $R=2C_b/C_a$ for instance, provided $\dt$ is sufficiently small.
\end{proof}

\subsection{Some useful results}
\label{sec:useful}

\subsubsection{Expansion of the evolution operator}
\label{sec:expansion_evolution}

We give in this section an expression for the evolution operator 
\[
P_t = \rme^{t A_M} \dots \rme^{t A_1},
\]
which can easily be compared to the evolution operator $\rme^{t (A_1+\dots+A_M)}$. We assume that the generators $A_i$ of all elementary dynamics are well defined operators on a core~$X$, with image in~$X$ (typically, $X = \mathcal{S}$ or a subset of this space such as $\widetilde{\mathcal{S}}$). We also assume that the elementary evolution semigroups $\rme^{t A_i}$, as well as $P_t$, are well defined on~$X$ with values in~$X$. These semigroups may be extended to bounded operators on an appropriate Banach space using the Hille-Yosida theorem for instance (see~\cite{Pazy}). All the operator equalities stated in this section have to be considered in the strong sense, namely $T_1 = T_2$ means $T_1 \varphi = T_2 \varphi$ for all $\varphi \in X$. 

It is easy to check that the operators $A,B,C$ defined in~\eqref{eq:def_ABC} map $\mathcal{S}$ to itself.
It is in fact possible to analytically write down the action of the associated semigroups:
\begin{equation}
\label{eq:analytic_expressions_semigroups}
\left\{ \begin{aligned}
\left(\etA \varphi\right)(q,p) & = \varphi\Big(q+tM^{-1}p,p\Big), \\
\left(\etB \varphi\right)(q,p) & = \varphi\Big(q,p-t\nabla V(q)\Big), \\
\left(\etC \varphi\right)(q,p) & = \int_{\RR^{dN}} \varphi\left(q,\rme^{-\gamma M^{-1} t}p + \left(\frac{1-\rme^{-2\gamma M^{-1} t}}{\beta}M\right)^{1/2} x\right) \frac{\rme^{-|x|^2/2}}{(2\pi)^{dN/2}} \, \d{x}.
\end{aligned} \right.
\end{equation}

Coming back to the general case, the key building block for the subsequent numerical analysis is the following equality:
\[
P_t = P_0 + t \left. \frac{dP_t}{dt} \right|_{t=0} + \frac{t^2}{2} \left. \frac{d^2 P_t}{dt^2} \right|_{t=0} + \dots + \frac{t^n}{n!} \left. \frac{d^n P_t}{dt^n} \right|_{t=0} + \frac{t^{n+1}}{n!} \int_0^1 (1-\theta)^n \left. \frac{d^{n+1} P_s}{ds^{n+1}} \right|_{s = \theta t} \, \d{\theta}.
\]
Now,
\[
\begin{aligned}
\frac{dP_t}{dt} & = A_M \rme^{t A_M} \dots \rme^{t A_1} + \rme^{t A_M} A_{M-1} \rme^{t A_{M-1}} \dots \rme^{t A_1} + \dots
+ \rme^{t A_M} \dots \rme^{t A_1} A_1 \\
& = \cT\left[ (A_1+\dots+A_M) P_t \right]
\end{aligned}
\]
where $\cT$ is a notation indicating that the operators with the smallest indices (or their associated semigroups) are farthest to the right. In fact, simple computations show that
\[
\frac{d^n P_t}{dt^n} = \cT\Big[ (A_1+\dots+A_M)^n P_t \Big].
\]
Therefore, the following equality holds when applied to functions~$\varphi \in X$:
\begin{equation}
\label{eq:P_t_expansion}
\begin{aligned}
P_t \varphi & = \varphi + t(A_1+\dots+A_M)\varphi + \frac{t^2}{2}\cT\Big[ (A_1+\dots+A_M)^2 \Big]\varphi +\dots + \frac{t^n}{n!}\cT\Big[ (A_1+\dots+A_M)^n \Big]\varphi \\
& \ \ + \frac{t^{n+1}}{n!} \int_0^1 (1-\theta)^n \cT\Big[ (A_1+\dots+A_M)^{n+1} P_{\theta t}\Big]\varphi\, \d{\theta}.
\end{aligned}
\end{equation}

\subsubsection{Baker-Campbell-Hausdorff (BCH) formula}
\label{sec:BCH}

It is important to rewrite the various terms in the right-hand side of~\eqref{eq:P_t_expansion} in a form more amenable to analytical computations. More precisely, it is  convenient to write the following equality in terms of operators defined on~$X$:
\[
\cT\Big[ (A_1+\dots+A_M)^n \Big] = (A_1+\dots+A_M)^n + S_n,
\]
where the operator $S_n$ involves commutators $[A_i,A_j]$, which can also be defined as operators on~$X$ with values in~$X$. In fact, the algebraic expressions of the operators $S_n$ can be formally obtained from the BCH formula for first order splittings (see for instance~\cite[Section~III.4.2]{HairerLubichWanner06}): for $M=3$,
\[
\rme^{\dt A_3} \rme^{\dt A_2} \rme^{\dt A_1} = \rme^{\dt \mathcal{A}}, 
\qquad \mathcal{A} = A_1+A_2+A_3 + \frac{\dt}{2} \Big( [A_3,A_1+A_2] + [A_2,A_1]\Big) + \dots,
\]
and from the symmetric BCH formula for second order involving 3~operators (obtained by composition of the standard BCH formula involving 2~operators):
\begin{equation}
\label{eq:BCH}
\rme^{\dt A_1/2} \rme^{\dt A_2/2} \rme^{\dt A_3} \rme^{\dt A_2/2} \rme^{\dt A_1/2}
= \rme^{\dt \mathcal{A}}, 
\end{equation}
with
\[
\begin{aligned}
\mathcal{A} = A_1+A_2+A_3 + \frac{\dt^2}{12} \left( [A_3,[A_3,A_2]] + [A_2+A_3,[A_2+A_3,A_1]] \phantom{\frac12} \right.&\\ 
\left. - \frac12 [A_2,[A_2,A_3]] - \frac12 [A_1,[A_1,A_2+A_3]]\right) + \dots&
\end{aligned}
\]
where we do not write down the expressions of the higher order terms $\dt^{2n}$ (for $n \geq 2$). Let us insist that these formulas are only formal (since the operators appearing the argument of the exponential on the right-hand side involve more and more derivatives), but nonetheless allow us to find the algebraic expressions of $S_n$ upon formally expanding the exponential as 
\[
\rme^{\dt \mathcal{A}} = \Id + \dt \mathcal{A} + \frac{\dt^2}{2} \mathcal{A}^2 + \dots
\]
and identifying terms with the same powers of $\dt$ in~\eqref{eq:P_t_expansion}.

\subsubsection{Approximate inverse operators}
\label{sec:approx_inv}

Consider an operator $A$ defined on some core~$X$ (typically some subspace of~$\mathcal{S}$), and whose inverse is also defined on~$X$ in the following sense: for any $g \in X$, there exist $f \in X$ such that $Af=g$. We denote by $A^{-1}g$ the element~$f \in X$. At this stage, we do not assume any boundedness in an appropriate operator norm for~$A^{-1}$ or some extension of it. We next consider a perturbation $\dt^\alpha B$ for some exponent $\alpha \geq 1$, where $B$ is also defined on~$X$ and has values in~$X$. In the typical situations encountered in this article, $B$ is not bounded with respect to $A$ in an appropriate operator norm since it may involve higher order derivatives than $A$ does. It is therefore impossible in general to properly define the inverse of $A + \dt^\alpha B$.

However, it is possible to introduce an approximate inverse, which we define as an operator $Q_{\dt,n}$ from $X$ to~$X$ such that there exists an operator $\widetilde{Q}_{\dt,n}$ from~$X$ to~$X$ for which the following equality holds for any function $f\in X$:
\begin{equation}
  \label{eq:pseudo_inverse_general_definition}
  (A + \dt^\alpha B)Q_{\dt,n}f = f + \dt^{(n+1)\alpha} \widetilde{Q}_{\dt,n}f. 
\end{equation}
To this end we simply truncate the formal series expansion of the inverse of the operator $A + \dt^\alpha \, B = A(\Id + \dt^\alpha \, A^{-1} B)$, which formally reads $A^{-1} - \dt^\alpha \, A^{-1}B A^{-1} + \dt^{2\alpha} \, A^{-1}B A^{-1}B A^{-1} + \dots$. For instance, $Q_{\dt,1} = A^{-1} - \dt^\alpha \, A^{-1}B A^{-1}$ and $Q_{\dt,2} = A^{-1} - \dt^\alpha \, A^{-1}B A^{-1} + \dt^{2\alpha} \, A^{-1}B A^{-1}B A^{-1}$ indeed are operators from $X$ to~$X$ satisfying~\eqref{eq:pseudo_inverse_general_definition}, respectively with $n=1$ and $n=2$.

\subsection{Proof of Theorem~\ref{thm:error_first_order_schemes}}
\label{sec:proof_thm:error_first_order_schemes}

We write the proof for the scheme associated with $P_\dt^{\gamma C, B, A} = \rme^{\gamma \dt C} \rme^{\dt B} \rme^{\dt A}$, the proof for the scheme $P_\dt^{\gamma C, A, B}$ following the same lines. The results for the other schemes are then obtained with the TU lemma (Lemma~\ref{lem:TU}). Without loss of generality, we perform the proof for a function~$\psi$ with average zero with respect to~$\mu$ (recovering the general case by adding a constant to $\psi$ in the final expression).

\paragraph{Proof of~\eqref{eq:error_first_order_schemes}.}
First, note that, by definition of the invariant measure $\mu_{\gamma,\dt}$, it holds that, for any $\varphi \in \mathcal{S}$,
\begin{equation}
  \label{eq:avg_0_mu_dt}
  \int_\cE \left(\frac{\Id-P_\dt^{\gamma C,B,A}}{\dt}\right) \varphi \, \d{\mu}_{\gamma,\dt} = 0.
\end{equation}
The next step is to choose the correction function~$f_{1,\gamma}$. Using the results of Section~\ref{sec:useful}, a simple computation shows that
\begin{equation}
\label{eq:P_dt_eq_order1}
P_\dt^{\gamma C, B, A} = \Id + \dt \Lgam + \frac{\dt^2}{2} \left(\Lgam^2 + S_1\right) + \dt^3 R_{1,\dt}, \qquad S_1 = [C,A+B] + [B,A],
\end{equation}
where the subscript index~1 refers to the order of the splitting, and where all operators are understood as operators on~$\mathcal{S}$. More precisely, 
\[
R_{1,\dt} = \frac12 \int_0^1 (1-\theta)^2 \mathcal{R}_{\theta \Delta t} \, \d{\theta},
\]
where $\mathcal{R}_s$ is a finite linear combination of terms of the form $C^\gamma \esC B^\beta \esB A^\alpha \esA$ with $\alpha,\beta,\gamma \geq 0$ and $\alpha + \beta + \gamma = 3$. In any case, $R_{1,\dt}$ is a differential operator involving at most 6 derivatives, and with smooth coefficients of at most polynomial growth. It is easily seen that $R_{1,\dt}\psi$ is uniformly bounded in some space~$L^\infty_{\Li_s}$ (with $s$ chosen sufficiently large) for $\dt$ small enough when $\psi \in \mathcal{S}$. Therefore, for any $\varphi \in \mathcal{S}$ and $f_{1,\gamma} \in \widetilde{\mathcal{S}}$, 
\[
\begin{aligned}
& \int_\cE \left[\left(\frac{\Id-P_\dt^{\gamma C,B,A}}{\dt}\right) \varphi\right] (1+\dt f_{1,\gamma}) \, \d{\mu} \\
& \qquad = -\int_\cE \left[\left(\Lgam + \frac{\dt}{2} \left(\Lgam^2 + S_1\right) + \dt^2 R_{1,\dt}\right)\varphi\right] (1+\dt f_{1,\gamma}) \, \d{\mu} \\
& \qquad = -\dt \int_\cE  \left( \frac12 S_1 \varphi + (\Lgam \varphi) f_{1,\gamma} \right)d\mu 
- \dt^2 \int_\cE \left( \left[\frac{1}{2} \left(\Lgam^2 + S_1\right)\varphi\right] f_{1,\gamma} + (R_{1,\dt}\varphi) (1+\dt f_{1,\gamma}) \right) \d{\mu}.
\end{aligned}
\]
The dominant term on the right-hand side can be written, using integration by parts,
\[
\int_\cE \left(\frac12 S_1 \varphi + (\Lgam \varphi) f_{1,\gamma} \right) \d{\mu} = \int_\cE \varphi \left[ \frac12 S_1^* \mathbf{1} + \Lgam^* f_{1,\gamma}\right] \d{\mu}.
\]
In view of~\eqref{eq:avg_0_mu_dt}, we choose the correction function in order to eliminate the dominant term:
\begin{equation}
\label{eq:choice_f1}
\Lgam^* f_{1,\gamma} = -\frac12 S_1^* \mathbf{1}.
\end{equation}
Relying on Theorem~\ref{thm:stability_S} and~\eqref{eq:stability_H1}, the function $f_{1,\gamma}$ is a well defined element from $\widetilde{\mathcal{S}}$ since the right-hand side of~\eqref{eq:choice_f1} belongs to $\widetilde{\mathcal{S}}$. A direct computation using integration by parts indeed shows that $S_1^* \mathbf{1} \in \mathcal{S}$ (see~\eqref{eq:S1_star} below). The centering condition follows from the fact that $\mathbf{1} \in \mathrm{Ker}(S_1)$: indeed,
\[
\int_\cE  S_1^* \mathbf{1} \, \d{\mu} =  \int_\cE S_1 \mathbf{1} \, \d{\mu} = 0.
\]
With the choice~\eqref{eq:choice_f1},
\begin{equation}
\label{eq:first_order_corrected}
\begin{aligned}
& \int_\cE \left[\left(\frac{\Id-P_\dt^{\gamma C,B,A}}{\dt}\right) \varphi\right] (1+\dt f_{1,\gamma}) \, \d{\mu} \\
& \qquad = - \dt^2 \int_\cE \left( \left[\frac{1}{2} \left(\Lgam^2 + S_1\right)\varphi\right] f_{1,\gamma} + (R_{1,\dt}\varphi) (1+\dt f_{1,\gamma}) \right) \d{\mu}.
\end{aligned}
\end{equation}
We would like, at this stage, to replace the observable~$\varphi$ appearing on the left hand side by the function 
\[
\left(\frac{\Id-P_\dt^{\gamma C,B,A}}{\dt}\right)^{-1} \psi.
\] 
However, we do not have any control on the derivatives of this function (Corollary~\ref{corr:resolvent_estimates_I_Pdt} allows to control the norm of the function, not of its derivatives), whereas such a control is required to bound the remainder terms. In order to use an approximate inverse operator involving iterated powers of $\mathcal{L}_\gamma^{-1}$ (see Section~\ref{sec:approx_inv}), we first need to make sure that all operators are defined on~$\widetilde{\mathcal{S}}$, with values in~$\widetilde{\mathcal{S}}$. This is the case for $\mathcal{L}_\gamma$ and its inverse, but not for the other operators appearing in~\eqref{eq:P_dt_eq_order1}, which have values in~$\mathcal{S}$. We therefore project out averages with respect to~$\mu$. Define to this end the projector
\begin{equation}
  \label{eq:def_Pi}
  \Pi^\perp f = f - \int_\mathcal{E} f \, \d \mu,
\end{equation}
which maps $\mathcal{S}$ to~$\widetilde{\mathcal{S}}$. Then, for a function $\varphi \in \widetilde{\mathcal{S}}$ (for which $\Pi^\perp \varphi = \varphi$), \eqref{eq:first_order_corrected} can be rewritten as
\[
\begin{aligned}
& \int_\cE \left[ \Pi^\perp \frac{\Id-P_\dt^{\gamma C,B,A}}{\dt} \Pi^\perp \varphi\right] (1+\dt f_{1,\gamma}) \, \d{\mu} \\
& \qquad = \frac{1}{\dt} \int_\cE P_\dt^{\gamma C,B,A} \varphi \, \d{\mu} - \dt^2 \int_\cE \left( \left[\frac{1}{2} \left(\Lgam^2 + S_1\right)\varphi\right] f_{1,\gamma} + (R_{1,\dt}\varphi) (1+\dt f_{1,\gamma}) \right) \d{\mu},
\end{aligned}
\]
where we have used the fact that $f_{1,\gamma}$ is of average zero with respect to~$\mu$. On the other hand, \eqref{eq:avg_0_mu_dt} may be rewritten
\[
\int_\cE \Pi^\perp \frac{\Id-P_\dt^{\gamma C,B,A}}{\dt} \Pi^\perp \varphi \, \d{\mu}_{\gamma,\dt} = \frac{1}{\dt} \int_\cE P_\dt^{\gamma C,B,A} \varphi \, \d{\mu}.
\]
Therefore,
\begin{equation}
\label{eq:first_order_corrected_bis}
\begin{aligned}
& \int_\cE \left[ \Pi^\perp \frac{\Id-P_\dt^{\gamma C,B,A}}{\dt} \Pi^\perp \varphi\right] (1+\dt f_{1,\gamma}) \, \d{\mu} - \int_\cE \Pi^\perp \frac{\Id-P_\dt^{\gamma C,B,A}}{\dt} \Pi^\perp \varphi \, \d{\mu}_{\gamma,\dt} \\
& \qquad \qquad =  - \dt^2 \int_\cE \left( \left[\frac{1}{2} \left(\Lgam^2 + S_1\right)\varphi\right] f_{1,\gamma} + (R_{1,\dt}\varphi) (1+\dt f_{1,\gamma}) \right) \d{\mu}.
\end{aligned}
\end{equation}
As a consequence of the presence of the projection~$\Pi^\perp$, all of  the operators in~\eqref{eq:P_dt_eq_order1} are restricted to the range of $\Pi^\perp$, \textit{i.e.} the following equality holds on~$\widetilde{\mathcal{S}}$: 
\[
-\Pi^\perp \frac{\Id-P_\dt^{\gamma C, B, A}}{\dt} \Pi^\perp = \Lgam + \frac{\dt}{2} \left(\Lgam^2 + \Pi^\perp S_1 \Pi^\perp \right) + \dt^2 \Pi^\perp R_{1,\dt} \Pi^\perp.
\]
We therefore introduce the operator
\[
Q_{1,\dt} = -\Lgam^{-1} + \frac\dt2 (\Pi^\perp + \Lgam^{-1} \Pi^\perp S_1 \Pi^\perp \Lgam^{-1}),
\] 
which is a well defined operator from $\widetilde{\mathcal{S}}$ to $\widetilde{\mathcal{S}}$ such that
\[
\left(\Pi^\perp \frac{\Id-P_\dt^{\gamma C,B,A}}{\dt} \Pi^\perp \right) Q_{1,\dt} = \Pi^\perp + \dt^2 Z_{1,\dt},
\]
where $Z_{1,\dt}$ maps~$\mathcal{S}$ to~$\mathcal{S}$. We finally replace $\varphi$ by $Q_{1,\dt}\psi$ in~\eqref{eq:first_order_corrected_bis}. This gives (recall that $\Pi^\perp \psi = \psi$ by assumption)
\[
\int_\cE \psi \, (1+\dt f_{1,\gamma}) \, \d{\mu} - \int_\mathcal{E} \psi \, \d{\mu_\dt} = \dt^2 \int_\cE \left[\left(\widetilde{R}_{1,\dt} \psi\right) f_{1,\gamma} + \widehat{R}_{1,\dt} \psi\right] \d{\mu},
\]
where the functions $\widetilde{R}_{1,\dt} \psi, \widehat{R}_{1,\dt} \psi$ belong to $\mathcal{S}$ when $\psi$ does. The integral on the right-hand side is uniformly bounded for small $\dt$ (using the fact that the functions appearing in the integral are in~$\mathcal{S}$ and relying on Proposition~\ref{prop:ergodicity_MC}). This gives~\eqref{eq:error_first_order_schemes} for the splitting scheme $P_\dt^{\gamma C, B, A}$.

\paragraph{Proof of~\eqref{eq:correction_first_order_schemes}.}
The function $f_1^{\gamma C,B,A} \in \widetilde{\mathcal{S}}$ (denoted by $f_{1,\gamma}$ above) is uniquely determined by the equation
\[
\Lgam^* f_1^{\gamma C,B,A} = -\frac12 S_1^* \mathbf{1} = - \frac12 \Big( [C,A+B] + [B,A] \Big)^* \mathbf{1}, \qquad \int_\mathcal{E} f_1^{\gamma C,B,A} \, \d{\mu} = 0,
\]
where we have used $[\Lgam^2]^* \mathbf{1} = 0$ to simplify the right-hand side. Now, $[C,A+B]^* = [C,A+B]$ since $C^* = C$ and $(A+B)^* = -(A+B)$. Therefore, $[C,A+B]^* \mathbf{1} = 0$. In addition, 
\[
[B,A]^* \mathbf{1} = - (A+B)^*g = (A+B)g,
\]
since $A^* = -A + g$ and $B^* = -B -g$. Therefore,
\begin{equation}
\label{eq:S1_star}
S_1^* \mathbf{1} = (A+B)g.
\end{equation}
This gives the first expression in~\eqref{eq:correction_first_order_schemes}.

To obtain the expressions of $f_1^{A,\gamma C,B}$ and $f_1^{B,A,\gamma C}$, we use the TU lemma, where the operators $U_\dt$ respectively read $\rme^{\gamma \dt C}\rme^{\dt B} = \Id + \dt (B+\gamma C) + \dt^2 R_\dt$ and $\rme^{\gamma \dt C}$ (which preserves $\mu$). 
We actually are in a situation similar to~\eqref{eq:interesting_TU_situation}:
\[
f_1^{B,A,\gamma C} = f_1^{\gamma C,B,A}, \qquad f_1^{A,\gamma C,B} = f_1^{\gamma C,B,A} + B^* \mathbf{1}.
\]
The expressions for the first order corrections when the operators $A$ and $B$ are exchanged are obtained by noting that the sign of $S_1^* \mathbf{1}$ is changed and that $f_1^{B,\gamma C,A} = f_1^{\gamma C,A,B} + A^* \mathbf{1}$.

\begin{remark}
\label{rmk:structure_proof}
Let us highlight the structure of the proof, in order to make clear which technical extensions are required in order to state error estimates for other dynamics:
\begin{enumerate}[(i)]
\item first, an expansion of the evolution operator $P_\dt$ in powers of~$\dt$ has to be written out. This step is usually quite simple although sometimes algebraically involved. The expansion of $P_\dt$ is the same as the one used to prove weak error estimates;
  \item second, good control on the resolvent has to be established, such as the stability result provided by Theorem~\ref{thm:stability_S}. This step may already be quite complicated since it involves proving that $\mu$ is the unique invariant measure, and that the resolvent can be inverted for functions with average zero with respect to~$\mu$. A typical way to do so is to establish decay properties of the semigroup. Such decay estimates may be hard to obtain for degenerate noises;
  \item the existence of an invariant measure $\mu_{\dt}$ for the numerical scheme has to be demonstrated (uniqueness is not required), typically by finding a Lyapunov function. Again, this may be difficult if the dynamics is highly degenerate.
\end{enumerate}
Once the above steps have been performed, the correction function can be identified as the solution of a Poisson equation, by comparing the average of $(\Id-P_\dt)\varphi$ under~$\mu$ and~$\mu_\dt$. The remainder of the proof allows one to state error estimates for any smooth function (and not just functions in the range of $\Id - P_\dt$) using appropriate pseudo-inverses.
\end{remark}

\subsection{Proof of Proposition~\ref{prop:Ham_limit_correction}}
\label{sec:proof_Ham_limit}

We use a very standard strategy: first, we propose an ansatz for the correction term $f_{1,\gamma}$ as 
\[
f_{1,\gamma} = f_1^0 + \gamma f_1^1 + \gamma^2 f_1^2 + \dots,
\]
then identify the two leading order terms in this expression, and finally use the resolvent estimate of Theorem~\ref{thm:Ham_limit_Lgam} to conclude. Note that our ansatz is not obvious since the estimate of Theorem~\ref{thm:Ham_limit_Lgam} shows that, in general, a leading order correction term of order $1/\gamma$ should be considered. It turns out however that, due to the specific structure of the right-hand side of~\eqref{eq:correction_first_order_schemes} (namely the fact that the right-hand is at leading order in $\gamma$ the image under the Hamiltonian operator of some function), such a divergent leading order term is not necessary.

Consider for instance the case when $f_{1,\gamma}$ is $f_1^{\gamma C, B, A}$. This function solves
\[
\Big[-(A+B)+\gamma C\Big]f_1^{\gamma C, B, A} = -\frac12 (A+B)g, \qquad \int_\mathcal{E} f_1^{\gamma C, B, A} \, \d{\mu} = 0,
\]
so that we consider the ansatz $f_1^{\gamma C, B, A} = g/2 + \gamma f_1^1 + \dots$. Identifying terms with same powers of~$\gamma$, we see that the correction term $f_1^1$ should satisfy
\[
(A+B)f_1^1 = \frac12 Cg = \frac\beta2 p^T M^{-2}\nabla V.
\]
Possible solutions are defined up to elements of the kernel of $A+B$ (which contains function of the form $\varphi \circ H$). One possible choice is to set $f_1^1 = \beta p^T M^{-2} p/4 + c_1^1$, where the constant $c_1^1$ is chosen in order for $f_1^1$ to have a vanishing average with respect to~$\mu$. Then, 
\[
\Lgam^*\left(f_1^{\gamma C, B, A} -\frac{g}{2}-\gamma f_1^1\right) = \gamma^2 Cf_1^1.
\]
In view of Theorem~\ref{thm:Ham_limit_Lgam}, this implies that there exists a constant $K > 0$ such that
\[
\left\| f_1^{\gamma C, B, A} - \frac{g}{2}-\gamma f_1^1 \right\|_{L^2(\mu)} \leq K\gamma,
\]
for $\gamma \leq 1$, which gives the desired estimate on $f_1^{\gamma C, B, A}$. Similar computations give the estimate on $f_1^{\gamma C, A, B}$, while the estimates on the remaining functions are obtained from~\eqref{eq:correction_first_order_schemes}.

\subsection{Proof of Theorem~\ref{thm:error_second_order_schemes}}
\label{sec:proof_thm:error_second_order_schemes}

The proof follows the same lines as the proof for the first order splitting schemes (see Section~\ref{sec:proof_thm:error_first_order_schemes}). We present only the required modifications. We write the proof for $P_\dt^{\gamma C,B,A,B,\gamma C}$ since the correction term has a much simpler right-hand side than $P_\dt^{A,B,\gamma C,B,A}$.

\paragraph{Proof of~\eqref{eq:error_second_order_schemes}.}
Expanding up to terms of order~$\dt^5$ the formal expression of $P_\dt^{\gamma C,B,A,B,\gamma C}$ given by the BCH expansion~\eqref{eq:BCH}, we obtain the following equality (as operators on~$\mathcal{S}$)
\[
P_\dt^{\gamma C,B,A,B,\gamma C} = \Id + \dt (\Lgam + \dt^2 S_2) + \frac{\dt^2}{2}\Big(\Lgam^2 + \dt^2 (\Lgam S_2 + S_2 \Lgam) \Big) + \frac{\dt^3}{6} \Lgam^3 + \frac{\dt^4}{24} \Lgam^4 + \dt^5 R_{2,\dt},
\]
where 
\[
R_{2,\dt} = \frac{1}{24} \int_0^1 (1-\theta)^4 \mathcal{R}_{\theta \Delta t} \, \d{\theta},
\]
$\mathcal{R}_s$ being a finite linear combination of terms of the form $C^\gamma \esC B^\beta \esB A^\alpha \esA$ with $\alpha,\beta,\gamma \geq 0$ and $\alpha + \beta + \gamma = 5$; and
\begin{equation}
\label{eq:def_S2}
S_2 = \frac{1}{12} \left( S_{2,0} + \gamma S_{2,1} + \gamma^2 S_{2,2}\right),
\end{equation}
with 
\[
\left\{ \begin{aligned}
S_{2,0} & = [A,[A,B]] - \frac12 [B,[B,A]], \\ 
S_{2,1} & = [A+B,[A+B,C]], \\
S_{2,2} & = -\frac12 [C,[C,A+B]].
\end{aligned} \right.
\]
Therefore,
\begin{equation}
\label{eq:I_Pdt_eq_order2}
\frac{\Id-P_\dt^{\gamma C,B,A,B,\gamma C}}{\dt} = -\Lgam - \frac\dt2 \Lgam^2 - \dt^2 \left( \frac16 \Lgam^3 + S_2\right) - \frac{\dt^3}{2} \left(  \frac{1}{12} \Lgam^4 + S_2 \Lgam + \Lgam S_2 \right) - \dt^4 R_{2,\dt}.
\end{equation}
We choose $f_2^{\gamma C,B,A,B,\gamma C} \in \widetilde{\mathcal{S}}$ as the unique solution of the Poisson equation $\Lgam^* f_2^{\gamma C,B,A,B,\gamma C} = -S_2^* \mathbf{1}$ (which is indeed well posed since the right hand side has a vanishing average with respect to~$\mu$ since it is in the image of~$S_2$, and is regular as shown by~\eqref{eq:S2_star} below). Then, for a function $\varphi \in \mathcal{S}$,
\[
\begin{aligned}
& \int_\cE \left(\frac{\Id-P_\dt^{\gamma C,B,A,B,\gamma C}}{\dt}\right)\varphi \, \left(1+\dt^2 f_2^{\gamma C,B,A,B,\gamma C} \right) \d{\mu} = \\
& \qquad \qquad -\frac{\dt^3}{2} \int_\cE S_2 \Lgam \varphi + \left(\Lgam^2 \varphi \right) f_2^{\gamma C,B,A,B,\gamma C} \, \d{\mu} - \dt^4 \int_\cE \left[ \widetilde{R}_{2,\dt} \varphi + \widehat{R}_{2,\dt} \varphi f_2^{\gamma C,B,A,B,\gamma C} \right] \d{\mu},
\end{aligned}
\]
where many terms cancel by the invariance of $\mu$ by $\left(\Lgam^\alpha\right)^*$ (for integer powers $\alpha$). The leading order term on the right-hand side in fact vanishes since it can be rewritten as
\[
\int_\cE S_2 \Lgam \varphi + \Lgam^2 \varphi \, f_2^{\gamma C,B,A,B,\gamma C} \, \d{\mu} = \int_\cE \Lgam \varphi \left( S_2^* \mathbf{1} + \Lgam^*f_2^{\gamma C,B,A,B,\gamma C} \right) \d{\mu} = 0.
\]
Therefore,
\[
\int_\cE \left(\frac{\Id-P_\dt^{\gamma C,B,A,B,\gamma C}}{\dt}\right)\varphi \, \left(1+\dt^2 f_2^{\gamma C,B,A,B,\gamma C} \right) \d{\mu} = - \dt^4 \int_\cE \left[ \widetilde{R}_{2,\dt} \varphi + \widehat{R}_{2,\dt} \varphi f_2^{\gamma C,B,A,B,\gamma C} \right] \d{\mu}.
\]
We then restrict the above equality to functions $\varphi \in \widetilde{\mathcal{S}}$, project out the average with respect to~$\mu$ of the first factor in the integral on the left using the projector~$\Pi^\perp$ introduced in~\eqref{eq:def_Pi}, and finally replace $\varphi$ by $Q_{2,\dt} \psi$ where $Q_{2,\dt}$ is an approximate inverse satisfying
\[
\Pi^\perp \frac{\Id-P_\dt^{\gamma C,B,A,B,\gamma C}}{\dt} \Pi^\perp Q_{2,\dt} = \Pi^\perp + \dt^4 Z_\dt.
\]
The proof is concluded as in Section~\ref{sec:proof_thm:error_first_order_schemes}.

\paragraph{Proof of~\eqref{eq:correction_second_order_schemes}.}

To evaluate the expression  $S_2^* \mathbf{1}$, we need to compute the actions of the formal adjoints of the various commutators. Using $C\mathbf{1} = (A+B)\mathbf{1} = 0$ and
\[
C^* = C, \qquad A^* = -A + g, \qquad B^* = -B-g,
\]
straightforward computations show that $S_{2,2}^* \mathbf{1} = S_{2,1}^*\mathbf{1} = 0$. In addition, since 
\[
A\left(g^2\right) = 2g Ag, \qquad B\left(g^2\right) = 2g \, Bg,
\]
it follows that
\[
\begin{aligned}
\big([A,[A,B]]\big)^*\mathbf{1} & = \left(A^2 B-2ABA + BA^2 \right)^*\mathbf{1} = \left( B^* A^* - 2 A^* B^* - (A^*)^2 \right)g \\
& = \Big( (B+g)(A-g) - 2 (A-g)(B+g) - (A-g)^2 \Big)g \\
& = (BA-2AB-A^2)g = -(A+B)Ag,
\end{aligned}
\]
where we have used $AB g = BA g$ (as can be checked by a direct computation). A similar computation shows that $\big([B,[B,A]]\big)^*\mathbf{1} = (-AB+2BA+B^2)g = (A+B)Bg = ABg$ (since $B^2g = 0$ by a direct verification). Finally, 
\begin{equation}
\label{eq:S2_star}
S_{2}^*\mathbf{1} = -\frac{1}{12} (A+B)\left(A + \frac{B}{2}\right)g. 
\end{equation}
To obtain the expression of $f_{2}^{A,B,\gamma C,B,A}$, we use the TU lemma with the operator 
\[
U_\dt = \rme^{\gamma\dt C/2} \rme^{\dt B/2} \rme^{\dt A/2}.
\]
A simple computation shows that
\[
U_\dt^* \mathbf{1} = \mathbf{1} + \frac{\dt^2}{8} (A+B)g + \dt^3 R_\dt^* \mathbf{1}.
\]
In fact, it can be shown that the $\dt^3$ term does not pollute the remainder since the next order correction in the invariant measure has to be of order $\dt^4$ (see~\eqref{eq:error_second_order_schemes}). The expressions for $f_{2}^{\gamma C,A,B,A,\gamma C}$ and $f_{2}^{B,A,\gamma C,A,B}$ are obtained in a similar manner.

\subsection{Proof of Corollary~\ref{cor:error_GLA}}
\label{sec:proof_cor:error_GLA}

The proof relies on the results of Theorem~\ref{thm:error_second_order_schemes} and the TU lemma (Lemma~\ref{lem:TU}). More precisely, the error estimate~\eqref{eq:error_GLA_general} is established by following the same lines of proof as for second order splitting schemes, except that the contributions of order~$\dt^3$ do not vanish. We then use the TU lemma by considering the GLA evolution as the reference, and express the invariant measure of second order splitting schemes in terms of the invariant measure of the GLA scheme. For instance, consider $P_\dt^{\gamma C,B,A,B}$ and $P_\dt^{\gamma C,B,A,B,\gamma C}$, in which case $U_\dt = \rme^{\gamma \dt C/2}$. Then,
\[
\begin{aligned}
&\int_\cE \psi \, \d{\mu}_\dt^{\gamma C,B,A,B,\gamma C} = \int_\cE \left( U_\dt \psi \right)  \d{\mu}_\dt^{\gamma C,B,A,B} \\
& = \int_\cE U_\dt \psi \, \d{\mu} + \dt^2 \int_\cE (U_\dt \psi) f_2^{\gamma C,B,A,B} \, \d{\mu}+ \dt^3 \int_\cE (U_\dt \psi) f_3^{\gamma C,B,A,B} \, \d{\mu} + \dt^4 r_{\psi,\gamma,\dt} \\
& = \int_\cE \psi \, \d{\mu} + \dt^2 \int_\cE \psi \, f_2^{\gamma C,B,A,B} \, \d{\mu}+ \dt^3 \int_\cE \psi \left(f_3^{\gamma C,B,A,B}+\frac{\gamma}{2}Cf_2^{\gamma C,B,A,B}\right) \d{\mu} + \dt^4 \widetilde{r}_{\psi,\gamma,\dt}, 
\end{aligned}
\]
where we have used the invariance of $\mu$ by $U_\dt$. The comparison with~\eqref{eq:error_second_order_schemes}-\eqref{eq:correction_second_order_schemes} gives the desired result.

\subsection{Approximation of integrated correlation functions}
\label{sec:proof_approx_GK_formula}

\begin{proof}[Proof of Theorem~\ref{thm:approx_GK_formula}]
The proof makes use of the projection operator defined on $\mathcal{S}$ as (compare~\eqref{eq:def_Pi})
\[
\Pi^\perp_\dt \varphi = \varphi - \int_\mathcal{E} \varphi \, \d{\mu}_\dt.
\]
The range of $\Pi^\perp_\dt$ is contained in the set of functions with average zero with respect to the invariant measure~$\mu_\dt$ of the numerical scheme. We first introduce the invariant measure for the numerical scheme, using the fact that $-\cL_\gamma^{-1} \psi$ has zero average with respect to~$\mu$:
\begin{align}
\int_\cE \left(-\cL_\gamma^{-1} \psi\right) \varphi \, \d{\mu} & = \int_\cE \left(-\cL_\gamma^{-1} \psi\right) \Pi^\perp_\dt \varphi \, \d{\mu} \nonumber \\
& = \int_\cE \left(-\cL_\gamma^{-1} \psi\right) \Pi^\perp_\dt \varphi \, \d{\mu}_\dt + \dt^\alpha r^{\psi,\varphi}_\dt, \nonumber \\
& = \int_\cE \Pi^\perp_\dt \left(-\cL_\gamma^{-1} \psi\right) \Pi^\perp_\dt \varphi \, \d{\mu}_\dt + \dt^\alpha r^{\psi,\varphi}_\dt, \label{eq:introduce_mu_dt_in_correlation}
\end{align}
where $r^{\psi,\varphi}_\dt$ is uniformly bounded for $\dt$ sufficiently small by~\eqref{eq:asssumption_order_method}. In addition, by~\eqref{eq:expansion_I_Pdt}, 
\[
\begin{aligned}
-\Pi^\perp_\dt \cL_\gamma^{-1} \psi & = -\Pi^\perp_\dt \left(\dt\sum_{n=0}^{+\infty} P_\dt^n \right) \Pi^\perp_\dt \left(\frac{\Id - P_\dt}{\dt}\right) \Lgam^{-1} \psi \\
& = \dt \left(\sum_{n=0}^{+\infty} \left[ \Pi^\perp_\dt P_\dt \Pi^\perp_\dt\right]^n \right) \left(\cL_\gamma + \dt S_1 + \dots + \dt^{\alpha-1} S_{\alpha-1} + \dt^\alpha \widetilde{R}_{\alpha,\dt}\right) \Lgam^{-1} \psi, \\
& = \dt \sum_{n=0}^{+\infty} \left[ \Pi^\perp_\dt P_\dt \Pi^\perp_\dt\right]^n \widetilde{\psi}_{\dt,\alpha} + \dt^\alpha \left(\frac{\Id - P_\dt}{\dt}\right)^{-1} \Pi^\perp_\dt \widetilde{R}_{\alpha,\dt} \Lgam^{-1} \psi.
\end{aligned}
\]
Note that the sum on the right hand side is well defined in view of the decay estimates~\eqref{eq:ergodicity_num}. Plugging the above equality in~\eqref{eq:introduce_mu_dt_in_correlation} leads to
\[
\begin{aligned}
\int_\cE \left(-\cL_\gamma^{-1} \psi\right) \varphi \, \d{\mu} & = \dt \int_\mathcal{E} \sum_{n=0}^{+\infty} \left( \Pi_\dt^\perp P_\dt^n \widetilde{\psi}_{\dt,\alpha} \right) \left( \Pi^\perp_\dt \varphi \right) \d{\mu}_\dt \\ 
& \qquad +  \dt^\alpha \int_\mathcal{E} \left( \left(\frac{\Id - P_\dt}{\dt} \right)^{-1} \Pi^\perp_\dt \widetilde{R}_{\alpha,\dt} \Lgam^{-1} \psi\right)\Pi^\perp_\dt \varphi \, \d{\mu}_\dt + \dt^\alpha r^{\psi,\varphi}_\dt.
\end{aligned}
\]
The second term on the right hand side is uniformly bounded in view of the moment estimates on~$\mu_\dt$, the resolvent bounds provided by Corollary~\ref{corr:resolvent_estimates_I_Pdt} and the uniform boundedness of the remainder $\widetilde{R}_{\alpha,\dt}f$ for a given, smooth function~$f$. Since
\[
\int_\mathcal{E} \sum_{n=0}^{+\infty} \left( \Pi_\dt^\perp P_\dt^n \widetilde{\psi}_{\dt,\alpha} \right)\left( \Pi^\perp_\dt \varphi \right) \d{\mu}_\dt = \int_\mathcal{E} \sum_{n=0}^{+\infty} \left(P_\dt^n \widetilde{\psi}_{\dt,\alpha} \right) \varphi \, \d{\mu}_\dt = \sum_{n=0}^{+\infty} \mathbb{E}_\dt \left(\widetilde{\psi}_{\dt,\alpha}\left(q^{n},p^{n}\right)\varphi\left(q^0,p^0\right)\right),
\]
\eqref{eq:correction_GK_general} finally follows.
\end{proof}

\bigskip

\begin{proof}[Proof of Corollary~\ref{cor:GK_trapezoidal}]
The idea is to start from~\eqref{eq:correction_GK_general} and to appropriately rewrite the first order correction term. We use to this end~\eqref{eq:correction_GK_general} with $\psi$ replaced by its first order correction $(\psi_{\dt,2} - \psi)/\dt = S_1\cL_\gamma^{-1}\psi$, and discard terms of order~$\dt^2$:
\[
\int_0^{+\infty} \mathbb{E} \Big( S_1\cL_\gamma^{-1}\psi(q_t,p_t) \varphi(q_0,p_0) \Big) \d{t} = \dt \sum_{n=0}^{+\infty} \mathbb{E}_\dt \Big(S_1\cL_\gamma^{-1}\psi\left(q^{n+1},p^{n+1}\right)\varphi_{\dt,0}\left(q^0,p^0\right)\Big) + \dt \,r^{\psi,\varphi}_\dt, 
\]
where $r^{\psi,\varphi}_\dt$ is uniformly bounded for $\dt$ sufficiently small and $\varphi_{\dt,0} = \Pi_\dt^\perp \varphi$.
On the other hand, using $S_1 = \mathcal{L}_\gamma^2/2$,
\[
\int_0^{+\infty} \mathbb{E} \Big( S_1\cL_\gamma^{-1}\psi(q_t,p_t) \varphi(q_0,p_0) \Big) \d{t} = -\int_\cE \cL_\gamma^{-1}S_1\cL_\gamma^{-1} \psi \, \varphi \, \d{\mu} = -\frac12 \int_\cE \psi \varphi \, \d{\mu},
\]
so that, 
\[
\begin{aligned}
& \dt \sum_{n=0}^{+\infty} \mathbb{E}_\dt \Big( \left(S_1\cL_\gamma^{-1}\psi\right)_{\dt,0}\left(q^{n+1},p^{n+1}\right)\varphi\left(q^0,p^0\right)\Big) \\
& \qquad = \dt \sum_{n=0}^{+\infty} \mathbb{E}_\dt \Big(S_1\cL_\gamma^{-1}\psi\left(q^{n+1},p^{n+1}\right)\Pi_\dt^\perp\varphi\left(q^0,p^0\right)\Big) \\
& \qquad = \int_0^{+\infty} \mathbb{E} \Big( S_1\cL_\gamma^{-1}\psi(q_t,p_t) \Pi_\dt^\perp\varphi(q_0,p_0) \Big) \d{t} - \dt \,r^{\psi,\varphi}_\dt \\
& \qquad = -\frac12 \int_\cE \psi \, \Pi_\dt^\perp\varphi \, \d{\mu} - \dt \,r^{\psi,\varphi}_\dt = -\frac12 \int_\cE \psi_{\dt,0} \varphi \, \d{\mu} - \dt \,r^{\psi,\varphi}_\dt \\
& \qquad = -\frac12 \mathbb{E}_\dt( \psi_{\dt,0} \varphi) + \dt \, \widetilde{r}^{\psi,\varphi}_\dt.
\end{aligned}
\]
This gives~\eqref{eq:correction_GK_second}.
\end{proof}

\subsection{Proof of Theorem~\ref{thm:ovd_limit}}
\label{sec:proof_thm:ovd_limit}

We write the proof for the evolution operator $P_\dt^{\gamma C, A,B,A,\gamma C}$ first, and mention then how to extend the result to $P_\dt^{B,A,\gamma C,A,B}$ using the TU lemma. The proofs for $P_\dt^{\gamma C,B,A,B,\gamma C}$ and $P_\dt^{A,B,\gamma C,B,A}$ are very similar, so we only briefly mention the required modifications. By default, all operators appearing in this section are defined on the core~$\mathcal{S}$.

\paragraph{Reduction to a limiting operator up to exponentially small terms.}

Let us introduce the evolution operator corresponding to the standard position Verlet scheme: $P_{{\rm ham},\dt} = \rme^{\dt A/2} \rme^{\dt B} \rme^{\dt A/2}$, so that $P_\dt^{\gamma C, A,B,A\gamma C} = \rme^{\gamma \dt C/2} P_{{\rm ham},\dt} \rm\, e^{\gamma \dt C/2}$. On the other hand, we have the following convergence result, whose proof is postponed to the end of the section.

\begin{lemma}
\label{lem:cv_etC}
Fix $s^* \in \mathbb{N}^*$. Then, there exist $K,\alpha > 0$ such that, for any $1 \leq s \leq s^*$ and any $t \geq 0$,
\[
\left\| \rme^{\gamma t C} - \pi \right\|_{\mathcal{B}(L^\infty_{\Li_s})} \leq K \rme^{-\alpha \gamma t}.
\]
\end{lemma}

This suggests to consider the limiting operator $P_{\infty,\dt} = \pi P_{{\rm ham},\dt} \pi$ and write
\begin{equation}
\label{eq:distance_to_limiting_operator}
P_\dt^{\gamma C, A,B,A,\gamma C} - P_{\infty,\dt} = \Big( \rme^{\gamma \dt C/2} - \pi\Big) P_{{\rm ham},\dt} \pi + \rme^{\gamma \dt C/2} P_{{\rm ham},\dt}  \Big( \rme^{\gamma \dt C/2} - \pi \Big).
\end{equation}
For a given smooth function $\varphi \in \mathcal{S}$ which depends only on the position variable~$q$,
\begin{equation}
\label{eq:ovd_diff_first_term}
\int_\cE \left(\Id - P_\dt^{\gamma C, A,B,A,\gamma C}\right)\varphi \, \d{\mu}_{\gamma,\dt} = 0 = 
\int_\cE (\Id - P_{\infty,\dt} )\varphi \, \d{\mu}_{\gamma,\dt} + r_{\varphi,\gamma,\dt}^1,
\end{equation}
with the remainder
\[
r_{\varphi,\gamma,\dt}^1 = \int_\cE  \left( P_{\infty,\dt}-P_\dt^{\gamma C, A,B,A,\gamma C}\right)\varphi \, \d{\mu}_{\gamma,\dt}.
\]
On the other hand,
\begin{equation}
\label{eq:ovd_diff_second_term}
\int_\cE \left[\left(\Id - P_\dt^{\gamma C, B,A,B,\gamma C}\right)\varphi \right] (1 + \dt^2 f_{2,\infty}) \, \d{\mu} = \int_\cE \left[(\Id - P_{\infty,\dt})\varphi \right] (1 + \dt^2 f_{2,\infty}) \, \d{\mu} + r_{\varphi,\gamma,\dt}^2,
\end{equation}
with the remainder
\[
r_{\varphi,\gamma,\dt}^2 = \int_\cE \left[\left(P_{\infty,\dt} - P_\dt^{\gamma C, B,A,B,\gamma C}\right)\varphi \right] (1 + \dt^2 f_{2,\infty}) \, \d{\mu}.
\]
The idea is that the remainders $r_{\varphi,\gamma,\dt}^1$ and $r_{\varphi,\gamma,\dt}^2$ are exponentially small when the function $\varphi$ is sufficiently smooth (see below for a more precise discussion, once $\varphi$ has been replaced by $Q_\dt \psi$ with $Q_\dt$ an appropriate approximate inverse). Therefore, the leading order terms in the error estimate are obtained by considering the limiting operator $P_{\infty,\dt}$ only.

\paragraph{Error estimates for the limiting operator $P_{\infty,\dt}$.}
We now study the error estimates associated with $P_{\infty,\dt}$, following the strategy used in Section~\ref{sec:proof_thm:error_first_order_schemes}. We first use the results of Section~\ref{sec:expansion_evolution} with $M=3$, $A_1 = A_3 = A/2$ and $A_2 = B$ to expand $P_{{\rm ham},\dt}$, so that 
\begin{equation}
\label{eq:dvpmt_P_infty}
P_{\infty,\dt} = \pi + \dt \pi(A+B)\pi + \frac{\dt^2}{2} \pi (A+B)^2\pi + \frac{\dt^3}{6} \pi S_3 \pi + \frac{\dt^4}{24} \pi S_4 \pi + \frac{\dt^5}{120} \pi S_5 \pi + \dt^6 \pi R_\dt \pi,  
\end{equation}
with $S_i = \cT[(A_1+A_2+A_3)^i]$.
To give more precise expressions of the operators appearing on the right-hand side of the above equality, we use the following facts:
\begin{equation}
\label{eq:rules_ovd_1}
\forall n \in \mathbb{N}, \qquad B^n \pi = 0, \qquad \pi A^{2n+1} \pi = 0, 
\end{equation}
and 
\begin{equation}
\label{eq:rules_ovd_2}
\forall n \geq m+1, \qquad B^n A^m \pi = 0.
\end{equation}
In addition,
\[
\pi A^2 \pi = \frac1\beta \Delta_q \pi,
\qquad
BA\pi = - \nabla V \cdot \nabla_q \pi.
\]
Using these rules in~\eqref{eq:dvpmt_P_infty} leads to
\begin{equation}
\label{eq:rules_A+B_carre}
\pi(A+B)\pi = 0, \qquad \pi(A+B)^2\pi = \pi(A^2+BA)\pi = \Lovd \pi.
\end{equation}
The operator $S_3$ is a combination of terms of the form $A^a B^b A^c$ with $a+b+c=3$ and $a,b,c \in \mathbb{N}$. In view of~\eqref{eq:rules_ovd_1}-\eqref{eq:rules_ovd_2}, only the terms with $c \geq 1$ and $b \leq c$ have to be considered, so that only $BA^2$ and $ABA$ remain. A simple computation shows that $BA^2\pi \varphi$ and $ABA\pi\varphi$ are functions linear in~$p$, so that $\pi BA^2 \pi = \pi ABA\pi\varphi = 0$. Finally, $\pi S_3 \pi = 0$. A similar reasoning shows that $\pi S_5 \pi = 0$ and that many terms appearing in the expression of $S_4$ also disappear. 

Plugging the above results in~\eqref{eq:dvpmt_P_infty} and introducing $h=\dt^2/2$,
\[
P_{\infty,\dt} = \pi + h \pi \Lovd \pi + \frac{h^2}{6}\pi \left(A^4 + \frac32 A^2 BA + \frac32 ABA^2 + \frac32 B^2 A^2 + \frac12 BA^3 \right)\pi+ h^3 R_{\infty,\dt}.
\]
Using
\begin{equation}
\label{eq:rules_A4_BA3_etc}
\begin{aligned}
\pi A^4 \pi \varphi & = \frac{3}{\beta^2} \Delta_q^2 \pi \varphi = 3 \left(\pi A^2 \pi\right)^2 \varphi, \\
\pi B A^3 \pi \varphi & = -\frac3\beta \nabla V \cdot \nabla_q \Big(\Delta_q \pi \varphi \Big) = 3 \pi BA \pi A^2 \pi \varphi, \\
\pi B^2 A^2 \pi \varphi & = 2 (\nabla V)^T \Big(\nabla_q^2 \pi\varphi\Big) \nabla V, \\
\pi ABA^2 \pi \varphi & = -\frac2\beta \left(\nabla^2 V : \nabla^2 \varphi + \nabla V \cdot \nabla (\Delta \varphi)\right), \\
\pi A^2BA \pi \varphi & = -\frac1\beta \left(2\nabla^2 V : \nabla^2 \varphi + \nabla V \cdot \nabla (\Delta \varphi) + \nabla (\Delta V) \cdot \nabla \varphi \right) = \pi A^2 \pi BA \pi \varphi, \\
\end{aligned} 
\end{equation}
it follows
\[
\begin{aligned}
& \left(A^4 + \frac32 A^2 BA + \frac32 ABA^2 + \frac32 B^2 A^2 + \frac12 BA^3 \right)\pi \varphi \\
& \qquad = \frac{3}{\beta^2} \Delta_q^2\varphi - \frac{6}{\beta} \nabla^2 V : \nabla^2 \varphi - \frac{6}{\beta}\nabla V \cdot \nabla (\Delta \varphi) - \frac{3}{2\beta} \nabla (\Delta V) \cdot \nabla \varphi + 3 (\nabla V)^T (\nabla^2 \varphi)\nabla V.
\end{aligned}
\]
A straightforward computation shows that
\[
\mathcal{L}_{\rm ovd}^2 \varphi = \frac{1}{\beta^2} \Delta_q^2 \varphi - \frac{2}{\beta} \nabla^2 V : \nabla^2 \varphi - \frac{2}{\beta}\nabla V \cdot \nabla (\Delta \varphi) - \frac{1}{\beta} \nabla (\Delta V) \cdot \nabla \varphi + (\nabla V)^T (\nabla^2 \varphi)\nabla V + (\nabla V)^T (\nabla^2 V) \nabla \varphi.
\]
Therefore, 
\[
\pi\left(A^4 + \frac32 A^2 BA + \frac32 ABA^2 + \frac32 B^2 A^2 + \frac12 BA^3 \right)\pi = 3 \left(\mathcal{L}_{\rm ovd}^2 + D\right)\pi,
\]
with 
\begin{equation}
\label{eq:def_D}
D \varphi = \frac{1}{2\beta} \nabla (\Delta V) \cdot \nabla \varphi - (\nabla V)^T (\nabla^2 V) \nabla \varphi.
\end{equation}
In conclusion,
\begin{equation}
\label{eq:final_expansion_dt_inf}
  P_{\infty,\dt} = \pi + h \mathcal{L}_{\rm ovd} + \frac{h^2}{2} \left(\mathcal{L}_{\rm ovd}^2 + D\right)\pi + h^3 R_{\infty,\dt}.
\end{equation}
Let us emphasize that this operator acts on functions of~$q$ (we define it on $\mathcal{S} \cap \mathrm{Ker}(\pi) = C^\infty(\mathcal{M})$), that $\pi$ is the identity operator for functions which are independent of~$p$, and note that for any $\phi\in C^\infty(\mathcal{M})$,
\begin{equation}
\label{eq:expansion_P_inf_dt}
\frac{\pi-P_{\infty,\dt}}{h} \phi = -\Lovd \phi - \frac{h}{2} \left(\Lovd^2 +D \right)\phi -h^2 R_\dt \phi.
\end{equation}
In fact, proceeding as in Section~\ref{sec:proof_thm:error_first_order_schemes}, we project out averages with respect to~$\overline{\mu}(\d q)$ in order to properly define approximate inverses. Introduce to this end the projector 
\[
\overline{\Pi}^\perp \phi = \phi - \int_\mathcal{M} \phi(q) \, \overline{\mu}(\d q)
\]
defined on the core~$C^\infty(\mathcal{M})$. The equality~\eqref{eq:expansion_P_inf_dt} then implies the following equality on $C^\infty(\mathcal{M}) \cap \mathrm{Ran}(\overline{\Pi}^\perp)$:
\[
\overline{\Pi}^\perp \frac{\pi-P_{\infty,\dt}}{h} \overline{\Pi}^\perp = -\Lovd - \frac{h}{2} \left(\Lovd^2 + \overline{\Pi}^\perp D \overline{\Pi}^\perp \right) -h^2 \overline{\Pi}^\perp R_\dt \overline{\Pi}^\perp.
\]
An approximate inverse of the operator appearing on the left hand side of the above equality is thus
\[
Q_h = -\Lovd^{-1} + \frac{h}{2} \left( \overline{\Pi}^\perp + \Lovd^{-1} \overline{\Pi}^\perp D \overline{\Pi}^\perp \Lovd^{-1} \right).
\]
Denote by $\overline{\mu}_{\infty,\dt}(\d{q})$ the invariant measure of the Markov chain generated by the limiting method $P_{\infty,\dt}$. Proceeding as in Section~\ref{sec:proof_thm:error_first_order_schemes} by first identifying the leading order correction $f_{2,\infty}$, projecting out averages with respect to $\overline{\mu}(\d q)$ using $\overline{\Pi}^\perp$, and replacing $\overline{\Pi}^\perp \varphi$ by $Q_h \psi$, the equality~\eqref{eq:final_expansion_dt_inf} allows us to obtain
\begin{equation}
\label{eq:error_estimate_limiting_operator}
\int_\mathcal{M} \psi(q) \, \overline{\mu}_{\infty,\dt}(\d{q}) = \int_\mathcal{M} \psi(q) \, \overline{\mu}(\d{q}) + \dt^2 \int_\mathcal{M} \psi(q) f_{2,\infty}(q) \, \overline{\mu}(\d{q}) + \dt^4 \overline{r}_{\dt,\psi},
\end{equation}
where $f_{2,\infty}$ is the unique solution of 
\begin{equation}
\label{eq:def_f2infty}
\mathcal{L}_{\rm ovd} f_{2,\infty} = -\frac14 D^* \mathbf{1}. 
\end{equation}
A more explicit expression can be obtained by noting that
\[
D\varphi = \frac12 \nabla \left(\frac1\beta \Delta V - |\nabla V|^2\right) \cdot \nabla \varphi,
\]
so that (recalling $\mathcal{L}_{\rm ovd} = - \beta^{-1} \nabla^* \nabla = -\beta^{-1} \sum_{i=1}^{dN} \partial_{q_i}^* \partial_{q_i}$ where the formal adjoints are taken on $L^2(\overline{\mu})$)
\[
\begin{aligned}
\int_\mathcal{M} \varphi \left(D^*\mathbf{1}\right) \, \d{\overline{\mu}} & = \int_\mathcal{M} D\varphi \, \d{\overline{\mu}} = \frac12 \int_\mathcal{M} \varphi \nabla^* \nabla \left(\frac1\beta \Delta V - |\nabla V|^2\right) \, \d{\overline{\mu}} \\
& = -\frac{1}{2} \int_\mathcal{M} \varphi \, \mathcal{L}_{\rm ovd}\left(\Delta V - \beta |\nabla V|^2\right) \, \d{\overline{\mu}}.
\end{aligned}
\]
Since $f_{2,\infty}$ should have a vanishing average with respect to~$\mu$, this proves that 
\begin{equation}
  \label{eq:f_2_infty}
  f_{2,\infty}(q) = \frac18 \left( \Delta V - \beta |\nabla V|^2\right) + a,
\end{equation}
where the constant~$a$ is adjusted to account for the constraint of vanishing average. A simple computation shows that it is equal to the constant $a_{\beta,V}$ defined in~\eqref{eq:correction_overdamped}.

In fact, it is possible for the scheme considered here to precisely determine the leading order correction for numerical averages by noting that
\begin{equation}
\label{eq:magic_identity_Laplacian}
\frac{1}{\beta} \int_\mathcal{M} \Delta \varphi \, \d{\overline{\mu}} = -\int_\mathcal{M} \varphi \left( \Delta V - \beta |\nabla V|^2\right) \, \d{\overline{\mu}},
\end{equation}
so that finally
\[
\int_\mathcal{M} \psi(q) \, \overline{\mu}_{\infty,\dt}(\d{q}) = \int_\mathcal{M} \psi(q) \, \overline{\mu}(\d{q}) - \frac{\dt^2}{8\beta} \int_\mathcal{M} \Delta \psi(q) \, \overline{\mu}(\d{q}) + \dt^4 r_{\dt,\psi}.
\]

\paragraph{Conclusion of the proof.}
We now come back to~\eqref{eq:ovd_diff_first_term}-\eqref{eq:ovd_diff_second_term} and replace $\overline{\Pi}^\perp \varphi$ by $Q_h \psi$:
\begin{equation}
\label{eq:final_error_estimate_ovd}
\int_\cE \psi \, \d{\mu}_{\gamma,\dt} = \int_\cE \psi (1 + \dt^2 f_{2,\infty}) \, \d{\mu} + r^1_{\psi,\gamma,\dt} + r^2_{\psi,\gamma,\dt} + \dt^4 \overline{r}_{\dt,\psi},
\end{equation}
where $\overline{r}_{\dt,\psi}$ is the same as in~\eqref{eq:error_estimate_limiting_operator}, while
\[
\begin{aligned}
r^1_{\psi,\gamma,\dt} = \int_\cE  \left( P_{\infty,\dt}-P_\dt^{\gamma C, A,B,A,\gamma C}\right)Q_h\psi \, \d{\mu}_{\gamma,\dt}, \\
r^2_{\psi,\gamma,\dt} = \int_\cE \left[\left(P_{\infty,\dt} - P_\dt^{\gamma C, B,A,B,\gamma C}\right)Q_h \psi \right] (1 + \dt^2 f_{2,\infty}) \, \d{\mu}. \\
\end{aligned}
\]
We then integrate with respect to momenta in~\eqref{eq:final_error_estimate_ovd}, and bound the remainders by $K \rme^{-\kappaK \gamma \dt}$ in view of the decomposition~\eqref{eq:distance_to_limiting_operator} and Lemma~\ref{lem:cv_etC} (the operators $P_{{\rm ham},\dt}$ and $\rme^{\gamma \dt C/2}$ being bounded on $L^\infty_{\Li_s}$ uniformly in~$\dt$).

\paragraph{Proof of~\eqref{eq:correction_overdamped} for $f^{B,A,\gamma C,A,B}_{2,\infty}$}
We set \[
U_{\gamma,\dt} = \rme^{\gamma\dt C/2} \rme^{\dt A/2} \rme^{\dt B/2}, 
\qquad 
T_{\gamma,\dt} = \rme^{\dt B/2} \rme^{\dt A/2} \rme^{\gamma\dt C/2},
\]
so that $P_\dt^{B,A,\gamma C,A,B} = T_{\gamma,\dt}U_{\gamma,\dt}$ while $P_\dt^{\gamma C,A,B,A,\gamma C} = U_{\gamma,\dt} T_{\gamma,\dt}$. By the TU lemma,
\begin{align}
\int_\cE \psi \, \d{\mu}_{\dt}^{B,A,\gamma C,A,B} & = \int_\cE \left(U_{\gamma,\dt}\psi\right)d\mu_\dt^{\gamma C,A,B,A,\gamma C} \nonumber \\
& = \int_\cE \left(U_{\infty,\dt}\psi\right)d\mu_\dt^{\gamma C,A,B,A,\gamma C} + \int_\cE \left(U_{\gamma,\dt}-U_{\infty,\dt}\right)\psi \, \d{\mu}_\dt^{\gamma C,A,B,A,\gamma C}, \label{eqn_rhsbacab}
\end{align}
where we have introduced $U_{\infty,\dt} = \pi \rme^{\dt A/2} \rme^{\dt B/2}$. The second term on the right hand side can be bounded by $K \rme^{-\kappaK \gamma \dt}$ in view of Lemma~\ref{lem:cv_etC} and the moment estimate~\eqref{eq:moment_estimate}. For the first term in the right-hand side of \eqref{eqn_rhsbacab}, we use~\eqref{eq:final_error_estimate_ovd} and the following expansion (using the rules~\eqref{eq:rules_ovd_1}-\eqref{eq:rules_ovd_2}): for $\psi \in \mathcal{S}$,
\[
U_{\infty,\dt} \psi = U_{\infty,\dt} \pi \psi = \psi + \frac{\dt^2}{8} \pi A^2 \pi \psi + \dt^4 \widetilde{r}_{\psi,\dt} = \psi + \frac{\dt^2}{8\beta} \Delta \psi + \dt^4 \widetilde{r}_{\psi,\dt}, 
\]
where the remainder $\widetilde{r}_{\psi,\dt}$ is uniformly bounded for $\dt$ sufficiently small. Therefore,
\[
\int_\cE \left(U_{\infty,\dt}\psi\right)d\mu_\dt^{\gamma C,A,B,A,\gamma C} = \int_\cE \psi (1 + \dt^2 f_{2,\infty}) \, \d{\mu} + \frac{\dt^2}{8\beta} \int_\cE \Delta \psi \, \d{\mu} + \widehat{r}_{\psi,\gamma,\dt},
\]
where $f_{2,\infty}$ is given in~\eqref{eq:f_2_infty}. The remainder $\widehat{r}_{\psi,\gamma,\dt}$ is the sum of terms of order~$\dt^4$ and others which can be bounded by $K \rme^{-\kappaK \gamma \dt}$. We conclude by resorting to~\eqref{eq:magic_identity_Laplacian} to compute the formal adjoint of the operator $\Delta_q$ on $L^2(\mu)$.

\paragraph{Proof of~\eqref{eq:correction_overdamped} for $f^{\gamma C,B,A,B,\gamma C}_{2,\infty}$ and $f^{A,B,\gamma C,B,A}_{2,\infty}$.}
We mimic the above proof for the evolution operator $P_\dt^{\gamma C,B,A,B,\gamma C}$. The equality~\eqref{eq:dvpmt_P_infty} still holds, but the operator $S_4$ now reads
\[
S_4 = A^4 + 2 BA^2 + \frac32 B^2A^2,
\]
so that
\[
D\varphi = \frac2\beta \nabla^2 V : \nabla^2\varphi + \frac1\beta \nabla(\Delta V)\cdot \nabla\varphi - \nabla V^T (\nabla^2 V) \nabla\varphi.
\]
A simple computation shows that
\[
\int_\mathcal{M} D\varphi \, \d{\overline{\mu}} = -\frac1\beta \int_\mathcal{M} \nabla\left(\Delta V - \frac\beta2 |\nabla V|^2 \right) \cdot \nabla \varphi \, \d{\overline{\mu}} = \int_\mathcal{M} \Lovd\left(\Delta V - \frac\beta2 |\nabla V|^2 \right) \varphi \, \d{\overline{\mu}},
\]
so that, in view of~\eqref{eq:def_f2infty},
\[
f^{\gamma C,B,A,B,\gamma C}_{2,\infty} = -\frac14 \left( \Delta V - \frac\beta2 |\nabla V|^2 -\frac{a_{\beta,V}}{2} \right).
\]
The expression of $f^{A,B,\gamma C,B,A}_{2,\infty}$ is obtained via the TU lemma, introducing the limiting operator
\[
U_{\infty,\dt}\pi = \pi \rme^{\dt B/2} \rme^{\dt A/2}\pi = \pi + \frac{\dt^2}{8} \pi (A^2 + 2BA)\pi + \dt^4 R_\dt,
\]
so that
\[
f^{A,B,\gamma C,B,A}_{2,\infty} = f^{\gamma C,B,A,B,\gamma C}_{2,\infty} + \frac18 \Big( \pi (A^2 + 2BA)\pi \Big)^* \mathbf{1} = f^{\gamma C,B,A,B,\gamma C}_{2,\infty} + \frac18 \Big( \pi BA \pi \Big)^* \mathbf{1} = -\frac18 \left( \Delta V - a_{\beta,V} \right).
\]

Let us conclude this section with the proof of Lemma~\ref{lem:cv_etC}.

\begin{proof}[Proof of Lemma~\ref{lem:cv_etC}]
The conclusion follows for instance by an application of Theorem~8.7 in~\cite{rey-bellet}, considering as a reference dynamics the Ornstein-Uhlenbeck process
\[
\d{p}_t = -M^{-1} p_t \, \d{t} + \sqrt{\frac{2\gamma}{\beta}} \, \d{W}_t 
\]
with generator~$C$ defined on functions of~$\mathcal{S}$ which are independent of~$q$ (recall that the unique invariant probability measure of this process is $\kappa(\d{p})$). To apply the theorem, we need to show that $\Li_s$ is a Lyapunov function for any~$s \geq 1$. We compute 
\[
C \Li_s = \left(-2s p^T p + \frac{2s(dN+2s-2)}{\beta} \right)|p|^{2(s-1)}
\leq - \Li_s + b_s
\]
for an appropriate constant $b_s \geq 0$. This shows the existence of constants $R_s, \alpha_s$ such that 
\[
\left| \left(\rme^{t C}f\right)(p) - \int_{\mathbb{R}^{dN}} f(p) \, \kappa(\d{p}) \right| \leq R_s \rme^{-\alpha_s t} \| f\|_{L^\infty_{\Li_s}(\d{p})} \Li_s(p), 
\]
where the notation $L^\infty_{\Li_s}(\d{p})$ emphasizes that the supremum is taken over a function of the momentum variable only. The desired result now follows by applying the above bound to the function $\psi(q,\cdot)$ for any element $\psi \in L^\infty_{\Li_s}$, and taking the supremum over~$q$.
\end{proof}

\subsection{Proof of Proposition~\ref{prop:ovd_limit_correction}}
\label{sec:proof_prop:ovd_limit_correction}

Recall that we set $M = \Id$ for overdamped limits. We consider first $f_2^{\gamma C, B,A,B, \gamma C}$, which satisfies~\eqref{eq:correction_second_order_schemes}. Let us first compute the right-hand side. Since
\[
\left[\left(A+\frac12 B\right)g\right] = \beta \left(p^T (\nabla^2 V) p - \frac12 |\nabla V|^2 \right),
\]
a simple computation shows that
\[
\widetilde{g} = \frac{1}{12} (A+B) \left[\left(A+\frac12 B\right)g\right] = \frac{\beta}{12} \Big[ (\nabla^3 V) : (p \otimes p \otimes p) - 3p^T (\nabla^2 V) \nabla V \Big].
\]
Note that the above function has average zero with respect to~$\kappa$. We then apply Theorem~\ref{lem:bounds_CL_gamma} to obtain
\[
\left\| f_2^{\gamma C, B,A,B, \gamma C} - \Lovd^{-1} \pi (A+B) C^{-1} \widetilde{g} \right\|_{H^1(\mu)} \leq \frac{K}{\gamma}.
\]
Since
\[
C \Big[ (\nabla^3 V) : (p \otimes p \otimes p) \Big] = -3 (\nabla^3 V) : (p \otimes p \otimes p) + \frac{6}{\beta} p^T \nabla \left(\Delta V\right),
\]
it is easily checked that
\[
\begin{aligned}
C^{-1} \widetilde{g} & = -\frac{\beta}{36} (\nabla^3 V) : (p \otimes p \otimes p) - \frac{1}{6} p^T\nabla (\Delta V) + \frac{\beta}{4} p^T (\nabla^2 V) \nabla V \\
& = -\frac{\beta}{36} A^3 \pi V + A\pi \left(- \frac16  (\Delta V) + \frac{\beta}{8} |\nabla V|^2\right).
\end{aligned}
\]
To compute $\pi (A+B) C^{-1} \widetilde{g}$, we rely on~\eqref{eq:rules_A+B_carre} and~\eqref{eq:rules_A4_BA3_etc} and obtain
\[
\begin{aligned}
\pi (A+B) C^{-1} \widetilde{g} & = -\frac{1}{12} \left(\frac1\beta \Delta^2 V - \nabla V \cdot \nabla (\Delta V) \right) + \Lovd \left(- \frac16  (\Delta V) + \frac{\beta}{8} |\nabla V|^2\right) \\
& = \Lovd\left(-\frac14 \Delta V + \frac{\beta}{8} |\nabla V|^2\right).
\end{aligned}
\]
This allows us to conclude that the limit of $f_2^{\gamma C, B,A,B, \gamma C}$ is the argument of the operator $\Lovd$ in the previous line, up to an additive constant chosen to ensure that $f_2^{\gamma C, B,A,B, \gamma C}$ has a vanishing average with respect to~$\mu$ (which turns out to be~$a_{\beta,V}/8$). We deduce the limit for $f_2^{A,B, \gamma C,B,A}$ with~\eqref{eq:correction_second_order_schemes} since $(A+B)g = p^T(\nabla^2 V)p-|\nabla V|^2$.

The expressions for the limits of $f_2^{\gamma C, A,B,A \gamma C}$ and $f_2^{B,A, \gamma C,A,B}$ are obtained in a similar fashion.

\subsection{Linear response theory}
\label{sec:proof_LRT}

\subsubsection{Definition of the mobility in~\eqref{eq:def_nu_LRT}}

We briefly sketch the discussion in~\cite[Section~3.1]{HDR} (see in particular Theorem~3.1 in this reference). Hypoellipticity arguments show that the measure $\mu_{\gamma,\eta}$ has a smooth density with respect to the Lebesgue measure. It moreover formally satisfies the Fokker-Planck equation
\begin{equation}
\label{eq:FP_xi}
\left(\Lgam+\eta \wcL\right)^* h_{\gamma,\eta} = 0, \qquad  \mu_{\gamma,\eta}(\d{q} \, \d{p}) = h_{\gamma,\eta}(q,p) \mu(\d{q}\,\d{p}), \qquad \int_\cE \d{\mu}_{\gamma,\eta} = 1.
\end{equation}
This equation can be given a rigorous meaning when $\eta$ is sufficiently small. We rely on the following result (proved at the end of this section), which is itself based on the fact that $\left(\Lgam^*\right)^{-1}$ can be extended to a bounded operator on $\cH^0$ (see Theorem~\ref{thm:Ham_limit_Lgam} and the comment after it).

\begin{lemma}
\label{lem:relatively_bounded_perturbation}
The operator $(\Lgam^*)^{-1}\wcL^*$, considered as an operator on the Hilbert space $\mathcal{H}^0 = L^2(\mu) \cap \{ \mathbf{1} \}^\perp$ introduced in~\eqref{eq:def_H0}, is bounded.
\end{lemma}

Denoting by $r$ the spectral radius of $(\Lgam^*)^{-1}\wcL^* \in \mathcal{B}(\mathcal{H}^0)$, it is easily checked that $\left(\Lgam+\eta \wcL\right)^*$ is invertible for $|\eta| < r^{-1}$ with
\[
\left[ \left(\Lgam+\eta \wcL\right)^* \right]^{-1} = \left(\sum_{n=0}^{+\infty} (-\eta)^n \left[\left(\Lgam^*\right)^{-1}\wcL^*\right]^n\right)\left(\Lgam^*\right)^{-1}.
\]
Therefore, a straightforward computation shows that
\begin{equation}
\label{eq:f_gamma_eta}
h_{\gamma,\eta}(q,p) = 1 + \sum_{n=1}^{+\infty} (-\eta)^n \left[\left(\Lgam^*\right)^{-1}\wcL^*\right]^n \mathbf{1}
\end{equation}
is an admissible solution of~\eqref{eq:FP_xi}, and it is in fact the only one in view of the uniqueness of the invariant probability measure (since $h_{\gamma,\eta}$ can be shown to be nonnegative). Note that the normalization of the measure $h_{\gamma,\eta} \d{\mu}$ does not depend on $\eta$. Finally, 
\[
\int_\cE F^T M^{-1}p \, \mu_{\gamma,\eta}(\d{q} \, \d{p}) = -\eta \int_\cE F^T M^{-1}p \left[\left(\Lgam^*\right)^{-1}\wcL^*\mathbf{1} \right] \mu(\d{q} \, \d{p}) + \eta^2 r_{\eta,\gamma},
\]
with $r_{\eta,\gamma}$ uniformly bounded as $\eta \to 0$. This gives~\eqref{eq:def_nu_LRT}.

\begin{proof}[Proof of Lemma~\ref{lem:relatively_bounded_perturbation}]
Note first that the image of $\wcL^*$ is contained in $\cH^0$ since, for any $u \in \mathcal{S}$, 
\[
\int_\cE \wcL^*u \, \d{\mu} = \int_\cE u \left(\wcL\mathbf{1}\right) \d{\mu} = 0.
\]
It is therefore possible to give a meaning to the operator $(\Lgam^*)^{-1}\wcL^*$ as an operator on~$\widetilde{\mathcal{S}}$. We then check that the perturbation $\wcL$ is $\Lgam$-bounded (with relative bound~0, in fact): for $u \in \widetilde{\mathcal{S}}$,
\[
\left\| \wcL u \right\|^2_{L^2(\mu)} \leq |F|^2 \| \nabla_p u\|^2_{L^2(\mu)} = -\beta |F|^2 \langle u, \Lgam u\rangle_{L^2(\mu)} \leq \beta |F|^2 \| u \|_{L^2(\mu)} \left\| \Lgam u \right\|_{L^2(\mu)},
\]
so that, for $u \in \cH^0$ (recall that $\Lgam^{-1}u$ is well defined in this case),
\[
\left\| \wcL \Lgam^{-1} u \right\|^2_{L^2(\mu)} \leq   \beta |F|^2 \| u \|_{L^2(\mu)}\left\| \Lgam^{-1} u \right\|_{L^2(\mu)} \leq \beta |F|^2 \left\| \Lgam^{-1} \right\|_{\mathcal{B}(\cH^0)} \| u \|^2_{L^2(\mu)}.
\]
This proves that $\wcL \Lgam^{-1}$ is bounded, hence its adjoint is bounded as well. 
\end{proof}

\subsubsection{Proof of Lemma~\ref{lem:ovd_mobility}}
\label{sec:proof_lem:ovd_mobility}

Recall that we set mass matrices to identity when considering overdamped limits. Since
\[
\Lgam \left(F^T p\right) = - \gamma F^T p - F^T \nabla V, 
\]
it follows (using first~\eqref{eq:f_gamma_eta} to compute the linear response and then~\eqref{eq:divergent_behavior_Lgamma} to obtain the asymptotic behavior of $\Lgam^{-1}(F^T\nabla V)$ as $\gamma \to +\infty$)
\[
\begin{aligned}
\gamma \nu_{F,\gamma} & = \lim_{\eta \to 0} \frac{\gamma}{\eta} \int_\cE F^T p \, \mu_{\gamma,\eta}(\d{q}\,\d{p}) = \lim_{\eta \to 0} \frac1\eta \int_\cE \left[-F^T \nabla V(q)-\Lgam \left(F^T p\right)\right] \mu_{\gamma,\eta}(\d{q}\,\d{p}) \\
& = \beta \int_\cE F^T p \, \Lgam^{-1} \left[F^T \nabla V(q)+\Lgam \left(F^T p\right)\right] \mu(\d{q}\,\d{p}) \\
& = |F|^2 + \beta \int_\cE \left(F^T p\right) \left[p^T \nabla_q \Lovd^{-1}\left(F^T \nabla V\right)\right]  \mu(\d{q}\,\d{p}) + \frac1\gamma r_\gamma \\
& = |F|^2 + \int_\mathcal{M} \left(F^T \nabla_q^* \mathbf{1} \right) \Lovd^{-1}\left(F^T \nabla V\right) \, \overline{\mu}(\d{q}) + \frac1\gamma r_\gamma \\
& = |F|^2 + \beta \int_\mathcal{M} \left(F^T \nabla V \right) \Lovd^{-1}\left(F^T \nabla V\right) \, \overline{\mu}(\d{q}) + \frac1\gamma r_\gamma \\
& = |F|^2 + \overline{\nu}_F + \frac1\gamma r_\gamma,
\end{aligned}
\]
where $r_\gamma$ is uniformly bounded for $\gamma \geq 1$. This gives the desired result.

\begin{remark}
\label{rmk:ovd_Einstein}
The article~\cite{HP08} in fact studies the limiting behavior of the autodiffusion coefficient, as computed from~\eqref{eq:def_nu_Einstein}:
\[
\beta \overline{D}_F = \int_\mathcal{M} \left| F + \nabla_q \Lovd^{-1} \left(F\cdot \nabla V\right) \right|^2 \d{\overline{\mu}}.
\]
Using $\Lovd = -\beta^{-1} \nabla_q^* \nabla_q$, a simple computation shows
\[
\begin{aligned}
\beta \overline{D}_F 
& = |F|^2 + 2 \int_\mathcal{M}F^T \nabla_q \Lovd^{-1} \left(F\cdot \nabla V\right) \, \d{\overline{\mu}}
+ \int_\mathcal{M} \left|\nabla_q \Lovd^{-1} \left(F\cdot \nabla V\right)\right|^2 \d{\overline{\mu}} \\
& = |F|^2 + 2 \int_\mathcal{M} \left(F^T \nabla_q^*\mathbf{1} \right) \Lovd^{-1} \left(F\cdot \nabla V\right) \, \d{\overline{\mu}} + \int_\mathcal{M} \nabla_q^* \nabla_q \Lovd^{-1} \left(F\cdot \nabla V\right)\, \Lovd^{-1} \left(F\cdot \nabla V\right)\, \d{\overline{\mu}} \\
& = |F|^2 + \beta \int_\mathcal{M} \left(F^T \nabla V \right) \Lovd^{-1} \left(F\cdot \nabla V\right) \, \d{\overline{\mu}},
\end{aligned}
\]
so that $\beta \overline{D}_F = |F|^2 + \overline{\nu}_F$.
\end{remark}

\subsection{Proof of Theorem~\ref{thm:error_noneq}}
\label{sec:proof_noneq}

The proof again is along the lines of the proof written in Section~\ref{sec:proof_thm:error_first_order_schemes}, and we are therefore very brief, mentioning only the most important modifications.

\paragraph{Case $\alpha=1$.}
Let us first consider the first order scheme $P_{\dt}^{\gamma C,B+ \eta \wcL,A}$. Using the notation introduced in Section~\ref{sec:expansion_evolution}, and recalling the definition $B_\eta = B + \eta \wcL$, we write
\begin{equation}
  \label{eq:expansion_P_dt_eta}
  P_{\dt}^{\gamma C,B + \eta\wcL,A} = \Id + \dt \left(A+B_\eta+\gamma C\right)+ \frac{\dt^2}{2} \cT\left[\Big( A+B_\eta+\gamma C \Big)^2\right] + \frac{\dt^3}{2} R_{\eta,\dt}, 
\end{equation}
with
\[
R_{\eta,\dt} = \int_0^1 (1-\theta)^2 \, \cT\left[(A+B_\eta+\gamma C)P_{\theta \dt}^{\gamma C,B + \eta\wcL,A}\right]^3 \d{\theta}.
\]
All the operators appearing in the expressions above are defined on the core~$\mathcal{S}$, and have values in~$\mathcal{S}$. Since 
\[
\rme^{\theta \dt B_\eta} - \rme^{\theta \dt B} = \eta \int_0^1 \rme^{\theta s B_\eta} \, \wcL \, \rme^{\theta (1-s) B} \, ds,
\]
it is easy to see that the operator $R_{\eta,\dt}$ can be rewritten as the sum of two contributions: $R_{\eta,\dt} = R_{0,\dt} + \eta \widetilde{R}_{\eta,\dt}$, where, for $\psi \in \mathcal{S}$, the smooth function $\widetilde{R}_{\eta,\dt} \psi$ can be uniformly controlled in $\eta$ for $|\eta| \leq 1$. Finally, the evolution operator can be rewritten as
\begin{equation}
\label{eq:Pdt_expansion_finite_gamma_eta}
P_{\dt}^{\gamma C,B + \eta\wcL,A} = \Id + \dt \left(\Lgam + \eta \wcL\right) + \frac{\dt^2}{2} \left(\Lgam^2 + S_1 + \eta D_1 \right) + \dt^2 \mathscr{R}_{\eta,\dt},
\end{equation}
where $S_1$ is defined in~\eqref{eq:P_dt_eq_order1} (which corresponds to the case $\eta = 0$), $D_1 = (2\gamma C+B)\wcL + \wcL (2A+B)$, and 
\[
\mathscr{R}_{\eta,\dt} = \frac{\dt}{2} R_{0,\dt} + \frac{\eta \dt}{2} \widetilde{R}_{\eta,\dt} + \frac{\eta^2}{2} \wcL^2.
\]
We then compute, for $\varphi \in \mathcal{S}$ and $f_{1,1,\gamma} \in \widetilde{\mathcal{S}}$ to be chosen later,
\[
\begin{aligned}
& \int_\cE\left[ \left(\frac{\Id-P_{\dt}^{\gamma C,B + \eta\wcL,A}}{\dt}\right)\varphi \right] \left(1 + \dt f_{1,0,\gamma} + \eta f_{0,1,\gamma} + \eta \dt f_{1,1,\gamma} \right) \, \d{\mu} \\
& = -\int_\cE \left[\left(\Lgam + \eta \wcL + \frac{\dt}{2} \left(\Lgam^2 + S_1 + \eta D_1 \right) + \dt \mathscr{R}_{\eta,\dt} \right)\varphi \right] \left(1 + \dt f_{1,0,\gamma} + \eta f_{0,1,\gamma} + \eta \dt f_{1,1,\gamma} \right) \, \d{\mu} \\
& = -\eta \int_\cE \left[ \wcL \varphi + (\Lgam \varphi) f_{0,1,\gamma} \right]\d\mu - \dt \int_\cE \left[ \frac12 S_1 \varphi + (\Lgam \varphi) f_{1,0,\gamma} \right]\d\mu \\
& \ \ \ - \eta\dt \int_\cE \left[ \left(\wcL \varphi\right) f_{1,0,\gamma} + \frac12\left(\Lgam^2 + S_1\right)\varphi \, f_{0,1,\gamma} + (\Lgam \varphi) f_{1,1,\gamma} + \frac12 D_1\varphi \right]\d\mu \\
& \ \ \ - \eta^2 \int_\cE \left(\wcL \varphi\right) (f_{0,1,\gamma} + \dt f_{1,1,\gamma})\, \d{\mu} - \frac{\dt^2}{2} \int_\cE \left[\left(\Lgam^2 + S_1 + \eta D_1 \right)\varphi \right] (f_{1,0,\gamma} + \eta f_{1,1,\gamma}) \, \d{\mu} \\
& \ \ \ - \dt \int_\cE \mathscr{R}_{\eta,\dt}\varphi \left(1 + \dt f_{1,0,\gamma} + \eta f_{0,1,\gamma} + \eta \dt f_{1,1,\gamma} \right) \, \d{\mu}.
\end{aligned}
\]
The first two terms in the last expression vanish by definition of $f_{0,1,\gamma}$ and $f_{1,0,\gamma}$, while the third one vanishes when the function $f_{1,1,\gamma}$ is defined by the Poisson equation
\begin{equation}
\label{eq:def_f_1_1_gamma}
\Lgam^* f_{1,1,\gamma} =  - \wcL^* f_{1,0,\gamma} - \frac12\left(\Lgam^2 + S_1\right)^* f_{0,1,\gamma} - \frac12 D_1^* \mathbf{1}.
\end{equation}
It is easy to check that the right-hand side of this equation has a vanishing average with respect to~$\mu$ (integrating with respect to~$\mu$ and letting the adjoints of the operators act on~$\mathbf{1}$). We then project~\eqref{eq:expansion_P_dt_eta} using $\Pi^\perp$ and introduce the approximate inverse, defined on~$\widetilde{\mathcal{S}}$ as
\[
\begin{aligned}
Q_{\eta,\dt} & = -\Lgam^{-1} + \eta \Lgam^{-1} \Pi^\perp\wcL\Pi^\perp \Lgam^{-1} + \frac{\dt}{2} \left[\Pi^\perp + \Lgam^{-1}\Pi^\perp\left(S_1 + \eta D_1\right)\Pi^\perp \Lgam^{-1} \right] \\
& \ \ - \frac{\eta\dt}{2} \Lgam^{-1} \Pi^\perp\wcL\Pi^\perp \Lgam^{-1}\left(\Lgam^2 + \Pi^\perp S_1\Pi^\perp + \eta \Pi^\perp D_1\Pi^\perp \right) \Lgam^{-1} \\
& \ \ - \frac{\eta\dt}{2}  \Lgam^{-1} \left(\Lgam^2 + \Pi^\perp S_1\Pi^\perp + \eta \Pi^\perp D_1\Pi^\perp \right) \Lgam^{-1} \Pi^\perp \wcL \Pi^\perp \Lgam^{-1},
\end{aligned}
\]
obtained by truncating the formal series expansion of the inverse operator by discarding terms associated with $\eta^2$ or $\dt^2$. The approximate inverse is such that
\[
\Pi^\perp \left(\frac{\Id-P_{\dt}^{\gamma C,B + \eta\wcL,A}}{\dt}\right) \Pi^\perp Q_{\eta,\dt} = \Pi^\perp + \eta^2 \mathcal{R}^1_{\eta,\dt} + \dt^2 \mathcal{R}^2_{\eta,\dt},
\]
with $\mathcal{R}^2_{\eta,\dt} = \mathcal{R}^2_{0,\dt} + \eta \widetilde{\mathcal{R}}^2_{\eta,\dt}$. We then replace $\Pi^\perp \varphi$ by $Q_{\eta,\dt} \psi$ and conclude as in Section~\ref{sec:proof_thm:error_first_order_schemes}.

\paragraph{Case $\alpha=2$.}
The result for the second order splitting is obtained by appropriate modifications of the proof written above for $p=1$, similar to the ones introduced in Section~\ref{sec:proof_thm:error_second_order_schemes}. We will therefore mention only the most important point, which is the following. Replacing $B$ by $B_\eta$ in the expansion~\eqref{eq:I_Pdt_eq_order2}, we see that
\[
\begin{aligned}
\frac{\Id-P_\dt^{\gamma C,B_\eta,A,B_\eta,\gamma C}}{\dt} & = -\Lgam - \eta \wcL - \frac\dt2 (\Lgam+\eta \wcL)^2 - \dt^2 \left( \frac16 (\Lgam+\eta\wcL)^3 + S_{2} + \eta \widetilde{S}_{2,\eta} \right) - \dt^3 R_{\eta,\dt} \\
& =  -\Lgam - \eta \wcL - \frac\dt2 \Lgam^2 - \frac{\eta\dt}{2} \left(\Lgam \wcL + \wcL \Lgam\right) - \frac{\eta^2\dt}{2} \wcL^2 - \dt^2 \left( \frac16 \Lgam^3 + S_{2} \right) \\
& \ \ \ - \eta \dt^2 \left(\frac16 \left(\Lgam^2 \wcL + \Lgam \wcL \Lgam + \wcL \Lgam^2 \right) + \widetilde{S}_{2,0}\right) + \mathscr{R}_{\eta,\dt},
\end{aligned}
\]
where $\mathscr{R}_{\eta,\dt}$ regroups operators of order $\dt^{3+\alpha} \eta^{\alpha'}$ or $\dt^{2+\alpha} \eta^{2+\alpha'}$ for $\alpha,\alpha' \geq 0$, the operator $S_2$ is defined in~\eqref{eq:def_S2} and $\widetilde{S}_{2,\eta}$ satisfies
\[
\begin{aligned}
12 \, \widetilde{S}_{2,\eta} & = \left[A,\left[A,\wcL\right]\right] - \frac12 \left[B,\left[\wcL,A\right]\right]- \frac12 \left[\wcL,\left[B,A\right]\right] + \gamma \left[\wcL,\left[A+B,C\right]\right]  + \gamma \left[A+B,\left[\wcL,C\right]\right]  \\
& \ \ \ - \frac{\gamma^2}{2} \left[C,\left[C,\wcL\right]\right] + \eta \left( \gamma \left[\wcL,\left[\wcL,C\right]\right] - \frac12 \left[\wcL,\left[\wcL,A\right]\right] \right).
\end{aligned}
\]
We next compute the dominant terms in
\[
\int_\cE \left[\left(\frac{\Id-P_\dt^{\gamma C,B_\eta,A,B_\eta,\gamma C}}{\dt}\right)\varphi \right] \left(1 + \dt^2 f_{2,0,\gamma} + \eta f_{0,1,\gamma} + \eta \dt^2 f_{2,1,\gamma} \right) \, \d{\mu}.
\]
We consider only contributions of the form $\eta^\alpha \dt^{\alpha'}$ with $\alpha = 0,1$ and $0 \leq \alpha' \leq 2$. The contributions in $\dt,\dt^2$ are the same as in the case $\eta = 0$ and therefore vanish. The contribution in $\eta$ vanishes in view of the choice of $f_{0,1,\gamma}$. For the same reason, the contribution in $\eta\dt$ vanishes as well:
\[
- \frac{\eta\dt}{2} \int_\cE \left[ \left(\Lgam \wcL + \wcL \Lgam\right)\varphi + \left(\Lgam^2\varphi\right) f_{0,1,\gamma} \right] \, \d{\mu} = - \frac{\eta\dt}{2} \int_\cE \left( \Lgam \varphi\right) \left(\wcL^*\mathbf{1} + \Lgam^* f_{0,1,\gamma}\right) \d{\mu} = 0.
\]
The contribution in $\eta\dt^2$ is proportional to
\[
\int_\cE \left[ \left(\frac{\Lgam^2 \wcL + \Lgam \wcL \Lgam + \wcL \Lgam^2}{6} + \widetilde{S}_{2,0}\right)\varphi + \left(\wcL \varphi\right) f_{2,0,\gamma} + \left[\left( \frac{\Lgam^3}{6} + S_{2} \right)\varphi\right] f_{0,1,\gamma} + \left(\Lgam\varphi\right) f_{2,1,\gamma} \right] \d{\mu}.
\]
The requirement that this expression vanishes for all functions $\varphi \in \mathcal{S}$ characterizes the function $f_{2,1,\gamma}$ (the discussion on the solvability of this equation following the same lines as the discussion on the solvability of~\eqref{eq:def_f_1_1_gamma}). The proof is then concluded as in the case $p=1$.

\subsection{Proof of Theorem~\ref{thm:error_estimate_noneq_ovd}}
\label{sec:proof_thm:error_estimate_noneq_ovd}

The proof of this result is obtained by modifying the proof of Theorem~\ref{thm:ovd_limit} presented in Section~\ref{sec:proof_thm:ovd_limit} by taking into account the nonequilibrium perturbation, as done in the proof of Theorem~\ref{thm:error_noneq} presented in Section~\ref{sec:proof_noneq}. We will therefore be very brief and only mention the most important modifications.

We write the proof for the scheme associated with the evolution operator $P_\dt^{\gamma C,A,B_\eta,A,\gamma C}$ for instance (since this is the case explicitly treated in Section~\ref{sec:proof_thm:ovd_limit} for $\eta = 0$). First, arguing as in Section~\ref{sec:proof_thm:ovd_limit}, we see that it is possible to replace $P_\dt^{\gamma C,A,B_\eta,A,\gamma C}$ by 
\[
\pi P_{{\rm ham},\dt,\eta} \pi = \pi \rme^{\dt A/2} \rme^{\dt B_\eta} \rme^{\dt A/2} \pi
\]
up to error terms in the invariant measure which are exponentially small in $\gamma \dt$. Note that $B_\eta = (F-\nabla V)\cdot \nabla_p$, so that the rules~\eqref{eq:rules_ovd_1}-\eqref{eq:rules_ovd_2} are still valid. Therefore, introducing again $h = \dt^2/2$,
\[
\begin{aligned}
& \pi P_{{\rm ham},\dt,\eta} \pi \\
& = \pi + \frac{\dt^2}{2} \pi (A+B_\eta)^2 \pi + \frac{\dt^4}{24}\pi \left(A^4 + \frac32 A^2 B_\eta A + \frac32 A B_\eta A^2 + \frac32 B_\eta^2 A^2 + \frac12 B_\eta A^3\right)\pi + \dt^6 R_{\dt,\eta} \\
& = \pi + h \pi \left(\Lovd + \eta\left[\wcL(A+B) + (A+B)\wcL\right] + \eta^2 \wcL^2 \right) \pi + \frac{h^2}{2} \left(\Lovd^2 + D + \eta \widetilde{D}_1 + \eta^2 \widetilde{D}_2\right) \pi \\
& \ \ \ + \dt^6 R_{\dt,\eta},
\end{aligned}
\]
where $D$ is defined in~\eqref{eq:def_D}, and the expressions of the operators $\widetilde{D}_i$ ($i=1,2$) are obtained by expanding the various terms $A^aB_\eta^b A^c$ in powers of $\eta$. Keeping only the dominant terms, we arrive at
\[
\pi P_{{\rm ham},\dt,\eta} \pi = \pi + h \Lovd\pi + \frac{h^2}{2}\left(\Lovd^2 + D\right) + \eta h \pi \left[\wcL(A+B) + (A+B)\wcL\right]\pi + \frac{\eta h^2}{2} \widetilde{D}_1 + \mathscr{R}_{\dt,\eta}.
\]
Since 
\[
\pi \left(\wcL(A+B) + (A+B)\wcL\right) \pi  = \pi \wcL A \pi = \wcL_{\rm ovd},
\]
we conclude
\[
\pi P_{{\rm ham},\dt,\eta} \pi = \pi + h \left( \Lovd + \eta \wcL_{\rm ovd}\right) \pi + \frac{h^2}{2}\left(\Lovd^2 + D + \eta  \widetilde{D}_1 \right) + \mathscr{R}_{\dt,\eta}.
\]
This relation is the analogue of~\eqref{eq:Pdt_expansion_finite_gamma_eta} in the overdamped limit, and the remainder of the proof is carried on following the strategy presented in Section~\ref{sec:proof_thm:ovd_limit}. 

\section*{Acknowledgements}

The authors thank Francis Nier for several fruitful discussions on the properties of Langevin-type generators. Ben Leimkuhler and Charles Matthews acknowledge the support of the Engineering and Physical Sciences Research Council (UK) and grant EP/G036136/1. G. Stoltz was partially supported by the project DYMHOM (De la dynamique mol\'eculaire, via l'homog\'en\'eisation, aux mod\`eles macroscopiques de poro\'elasticit\'e et \'electrocin\'etique) from the program NEEDS (Projet f\'ed\'erateur Milieux Poreux MIPOR). G. Stoltz also benefited from the scientific environment of the Laboratoire International Associ\'e between the Centre National de la Recherche Scientifique and the University of Illinois at Urbana-Champaign.

\bibliographystyle{IMANUM-BIB}
\bibliography{ref}

\end{document}